\documentclass[11pt]{article}
\usepackage{times}
\usepackage{caption}
\usepackage[margin=0.8in]{geometry}

\usepackage{setspace}
\usepackage{cite}
\usepackage{graphicx}
\usepackage{subfigure}
\usepackage{amsmath}
\usepackage{amssymb}
\usepackage{amsthm}
\usepackage{algorithm,algorithmic}
\usepackage{booktabs}
\usepackage{multirow}
\usepackage{multicol}
\usepackage{color}
\usepackage{paralist}
\usepackage{xiaoli}
\usepackage{wrapfig}

\newcommand{\GF}{\mathbb{F}_2}

\begin{document}

\title{\huge {\bf SPRIGHT}: A Fast and Robust Framework for Sparse Walsh-Hadamard Transform}

\author{
Xiao Li, Joseph K. Bradley, Sameer Pawar and Kannan Ramchandran\thanks{This work was supported by grants NSF CCF EAGER 1439725, and NSF CCF 1116404 and MURI CHASE Grant No. 556016.}\\
Department of Electrical Engineering and Computer Science (EECS)\\
University of California, Berkeley\\
\{xiaoli,~josephkb,~spawar,~kannanr\}@eecs.berkeley.edu
}
\date{}

\maketitle

\begin{abstract}
We consider the problem of stably computing the Walsh-Hadamard Transform (WHT) of some $N$-length input vector in the presence of noise, where the $N$-point {\it Walsh spectrum} is $K$-sparse  with $K = {O}(N^{\delta})$ scaling sub-linearly in the input dimension $N$ for some $0<\delta<1$. Note that $K$ is linear in $N$ (i.e. $\delta = 1$), then similar to the standard Fast Fourier Transform (FFT) algorithm, the classic Fast WHT (FWHT) algorithm offers an ${O}(N)$ sample cost and ${O}(N\log N)$ computational cost, which are order optimal. Over the past decade, there has been a resurgence in research related to the computation of Discrete Fourier Transform (DFT) for some length-$N$ input signal that has a $K$-sparse $N$-point {\it Fourier spectrum}. In particular, through a sparse-graph code design, our earlier work on the {\it Fast Fourier Aliasing-based Sparse Transform} (FFAST) algorithm \cite{pawar2013computing} computes the $K$-sparse DFT in time ${O}(K\log K)$ by taking ${O}(K)$ noiseless samples. Inspired by the coding-theoretic design framework in \cite{pawar2013computing}, Scheibler et al. in \cite{scheibler2013fast} proposed the {\it Sparse Fast Hadamard Transform (SparseFHT)} algorithm that elegantly computes the $K$-sparse WHT in the {\it absence} of noise using ${O}(K\log N)$ samples in time ${O}(K\log^2 N)$. However, the SparseFHT algorithm explicitly exploits the noiseless nature of the problem, and is not equipped to deal with scenarios where the observations are corrupted by noise, as is true in general. Therefore, a question of critical interest is whether this coding-theoretic framework can be made robust to noise. Further, if the answer is yes, what is the extra price that needs to be paid for being robust to noise?

In this paper, we show, quite interestingly, that there is {\it no extra price} that needs to be paid for being robust to noise other than a constant factor. In other words, we can maintain the same scaling for the sample complexity ${O}(K\log N)$ and the computational complexity ${O}(K\log^2 N)$ as those of the noiseless case, using our proposed {\it {\bf SP}arse {\bf R}obust {\bf I}terative {\bf G}raph-based {\bf H}adamard {\bf T}ransform} ({\bf SPRIGHT}) algorithm. Similar to the FFAST algorithm \cite{pawar2013computing} and the SparseFHT algorithm \cite{scheibler2013fast}, the proposed SPRIGHT framework succeeds with high probability with respect to a random ensemble of signals with sparse Walsh spectra, where the support of the non-zero WHT coefficients is uniformly random. Experiments further corroborate the robustness of the SPRIGHT framework as well as its scaling performance.

\end{abstract}

\section{Introduction}
Ever since the introduction of orthonormal Walsh functions, the Walsh-Hadamard Transform (WHT) has gained traction for signal analysis in place of the Discrete Fourier Transform (DFT) because of its simplicity in computations and applicability in the design of practical systems like digital circuits. Starting off as the ``poor man's fast Fourier Transform'', the WHT has been further deployed over the past few decades in image and video compression \cite{pratt1969hadamard}, spreading code design in multiuser systems such as CDMA and GPS \cite{wgispreading}, and compressive sensing \cite{haghighatshoar2013polarization}. More recently, sparsity in the {\it Walsh spectrum} is found in many real-world applications involving the processing of large datasets, such as learning (pseudo) Boolean functions, decision trees and disjunctive normative form (DNF) formulas, etc.  Therefore, it is of practical and theoretical interest to develop fast algorithms for computing the WHT of signals with sparse or approximately sparse Walsh spectra. Traditionally, the WHT can be computed using $N$ samples and ${O}(N\log N)$ operations via a recursive algorithm \cite{lee1986fast,johnson2000in} analogous to the Fast Fourier Transform (FFT). However, these costs can be significantly reduced if the signal has a sparse Walsh spectrum \cite{horadam2007hadamard,hedayat1978hadamard}.

\subsection{Motivation and Contributions}
There has been a recent resurgence in research on computing the Discrete Fourier Transform (DFT) of signals that have sparse {\it Fourier spectra} \cite{hassanieh2012simple, hassanieh2012nearly,ghazi2013sample, iwen2007empirical,iwen2010combinatorial,pawar2013computing}. Since the WHT is a special case of a multidimensional DFT over the binary field, recent advances in computing $K$-sparse $N$-point DFTs have provided insights in designing algorithms for computing sparse WHTs. In particular, major progress has been made in breaking the ``$N$-barrier'' for computing an $N$-point sparse DFTs, which means that the sample complexity and computational complexity do not depend on the signal dimension $N$. In particular, using a sparse-graph code design, the {\it Fast Fourier Aliasing-based Sparse Transform (FFAST)} algorithm \cite{pawar2013computing} uses ${O}(K)$ samples and ${O}(K\log K)$ operations for any sub-linear sparsity $K={O}(N^{\delta})$ with $0<\delta<1$ assuming a uniform support distribution. Under a similar uniform support distribution for the WHT coefficients, the Sparse Fast Hadamard Transform (SparseFHT) algorithm developed in\cite{scheibler2013fast} elegantly computes a $K$-sparse $N$-point WHT with $K={O}(N^{\delta})$ using ${O}(K\log (N/K))$ samples and ${O}(K\log K \log N/K)$ operations by following the sparse-graph code design in \cite{pawar2013computing} for DFTs. When $K$ is scales sub-linearly in $N$ as $K={O}(N^\delta)$ for some constant $0<\delta<1$, these results are hereby interpreted as achieving a sample complexity ${O}(K\log N)$ and a computational complexity ${O}(K\log^2 N)$. A limitation of the SparseFHT algorithm is that it is designed to explicitly exploit the {\it noiseless} nature of the underlying signals and it is not clear how to generalize it to noisy settings. A key question of theoretical and practical interest in this paper is: what price must be paid to be robust to noise? Interestingly, in this paper we show that {\it there are no extra costs in sample complexity and computational complexity for being robust to noise, other than a constant factor determined by the signal-to-noise ratio (SNR).} 

Inspired by the algorithm design from the FFAST algorithm in \cite{pawar2013computing} and the noisy FFAST analysis in \cite{pawar2013thesis}, we consider the problem of computing a $K$-sparse $N$-point WHT from the input vector {\it in the presence of noise}, when the sparsity $K={O}(N^{\delta})$ is sub-linear in the signal dimension $N$ for some $0<\delta<1$ assuming a uniform support distribution. We develop a {\it SParse Robust Iterative Graph-based Transform (SPRIGHT)} framework to stably compute the $K$-sparse $N$-length WHT at any constant SNRs with high probability. In particular, our framework achieves sub-linear run-time ${O}(K\log^2 N)$ using ${O}(K\log N)$ noisy samples, which maintains the same sample and computational scaling as the noiseless case. This result also contrasts with the work on computing the sparse DFT in the presence of noise  \cite{pawar2013thesis}, where the robustness to noise incurs an extra factor of ${O}(\log N)$ in terms of the sample complexity from ${O}(K)$ to ${O}(K \log N)$ (the same extra factor is manifested in the run-time as well). This can be intuitively explained by the fact that the complex-valued $N$-point Fourier transform kernel has a ``$1/N$ precision'' while the binary-valued WHT kernel has a ``bit precision''.

\subsection{Notation and Organization}
Throughout this paper, the set of integers $\{0,1, \cdots, N-1\}$ for some integer $N$ is denoted by $[N]$. Lowercase letters, such as $x$, are used for the time domain expressions and uppercase letters, such as ${X}$, are used for the transform domain signal. Any boldface lowercase letter such as $\mathbf{x}\in\mathbb{R}^N$ represents a column vector containing the corresponding $N$ samples. The operator $\mathrm{supp}(\mathbf{x})$ takes the support set of the vector $\mathbf{x}$ and $|\cdot|$ takes the cardinality of a certain set. The notation $\mathbb{F}_2$ refers to the finite field consisting of $\{0, 1\}$, with defined operations such as summation and multiplication modulo 2. Furthermore, we let $\GF^n$ be the $n$-dimensional column vector with each element taking values from $\mathbb{F}_2$. For any vector $\mathbf{i}\in\GF^n$, denote by $\mathbf{i}=[i[1],\cdots,i[n]]^T\in\GF^n$ the index vector containing the binary representation of some integer $i$, with $i[1]$ and $i[n]$ being the least significant bit (LSB) and the most significant bit (MSB), respectively. The inner product of two binary indices $\mathbf{i}\in\GF^n$ and $\mathbf{j}\in\GF^n$ is defined by $\ip{\mathbf{i}}{\mathbf{j}}=\sum_{t=0}^{n-1} i[t]j[t]$ with arithmetic over $\mathbb{F}_2$, and the inner product between two vectors $\mathbf{x},\mathbf{y}\in\mathbb{R}^N$ is defined as $\ip{\mathbf{x}}{\mathbf{y}}=\sum_{t=1}^{N}x[t]u[t]$ with arithmetic over $\mathbb{R}$. The sign function here is defined as 
\begin{align}
	\sgn{x} =
	\begin{cases}
	1, & x<0\\
	0, & x>0
	\end{cases}
\end{align}
such that $x=|x|(-1)^{\sgn{x}}$.

This paper is organized as follows. In Section \ref{sec:main_result}, we present our input (signal) model and our goal, followed by a summary of our main results. To motivate our design, we explain in Section \ref{sec:simple_example} the main idea of our SPRIGHT framework through a simple example. Then, we generalize the simple example and present the framework in Section \ref{fig:schematic}, followed by detailed discussions in Section \ref{sec:robust_bin_detection} about the noisy scenarios in our framework. Last but not least, in Section \ref{sec:applications} we briefly mention some machine learning applications that can be potentially cast as a sparse WHT computation problem, followed by numerical experiments in Section \ref{sec:simulations}.

%%%%%%%%%%%%%%%%%%%%%%%%%%%%%%%%%%%%%%%%%%%%%%%%%%%%%%%%%%%%

\section{Problem Setup and Main Results}\label{sec:main_result}
Given a signal $\mathbf{x}\in\mathbb{R}^N$ containing $N=2^n$ samples $x[\mathbf{m}]$ indexed by $\mathbf{m}\in\GF^n$ (i.e. the $n$-bit binary representation of $m\in[N]$), its WHT coefficient is computed as
\begin{align}
	X[\mathbf{k}] = \frac{1}{\sqrt{N}} \sum_{\mathbf{m}\in\GF^n} (-1)^{\ip{\mathbf{k}}{\mathbf{m}}}x[\mathbf{m}],
\end{align}
where $\mathbf{k}=[k[1],\cdots,k[n]]^T\in\GF^n$ denotes the $n$-tuple index in the transform domain. Likewise, each sample $x[\mathbf{m}]$ has a WHT expansion as
\begin{align}
	x[\mathbf{m}] = \frac{1}{\sqrt{N}} \sum_{\mathbf{k}\in\GF^n} (-1)^{\ip{\mathbf{m}}{\mathbf{k}}}X[\mathbf{k}].
\end{align}

\subsection{Problem Setup}
In this work, we consider the noisy scenario where the samples $x[\mathbf{m}]$ are corrupted by additive noise $w[\mathbf{m}] \sim \mathcal{N}(0,\sigma^2)$, which is independent and normally distributed for all $\mathbf{m}\in\GF^n$. Thus,  we have access to only the noise-corrupted samples:
\begin{align}\label{time_domain_samples}
	u[\mathbf{m}] 	
	= \frac{1}{\sqrt{N}} \sum_{\mathbf{k}\in\GF^n} (-1)^{\ip{\mathbf{m}}{\mathbf{k}}}X[\mathbf{k}] + w[\mathbf{m}],\quad \mathbf{m}\in\GF^n.
\end{align}

\begin{ass}\label{random_support_assumption}
Let $\mathbf{{X}}\in\mathbb{R}^N$ be the WHT coefficient vector with support $\mathcal{K}\defn\supp{\mathbf{{X}}}$. Throughout this paper, we make the following assumptions:
\begin{itemize}
	\item[\bf A1] Each element in the support set $\mathcal{K}$ is chosen independently and uniformly at random from $[N]$.
	\item[\bf A2] The sparsity $K=\left|\supp{\mathbf{{X}}}\right|={O}(N^{\delta})$ is sub-linear in the dimension $N$ for some $0<\delta<1$.
	\item[\bf A3] Each coefficient $X[\mathbf{k}]$ for $\mathbf{k}\in\mathcal{K}$ is chosen from a finite set $\mathcal{X}\defn\{\pm \rho\}$ uniformly at random. 
	\item[\bf A4] The signal-to-noise ratio (SNR) is defined as
	\begin{align}\label{def_SNR}
		\mathsf{SNR}=\frac{\|\mathbf{x}\|^2/N}{\sigma^2}=\frac{\rho^2}{\sigma^2 N/K}
	\end{align}
	and is assumed to be an arbitrary constant value (i.e., $\rho$ scales with $\sqrt{N/K}$).
\end{itemize}
\end{ass}
\begin{rmk}
While the uniform distribution assumption {\bf A1} on the support $\mathcal{K}$ is essential to the analysis of our algorithm (see also \cite{pawar2013computing} and \cite{scheibler2013fast}), it can be generalized to accommodate non-uniform distributions that are of practical interest in real world applications. If we fail to insist on the sub-linear sparsity regime imposed in {\bf A2}, our results reduce to ${O}(N)$ samples in time ${O}(N\log N)$, which is well understood in classic WHT computations. Further, the binary constellation assumption {\bf A3} is imposed to simplify our analysis and can be readily extended to any arbitrarily large but finite constellation, which subsumes all practical digital signals that have been quantized with finite precision (essentially any signal processed by a digital computer). Last but not least, the constant SNR assumption {\bf A4} covers all regimes of interest.
\end{rmk}

The goal of this paper is to develop a robust and efficient algorithm that reliably recovers {\it exactly} the entire support $\mathcal{K}$ of the sparse WHT of a signal as well as the associated non-zero coefficients $X[\mathbf{k}]$ for $\mathbf{k}\in\mathcal{K}$ in the presence of noise. The questions of interest are
\begin{enumerate}
	\item How many noisy samples are needed to reliably recover the support of the sparse WHT?
	\item Can we reduce the computational complexity of the sparse WHT over that of the conventional WHT algorithm, even in the presence of noise?
\end{enumerate}

In the following, we first provide a summary of our main technical results, followed by a brief mention of previous work on computing sparse transforms.

\subsection{Main Result}
Our design is characterized by the triplet $(M, T, \Pf)$, where $M$ is the {\it sample complexity}\footnote{Note that the sample complexity is the number of raw samples needed as input for computations, as opposed to the measurement complexity in compressed sensing, where each measurement may potentially require all the samples from the input vector.}, $T$ is the {\it computational complexity} in terms of arithmetic operations, and $\Pf$ is the {\it probability of failure} in recovering the exact support of the sparse WHT, given by
\begin{align}\label{def_Pe}
	\Pf \defn \mathbb{E}\left[1\{{\supp{\widehat{\mathbf{{X}}}}\neq\supp{\mathbf{{X}}}}\}\right],
\end{align}
where $1\{\cdot\}$ is the indicator function and $\supp{\cdot}$ represents the support of some vector and the expectation is obtained with respect to the randomization of our algorithm, the noise distribution as well as the random signal ensemble in Assumption \ref{random_support_assumption}.

\begin{thm}
\label{thm_main.result_sublinear}
Let Assumption \ref{random_support_assumption} hold for the signal of interest $\mathbf{x}$ and its WHT vector $\mathbf{{X}}$. Then for any sparsity regime $K={O}(N^\delta)$ with $0<\delta<1$, the {\bf SPRIGHT} framework computes the $K$-sparse $N$-point WHT $\mathbf{{X}}$ with a vanishing failure probability $\Pf \rightarrow 0$ asymptotically in $K$ and $N$ using the following two algorithm options: 
\begin{itemize}
	\item the {\bf Sample Optimal (SO) SPRIGHT} algorithm with a sample complexity of $M={O}(K\log N)$ and a computational complexity of $T={O}(K\log^2N)$;
	\item the {\bf Near Sample Optimal (NSO) SPRIGHT} algorithm with a sample complexity of $M={O}(K\log^2 N)$ and a computational complexity of $T={O}(K\log^3N)$.
\end{itemize}
\end{thm}
\begin{proof}
	See Appendix \ref{main.results.fast}.
\end{proof}

\begin{rmk}
Since we assume an arbitrarily large but finite constellation $\mathcal{X}$ for each non-zero coefficient, we show that the coefficients can in fact be recovered perfectly, even from the noisy measurements with high probability. The recovery algorithm is equally applicable to support recovery for signals with arbitrary coefficients over the real field, but the analysis becomes overly cumbersome without offering more insights to our design. Hence we do not pursue it in this paper. 
\end{rmk}

\begin{rmk}
Note that although the result in Theorem \ref{thm_main.result_sublinear} is obtained with a randomized algorithm, our SPRIGHT framework also admits the option of using a deterministic algorithm by spending an extra factor of ${O}(\log N)$ in both sample complexity and computational complexity. 
\end{rmk}

\subsection{Related Work}

Due to the similarities between the DFT and the WHT, we give a brief account of previous work on reducing the sample and computational complexity of obtaining a $K$-sparse $N$-point DFT.  The most related research thread in the literature is the computation of sparse DFT using theoretical computer science techniques such as sketching and hashing (see \cite{gilbert2002near,gilbert2005improved,iwen2010combinatorial,mansour1995randomized,gilbert2008tutorial}). Most of these algorithms aim at minimizing the approximation error of the DFT coefficients using an $\ell_2$-norm metric instead of exact support recovery (i.e., $\ell_0$-norm). 

Among these works, the most recent progress in this direction is the sFFT (Sparse FFT) algorithm developed in the series of papers \cite{hassanieh2012simple,hassanieh2012nearly,ghazi2013sample}. Most of these algorithms are based on first isolating (i.e., hashing) the non-zero DFT coefficients into different bins, using specific filters or windows that have `good' (concentrated) support in both time and frequency. The non-zero DFT coefficients are then recovered iteratively, one at a time. The filters or windows used for the binning operation are typically of length ${O}(K \log N)$. As a result, the sample complexity is typically ${O}(K \log N)$ or more, with potentially large big-Oh constants as demonstrated in \cite{iwen2007empirical}. Then, \cite{ghazi2013sample} further improved the $2$-D DFT algorithm for the special case of $K=\sqrt{N}$, which reduces the sample complexity to ${O}(K)$ and the computational complexity to ${O}(K\log K)$, albeit with a constant failure probability that does not vanish as the signal dimension $N$ grows. On this front, the deterministic algorithm in \cite{iwen2010combinatorial} is shown to guarantee zero errors but with complexities of ${O}(\mathrm{poly}(K,\log N))$. More recently, \cite{cheraghchi2015nearly} develops a deterministic algorithm for computing a sparse WHT in time ${O}(K^{1+\epsilon} \log^{{O}(1)} N)$ with an arbitrary constant $\epsilon>0$. 

One of the interesting recent advances in computing sparse DFTs is in the breaking of the ``$N$-barrier'', which means that the complexities no longer depend on the input dimension $N$. In particular, the FFAST algorithm \cite{pawar2013computing} uses only ${O}(K)$ samples and ${O}(K\log K)$ operations for any sparsity regime $K={O}(N^{\delta})$ and ${\delta}\in(0,1)$. Similar to the spirit of compressed sensing in linearly combining sparse components (i.e., DFT coefficients), the FFAST algorithm judiciously chooses subsampling patterns to create spectral aliasing patterns to make them look like ``good'' (i.e., near-capacity achieving) {\it erasure-correcting codes} \cite{luby2001efficient, richardson2001capacity}. The key insight is that we can effectively transform the sparse DFT computation problem into that of sparse-graph decoding to reconstruct the original ``message'' (i.e., sparse spectrum), which allows to use a simple peeling-based decoder with very low complexity.  The success of the FFAST algorithm depends on the {\it single-ton test} to pinpoint frequency bins containing only one ``erasure event'' (unknown non-zero DFT coefficient).  Given such a single-ton bin, the value and location of the coefficient can be obtained and then removed from other bins.  This procedure iterates until no more single-ton bins are found. In the same spirit of \cite{pawar2013computing}, the SparseFHT algorithm in \cite{scheibler2013fast} elegantly computes a $K$-sparse WHT of $\mathbf{x}$ using ${O}(K \log N)$ samples and ${O}(K\log^2 N)$ operations.

%%%

\section{Main Idea: A Simple Example}\label{sec:simple_example}
Since the sparsity is much smaller than the input dimension $K\ll N$, it is desirable if we can compute the WHT using very few samples  $M\ll N$ without reading the {\it entire} signal. The most straightforward way to reduce the number of samples to process is to {\it subsample}. However, from a reconstruction perspective, it is generally disastrous to subsample since it creates aliasing in the spectral domain that mixes the WHT coefficients $X[\mathbf{k}]$. 

The key idea of our SPRIGHT framework is to embrace (rather than avoid) the aliasing pattern as a form of ``alias code'', which is induced by the subsampling patterns guided by coding-theoretic designs, and more specifically, sparse-graph codes such as Low Density Parity Check (LDPC) codes. Then, our SPRIGHT framework exploits the aliasing pattern (alias code) to reconstruct the sparse Walsh spectrum in the presence of noise, by uncovering the sparse coefficients one-by-one iteratively in the spirit of decoding over noisy channels. While the design philosophy is similar to the FFAST algorithm in \cite{pawar2013computing} and the SparseFHT algorithm in \cite{scheibler2013fast}, our framework non-trivially generalizes this to the noisy scenario by robustifying the ``alias code'' for noisy decoding. Interestingly, we show that our framework can maintain the same scaling in both sample complexity and computational complexity as that in the noiseless case \cite{scheibler2013fast}. For completeness, we will repeat the noiseless design in the sequel, but using our setup and terminology.

\subsection{Subsampling and Aliasing}\label{sec:substream}
Our observation model is based on using multiple {\it basic observation sets} formed by {\it randomized subsampling} and {\it tiny-sized WHTs}, where each set contains $B=2^b$ (for some $b>0$) samples obtained as:
\begin{itemize}
	\item {\it Subsampling}: consider some integer $b<n$, the subsampling of noisy signal $u[\mathbf{m}]$ in \eqref{time_domain_samples} is performed by isolating a subset of $B=2^b$ samples indexed by $\mathbf{m}=\mathbf{M}\bdsb{\ell}+\mathbf{d}$ for $\bdsb{\ell}\in\GF^b$, where $\mathbf{M} \in \GF^{n\times b}$ is some binary matrix and $\mathbf{d}\in\GF^n$ is some random binary vector. In other words, after generating $\mathbf{M}\in \GF^{n\times b}$ and $\mathbf{d}\in\GF^n$, the subset of samples are selected by running the $b$-tuple $\bdsb{\ell}$ over $\GF^b$. 	
	
	\item {\it $B$-point WHT}: a much smaller $B$-point WHT is performed over the samples $u[\mathbf{M}\bdsb{\ell}+\mathbf{d}]$ for $\bdsb{\ell}\in\GF^b$. The subsampled signal has an aliased WHT spectrum readily obtained by a $B$-point WHT
\begin{align}
	{U}[\bdsb{j}]  
	&= \sum_{\bdsb{\ell}\in\GF^b} u[\mathbf{M}\bdsb{\ell}+\mathbf{d}](-1)^{\ip{\bdsb{j}}{\bdsb{\ell}}},\quad \bdsb{j}\in\GF^b.
\end{align} 
\end{itemize}

\begin{example}\label{example:subsampling}
We consider an example with $n=4$ and sparsity $K=B=2^b=4$ (i.e. $b=2$). For simplicity, we construct $2$ sets of observations using
\begin{align}
	\mathbf{M}_1 &= [\bdsb{0}_{2\times 2}^T, \mathbf{I}_{2\times 2}^T]^T,
	\quad
	\mathbf{M}_2 = [\mathbf{I}_{2\times 2}^T, \bdsb{0}_{2\times 2}^T]^T.
\end{align}
We call each set of observations using a different subsampling pattern a {\it subsampling group}. With these patterns, we access the following samples in each group for $\bdsb{\ell}=[\ell_1,\ell_2]^T\in\GF^2$
\begin{align*}
	u[\mathbf{M}_1\bdsb{\ell}] 
	= 
	u[0~0~\ell_1~\ell_2]
	\Longrightarrow
	\begin{cases}
		u[0000]\\
		u[0001]\\
		u[0010]\\
		u[0011]
	\end{cases},\quad
	u[\mathbf{M}_2\bdsb{\ell}]
	=
	u[\ell_1~\ell_2~0~0]
	\Longrightarrow
	\begin{cases}
		u[0000]\\
		u[0100]\\
		u[1000]\\
		u[1100]
	\end{cases}.
\end{align*}
After performing a $4$-point WHT on each set of these samples, we have $2$ sets of noisy observations:
\begin{align*}
		{U}_1[00] &= {X}[0000]+{X}[0100]+{X}[1000]+{X}[1100] \quad + \quad W_1[00]\\
		{U}_1[01] &= {X}[0001]+{X}[0101]+{X}[1001]+{X}[1101] \quad + \quad W_1[01]\\
		{U}_1[10] &= {X}[0010]+{X}[0110]+{X}[1010]+{X}[1110] \quad + \quad W_1[10]\\
		{U}_1[11] &= {X}[0011]+{X}[0111]+{X}[1011]+{X}[1111] \quad + \quad W_1[11]\\
		{U}_2[00] &= {X}[0000]+{X}[0001]+{X}[0010]+{X}[0011] \quad + \quad W_2[00]\\
		{U}_2[01] &= {X}[0100]+{X}[0101]+{X}[0110]+{X}[0111] \quad + \quad W_2[01]\\
		{U}_2[10] &= {X}[1000]+{X}[1001]+{X}[1010]+{X}[1011] \quad + \quad W_2[10]\\
		{U}_2[11] &= {X}[1100]+{X}[1101]+{X}[1110]+{X}[1111] \quad + \quad W_2[11].	\end{align*}
\end{example}

\subsection{Computing Sparse WHT as Sparse-Graph Decoding}\label{sec:simple_example_peeling}
In the presence of noise, the coefficients $X[\mathbf{k}]$ should be intuitively obtained as the ``least-squares'' solution over the $2$ sets of $B$ observations in Example \ref{example:subsampling}. However, the linear regression problem is underdetermined as we are given $8$ equations with $16$ unknowns. Fortunately, the coefficients are sparse, and this helps significantly. For simplicity, suppose that the $4$ non-zero coefficients are ${X}[0100]=2, {X}[0110]=4, {X}[1010]=1$ and ${X}[1111]=1$. Now we have $8$ equations with $4$ unknowns (non-zero), but we do not know which unknowns are non-zero. Then, we have
\begin{align*}
		{U}_1[00] &= {X}[0100]+W_1[00],~~~~~~~~~~~~~~~~~~~~~\, {U}_2[00] = W_2[00]\\
		{U}_1[01] &= W_1[01],  ~~~~~~~~~~~~~~~~~~~~~~~~~~~~~~~~~~~~~~~~\, {U}_2[01] = {X}[0100] + {X}[0110] + W_2[01]\\
		{U}_1[10] &= {X}[0110]+{X}[1010] + W_1[10], ~~~{U}_2[10] = {X}[1010] + W_2[10]\\
		{U}_1[11] &= {X}[1111] + W_1[11], ~~~~~~~~~~~~~~~~~~~~~\, {U}_2[11] = {X}[1111]+W_2[11].
\end{align*}

Now this problem seems quite a bit less daunting since the number of equations is more than the number of unknowns. The challenging part, however, is that {\it we do not know in advance which coefficients $X[\mathbf{k}]$ exist} in the equation since the sparse coefficients are randomly chosen over $\mathbf{k}\in\GF^n$. Here, we illustrate the principle of our recovery algorithm through the same simple example by showing that the recovery is an instance of sparse-graph decoding with the help of an ``oracle'' (described later). Then in the next subsection, we will introduce how to get rid of the oracle.

\subsubsection{Oracle-based Sparse-Graph Decoding}\label{sec:peeling_based}

The relationship between the observations $\{{U}_i[\bdsb{j}]\}_{i=1,2}^{\bdsb{j}\in\GF^b}$ and the unknown coefficients $X[\mathbf{k}]$ can be shown as a bipartite graph in \figref{fig:example_bipartite}, where the left nodes (unknown coefficients $X[\mathbf{k}]$) and right nodes (observations $\{{U}_i[\bdsb{j}]\}_{i=1,2}^{\bdsb{j}\in\GF^b}$) are referred to as the {\it variable nodes} and {\it check nodes} respectively in the language of sparse-graph codes.  Depending on the connectivity of the sparse bipartite graph, we categorize the observations into the following types:
\begin{enumerate}
	\item {\it Zero-ton}: a check node is a zero-ton if it has no non-zero coefficients (e.g., the color {\it blue} in \figref{fig:example_bipartite}). 
	\item {\it Single-ton}: a check node is a single-ton if it involves only one non-zero coefficient (e.g., the color {\it yellow} in \figref{fig:example_bipartite}).  Specifically, we refer to the index $\mathbf{k}$ and its associated value $X[\mathbf{k}]$ as the {\it index-value pair} $(\mathbf{k},X[\mathbf{k}])$.
	\item {\it Multi-ton}: a check node is a multi-ton if it contains more than one non-zero coefficient (e.g., the color {\it red} in \figref{fig:example_bipartite}). 
\end{enumerate}

\begin{wrapfigure}{r}{0.4\textwidth}
\centering
\vspace{-0.3cm}
\includegraphics[width=1\linewidth]{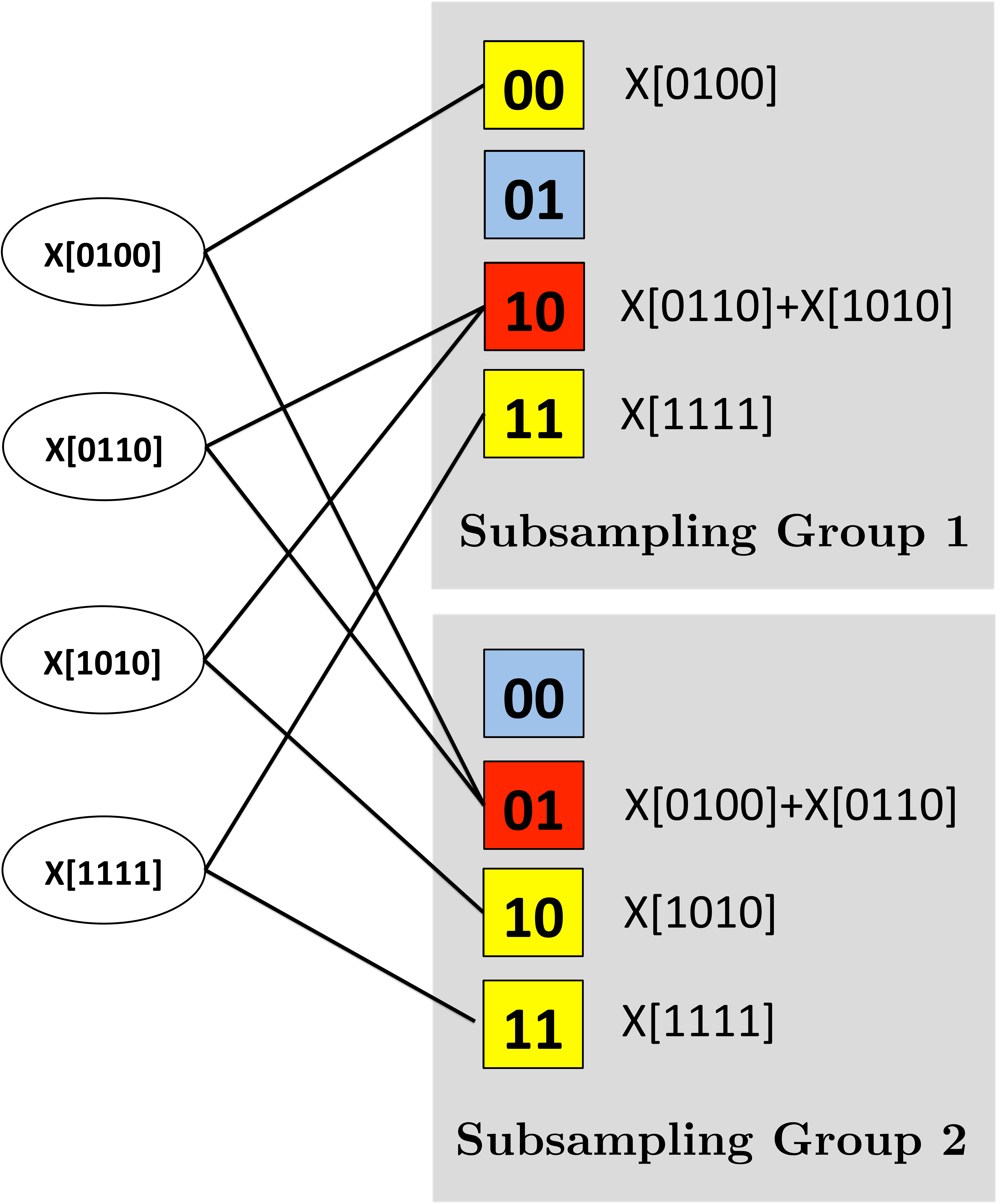}
\caption{Example of a sparse bipartite graph consisting of $4$ (non-zero) left nodes (variable nodes)  connected to the $2$ subsampling groups as a result of the sub-sampling-based randomized hashing in each group. Blue color represents ``zero-ton'', yellow color represents ``single-ton'' and red color represents ``multi-ton''.}\label{fig:example_bipartite}
\vspace{-1cm}
\end{wrapfigure}

To illustrate our reconstruction algorithm, we assume that there exists an ``oracle''  that informs the decoder exactly which check nodes are {\it single-tons}. Furthermore, the oracle further provides the index-value pair for that single-ton. In this example, the oracle informs the decoder that check nodes labeled ${U}_1[00]$, ${U}_1[11]$, ${U}_2[10]$ and ${U}_2[11]$ are single-tons with index-value pairs $(0100, {X}[0100])$, $(1111,{X}[1111])$, $(1010,{X}[1010])$ and $(1111,{X}[1111])$ respectively. Then the decoder can subtract their contributions from other check nodes, forming new single-tons. Therefore generally speaking, with the oracle information, the peeling decoder repeats the following steps:
\begin{itemize}
	\item[\bf Step (1)] select all the edges in the bipartite graph with right degree $1$ (identify single-ton bins);
	\vspace{-0.1cm}		
	\item[\bf Step (2)] remove (peel off) these edges as well as the corresponding pair of variable and check nodes connected to these edges.
	\vspace{-0.1cm}		
	\item[\bf Step (3)] remove (peel off) all other edges connected to the variable nodes that have been  removed in {\it Step (2)}. 
	\vspace{-0.1cm}	
	\item[\bf Step (4)] subtract the contributions of the variable nodes from the check nodes whose edges have been removed in {\it Step (3)}.
	\vspace{-0.1cm}
\end{itemize}
Finally, decoding is successful if all the edges are removed from the graph together with all the unknown coefficients $X[\mathbf{k}]$ such that all the WHT coefficients are decoded.

\subsubsection{Getting Rid of the Oracle : Bin Detection}\label{simple_example_bin_detection}

Since the oracle information is critical in the peeling process, we proceed with our example and explain briefly how to obtain such information without an oracle. We call this procedure ``bin detection''. For simplicity, we illustrate the design where the samples are noise-free. To obtain the oracle information, we exploit the diversity of using different offsets. For instance, in group $1$, we use the subsampling matrix $\mathbf{M}_1$ and the following set of offsets 
\begin{align*}	
	\mathbf{d}_{1,0} &= [0,0,0,0]^T, \quad
	\mathbf{d}_{1,1} = [1,0,0,0]^T, \quad
	\mathbf{d}_{1,2} = [0,1,0,0]^T,\quad
	\mathbf{d}_{1,3} = [0,0,1,0]^T,\quad
	\mathbf{d}_{1,4} = [0,0,0,1]^T.	
\end{align*}
In this way, using the subsampling pattern $\mathbf{M}_1$ and the offsets above, each check node is now assigned a $5$-dimensional vector $\mathbf{U}_1[\bdsb{j}]=[{U}_{1,0}[\bdsb{j}], {U}_{1,1}[\bdsb{j}], {U}_{1,2}[\bdsb{j}], {U}_{1,3}[\bdsb{j}], {U}_{1,4}[\bdsb{j}]]^T$, where ${U}_{1,p}[\bdsb{j}]$ is associated with the $p$-th offset $\mathbf{d}_{1,p}$ for $p=0,1,\cdots,4$. We call each vector of observations $\mathbf{U}_c[\bdsb{j}]$ in one group the {\it bin observation vector} $\bdsb{j}$. For example, the bin observation vectors for group $1$ are obtained as $\mathbf{U}_1[00]=\mathbf{0}$ and
\begin{align*}
	\mathbf{U}_1[01]&=
	{X}[0100]
	\begin{bmatrix}
	1\\
	(-1)^{0}\\
	(-1)^{1}\\
	(-1)^{0}\\	
	(-1)^{0}
	\end{bmatrix},
	~
	\mathbf{U}_1[10] =
	 {X}[0110]
	\begin{bmatrix}
	1\\
	(-1)^{0}\\
	(-1)^{1}\\
	(-1)^{1}\\	
	(-1)^{0}
	\end{bmatrix}
	 +{X}[1010]
	\begin{bmatrix}
	1\\
	(-1)^{1}\\
	(-1)^{0}\\
	(-1)^{1}\\	
	(-1)^{0}
	\end{bmatrix},
	~
	\mathbf{U}_1[11]=
	{X}[1111]
	\begin{bmatrix}
	1\\
	(-1)^{1}\\
	(-1)^{1}\\
	(-1)^{1}\\	
	(-1)^{1}
	\end{bmatrix}.		
\end{align*}
Now with these bin observations, one can effectively determine if a check node is a zero-ton, a single-ton or a multi-ton. We go through some examples:
\begin{itemize}
	\vspace{-0.1cm}	
	\item {\it zero-ton bin}: consider the zero-ton check node $\mathbf{U}_1[00]$. A zero-ton check node can be identified easily since the measurements are all zero $\mathbf{U}_1[00] =\mathbf{0}$.
	\vspace{-0.1cm}		
	\item {\it multi-ton bin}: consider the multi-ton check node $\mathbf{U}_1[10]$. A multi-ton can be easily identified since the magnitudes are not identical $|{U}_{1,0}[10]|\neq |{U}_{1,1}[10]|\neq|{U}_{1,2}[10]|\neq |{U}_{1,3}[10]|\neq |{U}_{1,4}[10]|$ or namely, the following ratio condition is not met:
	\begin{align}
		\frac{U_{1,p}[10]}{U_{1,0}[10]} \neq \pm 1,\quad p =1, 2, 3, 4.
	\end{align}	
	Therefore, if the ratio test does not produce $\pm 1$ or the magnitudes are not identical, we can conclude that this check node is a multi-ton.
	\vspace{-0.1cm}	
	\item {\it single-ton bin}: consider the single-ton check node $\mathbf{U}_1[01]$. The underlying node is a single-ton if $|{U}_{1,0}[01]|=|{U}_{1,1}[01]|=|{U}_{1,2}[01]|=|{U}_{1,3}[01]|=|{U}_{1,4}[01]|$, or namely the ratio test produces all $\pm 1$. Then, the index $\mathbf{k}=[k[1],k[2],k[3],k[4]]^T$ of a single-ton can be obtained by a simple ratio test
	\begin{align*}
		&
		\begin{cases}
		(-1)^{\widehat{k}[1]} &= \displaystyle \frac{{U}_{1,1}[01]}{{U}_{1,0}[01]} = (-1)^{0}\\
		(-1)^{\widehat{k}[2]} &= \displaystyle \frac{{U}_{1,2}[01]}{{U}_{1,0}[01]} = (-1)^{1}\\
		(-1)^{\widehat{k}[3]} &= \displaystyle \frac{{U}_{1,3}[01]}{{U}_{1,0}[01]} = (-1)^{0}\\
		(-1)^{\widehat{k}[4]} &= \displaystyle \frac{{U}_{1,4}[01]}{{U}_{1,0}[01]} = (-1)^{0}			
		\end{cases}
		\Longrightarrow
		\begin{cases}
		\widehat{k}[1] = 0\\
		\widehat{k}[2] = 1\\
		\widehat{k}[3] = 0\\
		\widehat{k}[4] = 0\\
		\widehat{{X}}[\widehat{\mathbf{k}}] = {U}_{1,0}[01]
		\end{cases}
	\end{align*}
	Both the ratio test and the magnitude constraints are easy to verify for all check nodes such that the index-value pair is obtained for peeling.
\end{itemize}

This simple example shows how the problem of recovering the $K$-sparse coefficients $X[\mathbf{k}]$ can be cast as an instance of oracle-based peeling decoding by proper subsampling-induced {\it sparse bipartite graphs} in the dual domain. It further shows that the freedom in choosing offsets $\mathbf{d}$ gets rid of the oracle by {\it bin detection}. However, this simple example will not work in the presence of noise. The key idea of our design is that by carefully choosing the offsets $\mathbf{d}$ and subsampling patterns $\mathbf{M}$ through a sparse-graph coding lens, we can induce ``peeling-friendly'' sparse bipartite graphs that lead to fast recovery of the unknown WHT coefficients even in the presence of noise, as illustrated next.

%%%

\section{The SPRIGHT Framework: General Architecture and Algorithm}\label{sec:randomized_hashing}
In this section, we generalize the simple example and present the our proposed SPRIGHT framework. Our framework consists of an {\it observation generator} and a {\it reconstruction engine}, as shown in \figref{fig:schematic}.

\begin{figure}[h]
	\centering
	\includegraphics[width=0.85\linewidth]{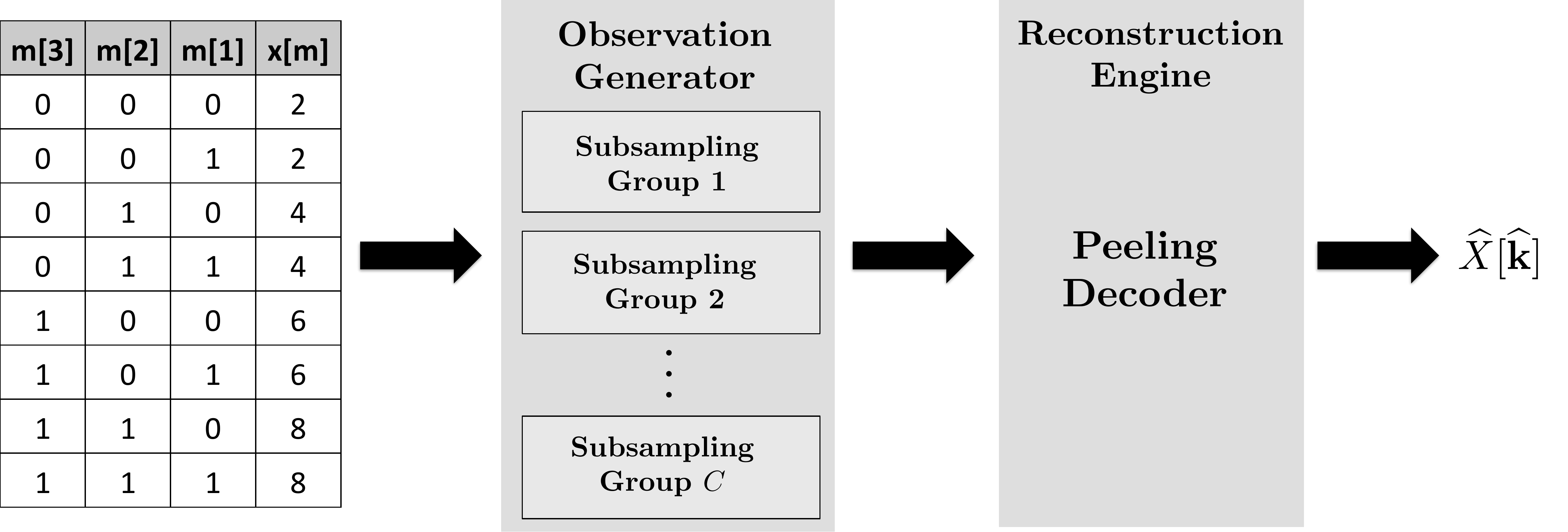}
	\caption{The conceptual diagram of our learning framework with $C$ subsampling groups, where each group generates $P$ basic query sets, each of size $B=2^b$.}\label{fig:schematic}
\end{figure}

\subsection{Observation Generator: Subsampling and Aliasing}\label{sec:frontend-architecture}
In our SPRIGHT framework, the observations are obtained from $C$ {\it subsampling groups}, where each group generates $P$ {\it basic observation sets} of size $B=2^b$. Each group uses a different matrix $\mathbf{M}_c\in\GF^{n\times b}$ and a different set of $P$ offsets $\mathbf{d}_{c,p}\in \mathbb{F}_2^n$ for $p\in[P]$, as summarized in \algref{alg:subsampling}.

%
%\begin{wrapfigure}{r}{0.5\textwidth}
%\vspace{-0.55cm}
%\begin{minipage}{0.5\textwidth}
\begin{algorithm}[H] 
  \caption{Subsampling and WHT}\label{alg:subsampling}
  \begin{algorithmic}
    \STATE ${\tt Input:}$ $u[\mathbf{m}]$ for $\mathbf{m} \in \GF^n$ with $N=2^n$;  
    \STATE ${\tt Set}:$ the number of subsampling groups $C$; observation set size $B$ and number of observation sets $P$.
    \STATE ${\tt Generate}:$ offsets $\mathbf{d}_{c,p}$ for $p \in [P]$; subsampling matrix $\mathbf{M}_c\in\GF^{n\times b}$ for some $b>0$
    \FOR{ $c=1$ \TO $C$}
    	\FOR{$p=1$ \TO $P$}
    	\STATE %Subsample and perform a $B$-point WHT: 
	$U_{c,p}[\bdsb{j}] = \sqrt{\frac{N}{B}} \sum_{\bdsb{\ell}\in\GF^b} u[\mathbf{M}_c \bdsb{\ell}+\mathbf{d}_{c,p}](-1)^{\ip{\bdsb{j}}{\bdsb{\ell}}}$.
	\ENDFOR		
    \ENDFOR
  \end{algorithmic}
\end{algorithm}
%\end{minipage}
%\vspace{0.35cm}
%\end{wrapfigure}
%

\begin{prop}[\bf Basic Observation Model]\label{prop_hashing_obs}
The $B$-point WHT coefficients indexed by $\bdsb{j}\in\GF^b$ can be written as:
\begin{align}\label{U_cp}
	{U}_{c,p}[\bdsb{j}] = \sum_{\mathbf{M}_c^T\mathbf{k} = \bdsb{j}} X[\mathbf{k}](-1)^{\ip{\mathbf{d}_{c,p}}{\mathbf{k}}}+{W}_{c,p}[\bdsb{j}],\quad p \in [P], 
\end{align}
where ${W}_{c,p}[\bdsb{j}] = \sum_{\mathbf{M}_c^T\mathbf{k} = \bdsb{j}} {W}[\mathbf{k}](-1)^{\ip{\mathbf{d}_{c,p}}{\mathbf{k}}}$ and ${W}[\mathbf{k}]$ is the WHT coefficient of noise samples $w[\mathbf{m}]$. 
\end{prop}

Clearly, the $\bdsb{j}$-th WHT coefficient ${U}_{c,p}[\bdsb{j}]$ in each observation set is an aliased version (hash output) of the Walsh spectral coefficient $X[\mathbf{k}]$ under the hash function $\mathcal{H}_c : \GF^n \rightarrow \GF^b$ in the $c$-th group
\begin{align}\label{eq:hash_function}
	\bdsb{j} = \mathcal{H}_c(\mathbf{k}) = \mathbf{M}_c^T\mathbf{k},\quad c\in[C].
\end{align}
It can be observed that the aliasing pattern (hash function) is invariant with respect to the offsets $\mathbf{d}_{c,p}$ used in subsampling. Similar to the {\it bin observation vector} in the simple example from Section \ref{simple_example_bin_detection}, we can regroup the observations ${U}_{c,p}[\bdsb{j}]$ according to the hash $\mathcal{H}_c(\bdsb{j})$
\begin{align}
	\mathbf{U}_c[\bdsb{j}]\triangleq [\cdots, {U}_{c,p}[\bdsb{j}],\cdots]^T, 
\end{align}
by stacking the $\bdsb{j}$-th WHT coefficient associated with all the offsets across the $P$ observation sets in a vector. 

\begin{prop}[\bf Bin Observation Model]\label{prop_meas.bin.model}
Given the offset matrix $\mathbf{D}_c\defn[\cdots;\mathbf{d}_{c,p};\cdots]\in\GF^{P\times n}$,  the $\bdsb{j}$-th bin observation vector in the $c$-th group can be written as
\begin{align}\label{meas.bin.model}
	\mathbf{U}_c[\bdsb{j}]
	=
	\sum_{\mathbf{k}:~\mathbf{M}_c^T\mathbf{k}=\bdsb{j}}
	X[\mathbf{k}](-1)^{\mathbf{D}_c\mathbf{k}}
	+
	\mathbf{W}_c[\bdsb{j}], \quad \bdsb{j}\in\GF^b,~c \in [C],
\end{align}
where $(-1)^{(\cdot)}$ is the element-wise exponentiation operator and $\mathbf{W}_c[\bdsb{j}]=\sum_{\mathbf{M}_c^T\mathbf{k} = \bdsb{j}} {W}[\mathbf{k}](-1)^{\mathbf{D}_c\mathbf{k}}$ is the noise vector with $W[\mathbf{k}]$ being the WHT coefficient of the noise $w[\mathbf{m}]$.
\end{prop}
\begin{proof}
The proof follows from WHT properties similar to that in \cite{scheibler2013fast}, and hence is omitted here.
\end{proof}
From a coding-theoretic perspective, the observation vectors $\mathbf{U}_c[\bdsb{j}]$ for $\bdsb{j}\in\GF^b$ across different groups $c\in[C]$ constitute the parity constraints of the coefficients $X[\mathbf{k}]$, where $X[\mathbf{k}]$ enters the $\bdsb{j}$-th parity of group $c$ if $\mathbf{M}_c^T\mathbf{k}=\bdsb{j}$.
It can be shown that if the set size $B=2^b$ and the number of subsampling groups $C$ are chosen properly, the bin observation vectors constitute parities of good error-correcting codes. Therefore, the coefficients can be uncovered iteratively in the spirit of peeling decoding (see Section \ref{sec:noiseless_backend}), similar to that in LDPC codes. The key idea is to to avoid excessive aliasing by maintaining $B$ on par with the sparsity ${O}(K)$ and imposing $P={O}(\log N)$ for denoising purposes in bin detection. To keep our discussions focused, we defer the specific constructions of the subsampling model in terms of $C$, $B=2^b$ and $\{\mathbf{M}_c\}_{c\in[C]}$ in Section \ref{sec:frontend-design}.

\subsection{Reconstruction Engine: Peeling Decoder}\label{sec:noiseless_backend}
The outputs from the subsampling operation are then used for reconstruction. As stated in Proposition \ref{prop_meas.bin.model}, each bin observation vector consists of linear combinations of the unknown WHT coefficients, which can be characterized by a {\it sparse bipartite graph} consisting of $K$ left nodes (variable nodes) and $CB$ right nodes (check nodes). 

\begin{defi}[\bf Random Graph Ensemble]\label{def:graph_ensemble}
For some {\it redundancy parameter} $\eta>0$ let $B=\eta K=2^b$ for some $b>0$. The graph ensemble $\mathcal{G}(K,\eta,C,\{\mathbf{M}_c\}_{c\in[C]})$ consists of left $C$-regular sparse bipartite graphs where 
\begin{itemize}
	\item there are $K$ left nodes (variable nodes), each labeled by a distinct element from the support $\mathbf{k}\in\mathcal{K}$;
	\item there are $B=2^b$ right nodes (check nodes) per group, each labeled by the bin index $\bdsb{j}\in\GF^b$ and assigned the bin observation vector $\mathbf{U}_c[\bdsb{j}]$;
	\item each left node $\mathbf{k}$ has degree $C$ and each edge is connected to a right node $\bdsb{j}$ in each group according to the hash function $\mathcal{H}_c : \GF^n \rightarrow \GF^b$ given in \eqref{eq:hash_function}. 
\end{itemize}	
\end{defi}

Based on our simple example in Section \ref{sec:simple_example_peeling}, the unknown WHT coefficients (i.e. variable nodes) can be recovered through a peeling decoder over the graph ensemble $\mathcal{G}(K,\eta,C,\{\mathbf{M}_c\}_{c\in[C]})$, as summarized in \algref{alg:peeling}. The key is to distinguish the observations $\mathbf{U}_c[\bdsb{j}]$ and identify single-ton bins for peeling.

In \algref{alg:peeling}, we denote the bin detection routine 
\begin{align}
	\psi : \mathbb{R}^P \rightarrow ({\tt type}, \widehat{\mathbf{k}},\widehat{X}[\widehat{\mathbf{k}}])
\end{align}
which determines the types of bin observations:
\begin{enumerate}
	\item $\mathbf{U}_c[\bdsb{j}]$ is a {\it zero-ton} if there does not exist $X[\mathbf{k}]\neq 0$ such that $\mathbf{M}_c^T\mathbf{k}=\bdsb{j}$, denoted by $\mathbf{U}_c[\bdsb{j}] \sim \mathcal{H}_{\textrm{Z}}$;
	\item $\mathbf{U}_c[\bdsb{j}]$ is a {\it single-ton} with the index-value pair $(\mathbf{k},X[\mathbf{k}])$ if there exists only one $X[\mathbf{k}]\neq 0$ such that $\mathbf{M}_c^T\mathbf{k}=\bdsb{j}$, denoted by $\mathbf{U}_c[\bdsb{j}] \sim \mathcal{H}_{\textrm{S}}(\mathbf{k},X[\mathbf{k}])$;
	\item $\mathbf{U}_c[\bdsb{j}]$ is a {\it multi-ton} if there exist more than one $X[\mathbf{k}]\neq 0$ such that $\mathbf{M}_c^T\mathbf{k}=\bdsb{j}$, denoted by $\mathbf{U}_c[\bdsb{j}] \sim \mathcal{H}_{\textrm{M}}$.
\end{enumerate}

\begin{figure}
\begin{minipage}{0.5\textwidth}
\begin{algorithm}[H] 
  \caption{Peeling Decoder}\label{alg:peeling}  
  \begin{algorithmic}
    \STATE ${\tt Input}:$ observation vectors $\mathbf{U}_c[\bdsb{j}]$ for $\bdsb{j} \in \GF^b$, $c\in[C]$;
    \STATE ${\tt Set}:$ the number of peeling iterations $I$;  
    \FOR{$i=1$ \TO $I$}
    	\FOR{$c=1$ \TO $C$}
		\FOR{$\bdsb{j}\in\GF^b$}
			\STATE $({\tt type},\widehat{\mathbf{k}},\widehat{{X}}[\widehat{\mathbf{k}}]) = \psi(\mathbf{U}_c[\bdsb{j}])$.		
			\IF{${\tt type}= \textrm{single-ton}$}
			\STATE ${\tt Peel~off}$ for all $p=[P], c'=[C]$
			\STATE ${\tt Locate~bin~index}$ $\bdsb{j}_{c'} = \mathbf{M}_{c'}^T \widehat{\mathbf{k}}$\\
			\STATE ${U}_{c',p}[\bdsb{j}_{c'}] \leftarrow {U}_{c',p}[\bdsb{j}_{c'}] - \widehat{{X}}[\widehat{\mathbf{k}}](-1)^{\ip{\mathbf{d}_{c,p}}{\widehat{\mathbf{k}}}}$.
			\ELSIF {${\tt type}\neq \textrm{single-ton}$}
				 \STATE continue to next $\bdsb{j}$.
			\ENDIF			
		\ENDFOR
	\ENDFOR
    \ENDFOR
  \end{algorithmic}
\end{algorithm}
\end{minipage}
~~~~
\begin{minipage}{0.5\textwidth}
\vspace{0.6cm}
\begin{center}
\includegraphics[scale=.27]{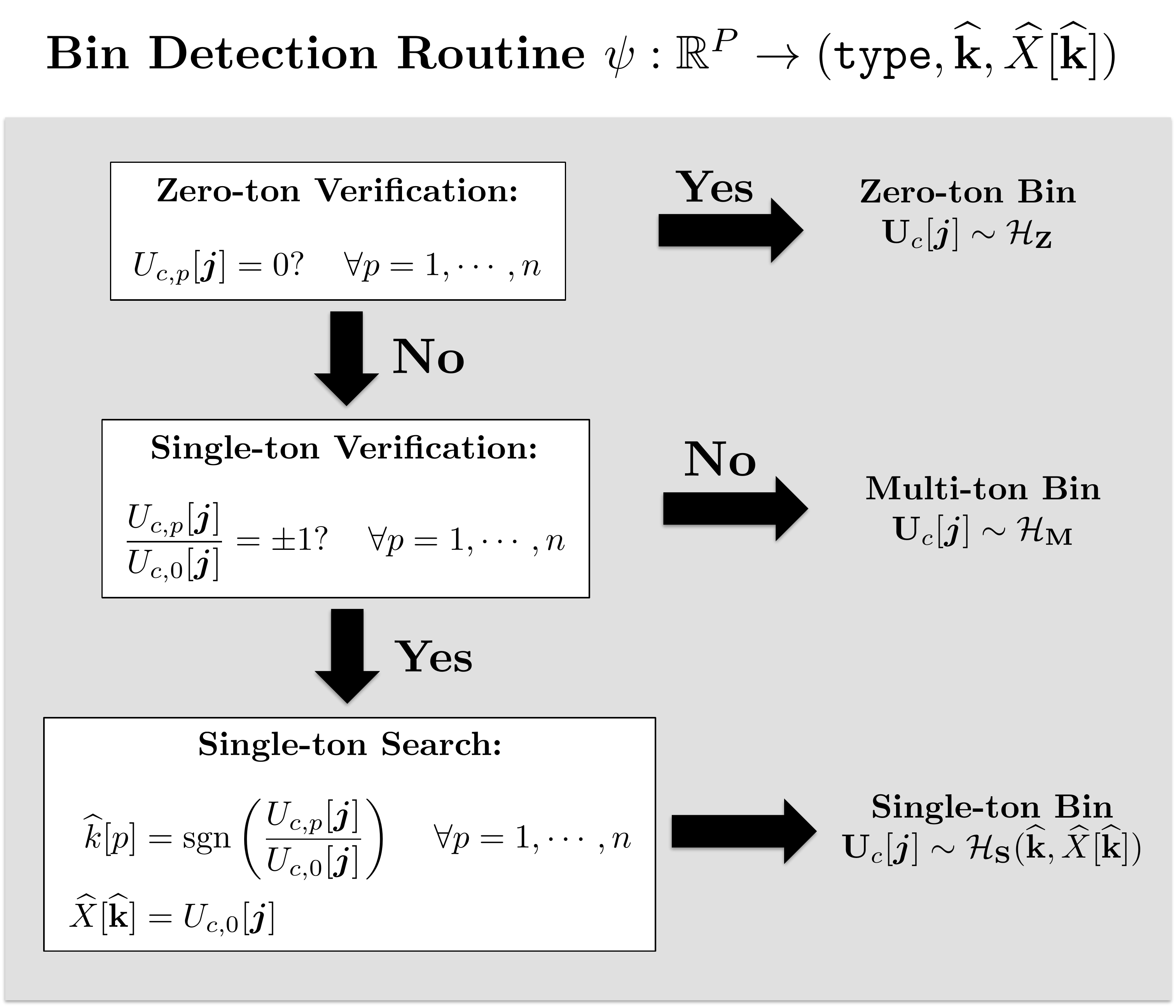}
\caption{The bin detection routine $\psi : \mathbb{R}^P \rightarrow ({\tt type}, \widehat{\mathbf{k}},\widehat{X}[\widehat{\mathbf{k}}])$ for the noiseless setting by choosing offsets $\mathbf{D}_c=\mathbf{I}_{n\times n}$.}
\label{fig:bindetectionroutine}
\end{center}
\end{minipage}
\end{figure}

We focus on the noiseless case here (generalization of the simple example), and then elaborate on {\it robust bin detection} in the presence of noise in Section \ref{sec:robust_bin_detection}. The noiseless bin detection requires $P=n$ offsets through the steps summarized in \figref{fig:bindetectionroutine}:
\begin{itemize}
	\vspace{-0.1cm}	
	\item $\mathbf{U}_c[\bdsb{j}] \sim \mathcal{H}_{\textrm{Z}}$ if $U_{c,p}[\bdsb{j}] = 0$ for all $p=1,\cdots,n$.
	\vspace{-0.1cm}	
	\item $\mathbf{U}_c[\bdsb{j}] \sim \mathcal{H}_{\textrm{M}}$ if $|{U}_{c,p}[\bdsb{j}]/{U}_{c,0}[\bdsb{j}]|\neq \pm 1$ for all $p=1,\cdots,n$.	
	\vspace{-0.1cm}	
	\item $\mathbf{U}_c[\bdsb{j}] \sim \mathcal{H}_{\textrm{S}}(\mathbf{k},X[\mathbf{k}])$ if the bin is neither a zero-ton nor a multi-ton. 
\end{itemize}
The index-value pair $(\mathbf{k},X[\mathbf{k}])$ of the single-ton is obtained as follows. Since each single-ton bin observation satisfies $U_{c,p}[\bdsb{j}] = X[\mathbf{k}](-1)^{\ip{\mathbf{d}_{c,p}}{\mathbf{k}}}$, 
the corresponding sign\footnote{Note that the definition of the sign function here is a bit different than usual, where $\sgn{x}=1$ if $x<0$ and $\sgn{x}=0$ if $x>0$.} satisfies 
\begin{align}
	\sgn{U_{c,p}[\bdsb{j}]} &= \ip{\mathbf{d}_{c,p}}{\mathbf{k}}\oplus\sgn{X[\mathbf{k}]},
\end{align}
where $\sgn{X[\mathbf{k}]}$ is the nuisance unknown sign. How do we get rid of such nuisance? This can be done by imposing a reference $\mathbf{d}_{c,0}=\mathbf{0}$ in addition to the offset matrix $\mathbf{D}_c\in\GF^{P\times n}$ such that $\sgn{U_{c,0}}=\sgn{X[\mathbf{k}]}$.

This gives us a set of linear equations with respect to the unknown index $\mathbf{k}$:
	\begin{align}\label{offset_as_code}
		\begin{bmatrix}
			\sgn{U_{c,1}[\bdsb{j}]}\oplus\sgn{U_{c,0}[\bdsb{j}]}\\			
			\sgn{U_{c,2}[\bdsb{j}]}\oplus\sgn{U_{c,0}[\bdsb{j}]}\\					
			\vdots\\
			\sgn{U_{c,n}[\bdsb{j}]}\oplus\sgn{U_{c,0}[\bdsb{j}]}
		\end{bmatrix}
		=		
		\mathbf{D}_c
		\mathbf{k}.
	\end{align}
	Clearly, if we choose the offsets in each group as $\mathbf{D}_c=\mathbf{I}_{n\times n}$, the unknown index $\mathbf{k}$ can be obtained directly from the signs of the observations. Finally, the value of the coefficient is obtained as $\widehat{{X}}[\widehat{\mathbf{k}}] = {U}_{c,0}[\bdsb{j}]$.

%\newpage
\subsection{Subsampling Design and Algorithm Guarantees}\label{sec:frontend-design}
With the general subsampling architecture given in Section \ref{sec:frontend-architecture}, we discuss the specific constructions of the graph ensemble $\mathcal{G}(K,\eta,C,\{\mathbf{M}_c\}_{c\in[C]})$ by choosing appropriately the observation set size $B=2^b$, the number of subsampling groups $C$, and the subsampling matrices $\{\mathbf{M}_c\}_{c\in[C]}$. We defer the discussion of how to choose offsets $\mathbf{d}_{c, p}$ to Section \ref{sec:robust_bin_detection} because its design is independent of the graph ensemble. 

\begin{wrapfigure}{r}{0.5\textwidth}
\vspace{-0.8cm}
\begin{center}
\includegraphics[scale=0.63]{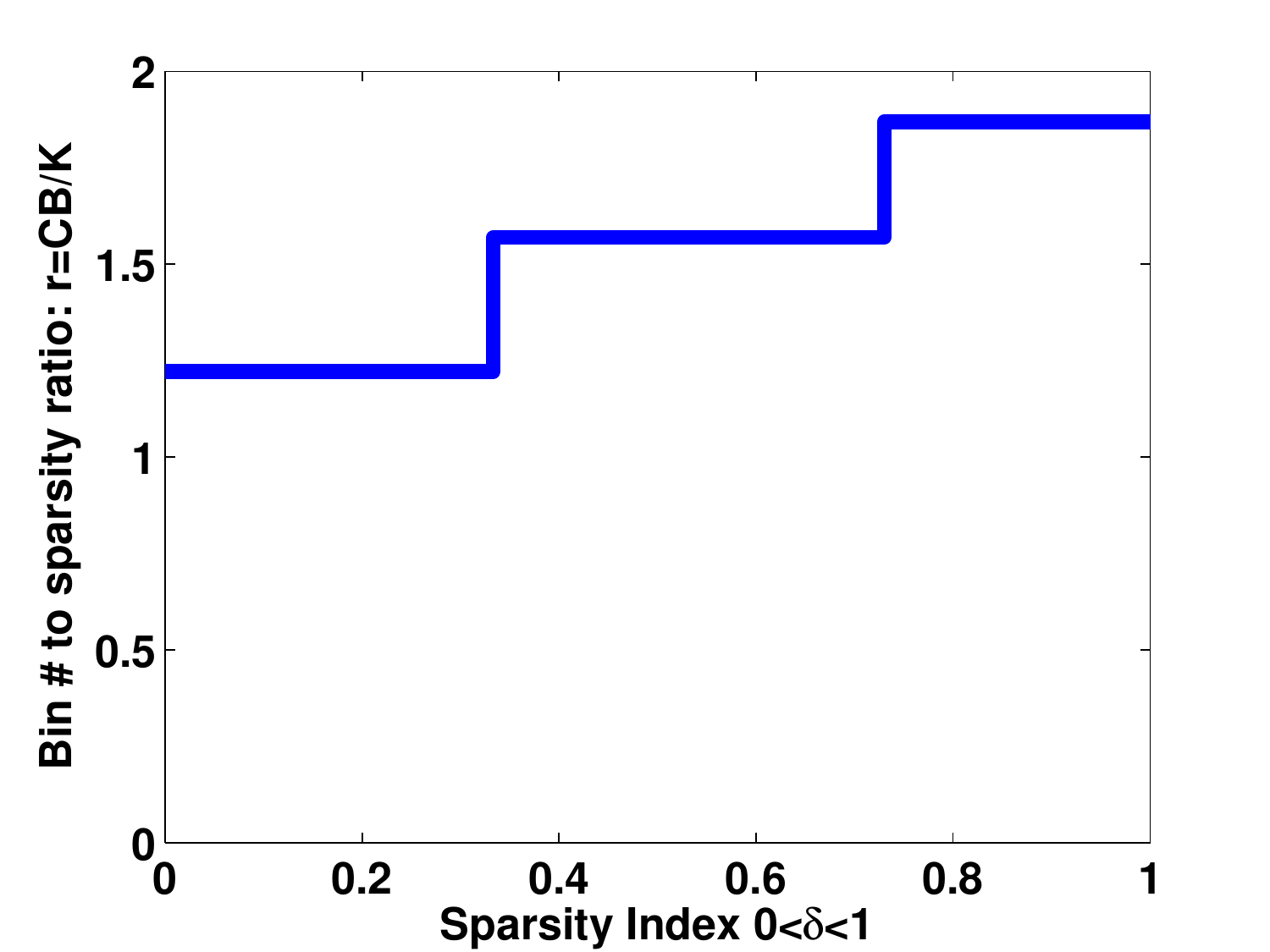}
\caption{The ratio of the total bin number to the sparsity $r=CB/K$ as a function of the index $\delta\in(0,0.99)$.}
\label{fig:sparse_index}
\end{center}
\vspace{-0.5cm}
\end{wrapfigure}

Let us first give some high level intuition of our subsampling design. Regardless of how many observation sets $P$ are generated in each subsampling group $c\in[C]$, it is desirable to keep the number of subsampling groups $C$ and the observation set size $B=2^b$ small such that the resulting sample complexity is small. However, if $C$ and $B$ are too small, the resulting observation bins will end up mostly with multi-tons so the peeling operations get stuck. As a result, the subsampling design is about finding the ``sweet spot'' for the number of subsampling groups $C$ and the observation set size $B$. In our analysis, we show that the product satisfies $CB = {O}(K)$, which implies that the subsampling using our generator does not introduce extra overheads other than a constant factor compared to the sparsity $K$. More importantly, from our analysis, such constant can be made explicit given the number of subsampling groups $C$. 

The subsampling design varies with the sparsity regime $0<\delta<1$ and hence, our results are stated with respect to different intervals of $\delta$ that cover the entire sparsity regime (see Appendix \ref{sec:analysis_peeling_decoder}). Our results stated below presents one constructive scheme using the partition\footnote{We choose to cover the regime $0<\delta\leq 0.99$ for the sake of presentation, and one can follow our proof in Appendix \ref{sec:analysis_peeling_decoder} to design subsampling patterns for $\delta > 0.99$.} $(0,1/3]\cup(1/3,0.73]\cup(0.73, 7/8]\cup(7/8,0.99]$. The sampling overhead (i.e. $CB/K$) introduced by the observation generator using this partition is shown in \figref{fig:sparse_index}. This is by no means the unique scheme and the reason for choosing $1/3$, $0.73$, $7/8$ and $0.99$ as break points is that we want to keep the number of intervals small for the sake of presentation, since each interval results in a different design. 

\begin{thm}[\bf Oracle-based Peeling Decoder Performance]\label{thm_peeling_decoder_general}
Consider an input vector with a $K$-sparse WHT such that $K={O}(N^\delta)$ for some $0<\delta < 1$. Given an observation generator with $C$ subsampling groups and an observation set size $B=\eta K$ for some $\eta>0$, the subsampling-induced graph ensemble $\mathcal{G}(K,\eta,C,\{\mathbf{M}_c\}_{c\in[C]})$ guarantees that with probability at least $1-{O}(1/K)$, the oracle-based peeling decoder recovers all $K$ unknown coefficients in time ${O}(K)$ as long as 
\begin{itemize}
	\item $C=3$ subsampling groups and $B\geq 0.4073K$ for $0<\delta\leq 1/3$ (see Section \ref{sec:delta1})
	\item $C=6$ subsampling groups and $B\geq 0.2616K$ for $1/3<\delta\leq 0.73$ (see Section \ref{sec:delta2}); 
	\item $C=8$ subsampling groups and $B\geq 0.2336K$ for $0.73<\delta\leq 0.875$ (see Section \ref{sec:delta3});
	\item $C=8$ subsampling groups and $B\geq 0.2336K$ for $0.875< \delta \leq 0.99$ (see Section \ref{sec:delta4}).
\end{itemize}
\end{thm}

\begin{proof}

Our analysis is similar to the arguments in \cite{luby2001efficient,richardson2001capacity} using the so-called {\it density evolution} analysis from modern coding theory, which tracks the average density\footnote{The density here refers to fraction of the remaining edges, or namely, the number of remaining edges divided by the total number of edges in the graph.} of the remaining edges in the graph at each peeling iteration of the algorithm. Although the proof techniques are similar to those from \cite{luby2001efficient} and \cite{richardson2001capacity}, the graph used in our peeling decoder is different from those in \cite{richardson2001capacity,luby2001efficient}. This leads to fairly important differences in the analysis, such as the degree distributions of the graphs and the expansion properties of the graphs (see Appendix \ref{sec:analysis_peeling_decoder}). Hence, we present an independent analysis here for our peeling decoder. In the following, we provide a brief outline of the proof elements highlighting the main technical components.

\begin{itemize}
	\item {\it Density evolution} in Lemma \ref{lem:DE_verysparse}: 
	We analyze the performance of our peeling decoder over a {\it typical graph} (i.e., cycle-free) of the
ensemble $\mathcal{G}(K,\eta,C,\{\mathbf{M}_c\}_{c\in[C]})$ for a fixed number of peeling iterations $i$. We assume that a local neighborhood of every edge in the graph is cycle-free (tree-like) and derive a recursive equation that represents the average density of remaining edges in the graph at iteration $i$. The recursive equation guarantees that the average density is shrinking as the iterations proceed, as long as the redundancy parameter $\eta$ is chosen accordingly with respect to the number of groups $C $ for subsampling.
	\item {\it Convergence to density evolution} in Lemma \ref{lem:convergence2DE}: 
	Using a Doob martingale argument \cite{richardson2001capacity} and \cite{pedarsani2014phasecode}, we show that the local neighborhood of most edges of a random graph from the ensemble $\mathcal{G}(K,\eta,C,\{\mathbf{M}_c\}_{c\in[C]})$ is cycle-free with high probability. This proves that with high probability, our peeling decoder removes all but an arbitrarily small fraction of the edges in the graph (i.e., the left nodes are removed at the same time after being decoded) in a constant number of iterations $i$. 
	\item {\it Graph expansion property} for complete decoding in Lemma \ref{lem_graph_expander}: 
	We show that if the sub-graph consisting of the remaining edges is an ``expander'' (as will be defined later in this section), and if our peeling decoder successfully removes all but a sufficiently small fraction of the left nodes from the graph, then it removes all the remaining edges of the graph successfully. As long as the number of subsampling groups $C$ is large enough for a given sparsity $\delta$, we show that our graph ensemble is an expander with high probability. This completes the decoding of all the non-zero WHT coefficients.
\end{itemize}
\end{proof}

\section{Robust Bin Detection}\label{sec:robust_bin_detection}
We have shown in Section \ref{sec:frontend-design} that given an oracle for bin detection, our subsampling design for any sparsity regime $0<\delta<1$ guarantees that peeling decoder successfully recovers all unknown WHT coefficients in the absence of noise. In the noisy scenario, it is critical to robustify the {\it bin detection} scheme by choosing subsampling offsets differently than the noiseless setting. In the following, we explain the robust bin detection routine $\psi$. For simplicity, we drop the group index $c$ and bin index $\bdsb{j}$ when we mention some bin observation. For example, the observation vector of some bin $\bdsb{j}$ from group $c$ is denoted by $\bdsb{U}=[\cdots,U_p,\cdots]^T$, where the associated set of offsets is $\mathbf{D}=[\mathbf{d}_1;\cdots;\mathbf{d}_P]\in\GF^{P\times n}$.

\subsection{Performance Guarantees of Robust Bin Detection}

\begin{wrapfigure}{r}{0.55\textwidth}
%\begin{minipage}{0.52\textwidth}
\vspace{-1cm}
\begin{center}
\includegraphics[scale=0.28]{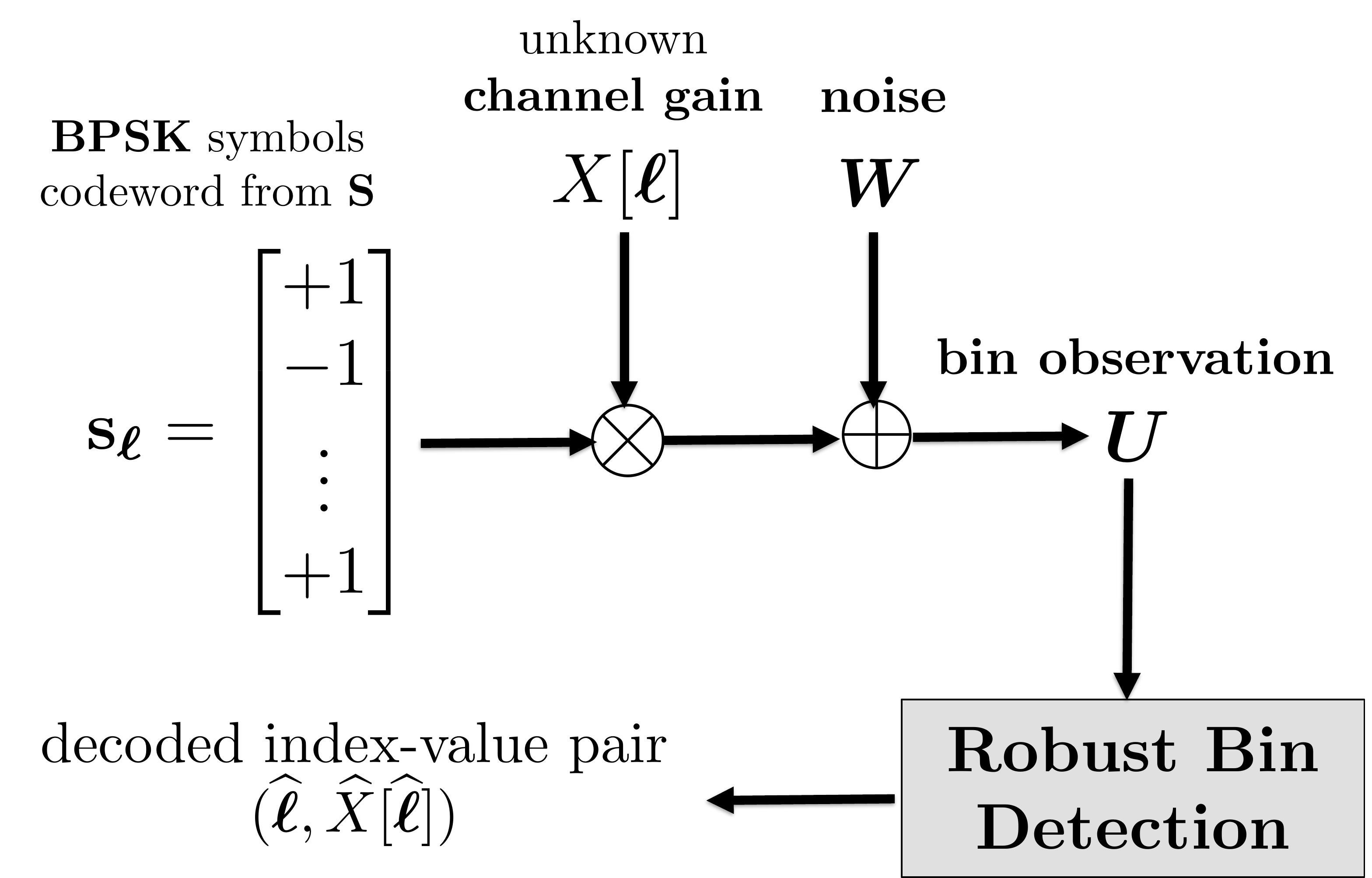}
\caption{An illustration of a single-ton detection.}
\vspace{-0.7cm}
\label{fig:single-ton_detection}
\end{center}
%\end{minipage}
\end{wrapfigure}

From the noiseless design given in Section \ref{sec:noiseless_backend}, we can see that the offset signature $(-1)^{\mathbf{D}\mathbf{k}}$ associated with each coefficient in Proposition \ref{prop_meas.bin.model} is the key to decode the unknown index-value pair $(\mathbf{k},X[\mathbf{k}])$ of a single-ton. Let $\mathbf{S}=[\cdots, \mathbf{s}_{\mathbf{k}} ,\cdots]$, where for each $\mathbf{k}\in\GF^n$ we denote by
\begin{align}
	\mathbf{s}_{\mathbf{k}} = (-1)^{\mathbf{D}\mathbf{k}}
\end{align}
the offset signature {\it codebook} associated with the offset matrix $\mathbf{D}$. Then in the presence of noise, the bin observation vector can be written as
\begin{align}\label{equiv.model.bin_simplified}
	\bdsb{U}  
	= \mathbf{S}\bdsb{\alpha} + \bdsb{W}
\end{align}
for some sparse vector $\bdsb{\alpha}=[\cdots,\alpha[\mathbf{k}],\cdots]^T$ such that $\alpha[\mathbf{k}]=X[\mathbf{k}]$ if $\mathbf{M}_c^T\mathbf{k}=\bdsb{j}$ and $\alpha[\mathbf{k}]=0$ if otherwise. Clearly, the sparsity of $\bdsb{\alpha}$ implies the type of the bin. For example, the underlying bin is a single-ton if it is $1$-sparse. It can be further shown from \eqref{prop_meas.bin.model} that $\bdsb{W}$ follows a multivariate Gaussian distribution with zero mean and a covariance $\mathbb{E}\left[\bdsb{W}\bdsb{W}^T\right]=\nu^2\mathbf{I}$ and $\nu^2\defn N\sigma^2/B$.

In the case of single-tons, the observation $\bdsb{U}$ can be regarded as the noise-corrupted version of some codeword from the codebook $\mathbf{S}$ (see \figref{fig:single-ton_detection}). In our noiseless design, each codeword $\mathbf{s}_{\mathbf{k}}\in\{-1,1\}^n$ encodes the $n$-bit index $\mathbf{k}$ into $n$ binary phase-shift keying (BPSK) symbols $(-1)^{\ip{\mathbf{d}_p}{\mathbf{k}}}\in\{\pm 1\}$ for $p \in [n]$. This set of $n$ BPSK symbols is scaled by the coefficient $X[\mathbf{k}]$ and observed as $U_p$ for $p\in[n]$. This resembles the communication scenario where the goal of a receiver is to decode a sequence of $n$ BPSK sequence with an unknown channel gain. Therefore, when there is additive noise in the channel, the codebook needs to be re-designed such that it can be robustly decoded.

In general, the vector $\bdsb{\alpha}$ is not necessarily $1$-sparse (multi-ton bin). Through the {\it robust bin detection} scheme, we can effectively detect out the bins carrying some $1$-sparse $\bdsb{\alpha}$ (i.e. single-tons), and recovers the index-value pair of the $1$-sparse coefficient. Then, as the peeling operations proceed, the non-zero coefficients in other bins carrying $\bdsb{\alpha}$ that is not $1$-sparse will be peeled off, which keeps forming new bins carrying $1$-sparse vectors (single-ton).  

In particular, we first present a straightforward design for {\it near-linear time detection} to shed some preliminary light on the noisy design, and then proceed to our proposed {\it sub-linear time detection} schemes. More specifically, we have two sub-linear time detection schemes that impose different sample complexities and computational complexities, called the {\it Sample-Optimal (SO)-SPRIGHT algorithm} and the {\it Near Sample-Optimal (NSO)-SPRIGHT algorithm} respectively. 

\begin{thm}\label{peeling-decoder-RBI}
Given the offsets $\mathbf{D}\in\GF^{P\times n}$ chosen by 
\begin{itemize}
	\item Definition \ref{def:near-linear_offset} for the near-linear time detection scheme, or 
	\item Definition \ref{def:random_offset} for the NSO-SPRIGHT algorithm and Definition \ref{def:code_offset} for the SO-SPRIGHT algorithm, 
\end{itemize}
the failure probability $\Pf$ of the peeling decoder in the presence of noise is ${O}(1/K)$.
\end{thm}
\begin{proof}
	See Appendix \ref{sec:RBI_perf_analysis}.
\end{proof}

\subsection{Near-linear Time Robust Bin Detection: A Random Design}\label{sec:near-linear-time-design}

The near-linear time bin detection scheme follows the principle of using random codes to resolve the different bin hypotheses and obtain the index-value pair. 

\begin{defi}\label{def:near-linear_offset}
Let $P={O}(\log N)$. The near-linear time detection scheme requires $P$ {\it random offsets} $\{\mathbf{d}_{p}\}_{p\in[P]}$ chosen independently and uniformly at random over $\GF^n$ in every group.
\end{defi}

For some $\gamma\in(0,1)$, the near-linear time detection routine is performed as follows:
\begin{itemize}
	\item {\it zero-ton verification}: for zero-tons, we can expect the energy $\left\|\bdsb{U}\right\|^2$ to be small relative to the energy of a single-ton. Therefore, this idea is used to eliminate zero-tons:
\begin{align}
	\bdsb{U}\sim\mathcal{H}_{\textrm{Z}}, \quad \mathrm{if}~\frac{1}{P}\left\|\bdsb{U}\right\|^2\leq (1+\gamma)\nu^2.
\end{align}
	\item {\it single-ton search}: after ruling out zero-tons and multi-tons, the ultimate goal is to identify single-tons in a certain group $c$ in terms of the underlying index $\mathbf{k}$ and the value $X[\mathbf{k}]$ in that hash set $\{\mathbf{k}:\mathcal{H}_c(\mathbf{k})=\bdsb{j}\}$. Therefore, assuming that the underlying bin $\bdsb{j}$ is a single-ton bin, we perform a single-ton search to estimate the pair of estimates $(\widehat{\mathbf{k}}, \widehat{{X}}[\widehat{\mathbf{k}}])$ for peeling. To do so, we employ a Maximum Likelihood Estimate (MLE) test. For each of $N/B$ possible coefficient locations $\mathbf{k}$ in $\mathbf{M}_c^T\mathbf{k}=\bdsb{j}$, we obtain the single-ton coefficient as
\begin{align}\label{X_MLE}
  	\widehat{\alpha}[\mathbf{k}] = \frac{1}{P} \mathbf{s}_{\mathbf{k}}^T \bdsb{U},\quad  \forall \mathbf{k}~\textrm{such that}~\mathbf{M}_c^T\mathbf{k}=\bdsb{j}.
\end{align}
Using the MLE of the coefficient, we choose among the locations by finding the location $k$ which minimizes the residual energy:
\begin{align}
  \widehat{\mathbf{k}} = \arg\min_{\mathbf{k}}~\left\|\bdsb{U} - \widehat{\alpha}[\mathbf{k}] \mathbf{s}_{\mathbf{k}} \right\|^2.
\end{align}
With the estimated index $\widehat{\mathbf{k}}$, the value of the coefficient is obtained as
\begin{align}\label{x_detection}
	\widehat{{X}}[\widehat{\mathbf{k}}] =
	\begin{cases}
		\rho, &\mathrm{if}~\mathbf{s}_{\mathbf{k}}^T\bdsb{U}/P \geq 0\\					-\rho, &\mathrm{if}~\mathbf{s}_{\mathbf{k}}^T\bdsb{U}/P < 0.
	\end{cases}	
\end{align}

	\item {\it single-ton verification}: this step confirms if the bin is a single-ton via a residual test using the single-ton search estimates
	\begin{align}\label{singleton-verification}
		\frac{1}{P}\left\|\bdsb{U} - \widehat{{X}}[\widehat{\mathbf{k}}] \mathbf{s}_{\widehat{\mathbf{k}}} \right\|^2 \leq (1+\gamma)\nu^2.
	\end{align}	
\end{itemize}

Since there are a total of $\eta K$ bins in each of the $C$ subsampling groups and each bin has $P={O}(\log N)$ measurements, the SPRIGHT framework using the near-linear time detection scheme leads to a sample cost of $M= C \eta K P = {O}(K\log N)$. In terms of complexity, solving the above minimizations requires an exhaustive search over all indices $\mathbf{M}_c^Tk=\bdsb{j}$ for some bin $\bdsb{j}\in\GF^b$. This leads to an exhaustive search over ${O}(N/K)$ elements on average in each peeling iteration, where each element imposes a search complexity of $P={O}(\log N)$ by the generalized likelihood ratio test. As a result, across all ${O}(K)$ peeling iterations, this results in a total complexity of $T={O}(N/K)\times {O}(\log N) \times {O}(K) = {O}(N \log N)$.

\subsection{Sub-linear Time Robust Bin Detection}\label{sec:sublinear-time-design}

Inspired by the near-linear time bin detection scheme, we devise two simple schemes to achieve the same performance with sub-linear time complexity. Recall that the robust bin detection involves three steps:
\begin{itemize}
	\item[1)] zero-ton verification $\frac{1}{P}\left\|\bdsb{U}\right\|^2\leq (1+\gamma)\nu^2$;
	\item[2)] single-ton search that estimates the index-value pair $(\widehat{\mathbf{k}},\widehat{X}[\widehat{\mathbf{k}}])$;
	\item[3)] single-ton verification $\frac{1}{P}\left\|\bdsb{U} - \widehat{{X}}[\widehat{\mathbf{k}}] \mathbf{s}_{\widehat{\mathbf{k}}} \right\|^2 \leq (1+\gamma)\nu^2$. 
\end{itemize}
The near-linear time design is a straightforward construction of the offset matrix $\mathbf{D}\in\GF^{P\times n}$ to guarantee success for step (1) and step (3). However, it does not optimize its choice of offsets to facilitate step (2) in the noisy setting, which causes the high complexity.

To avoid the joint estimation and detection approach in the near-linear time scheme, we use different offsets to tackle them separately. We perform the single-ton search using some offsets, while using other offsets for zero-ton and single-ton verifications. Since the fully random offsets already tackle the verifications with high probability, we can simply focus on designing offsets for the single-ton search. If the single-ton search can be performed with high probability of success using the same amount of samples and computations (in an order-sense), the entire bin detection scheme becomes sub-linear, as discussed in details below.

\begin{prop}\label{prop_BSC}
Given a single-ton bin with $(\mathbf{k},X[\mathbf{k}])$ observed in noise
\begin{align}
	U_p 
	&= X[\mathbf{k}](-1)^{\ip{\mathbf{d}_p}{\mathbf{k}}} + W_p,\quad p \in [P],
\end{align}
the sign of each observation satisfies
\begin{align}
	\sgn{U_p}
	&= \ip{\mathbf{d}_p}{\mathbf{k}}\oplus\sgn{X[\mathbf{k}]}\oplus Z_p,\quad p \in [P],
\end{align}
where $Z_p$ is a Bernoulli random variable with probability upper bounded as $\Pe=e^{-\frac{\eta}{2} \mathrm{SNR}}$.
\end{prop}
\begin{proof}
	See Appendix \ref{proof_prop_BSC}.
\end{proof}

\begin{wrapfigure}{r}{0.55\textwidth}
\vspace{-0.7cm}
\begin{minipage}{0.55\textwidth}
\begin{center}
\includegraphics[scale=.31]{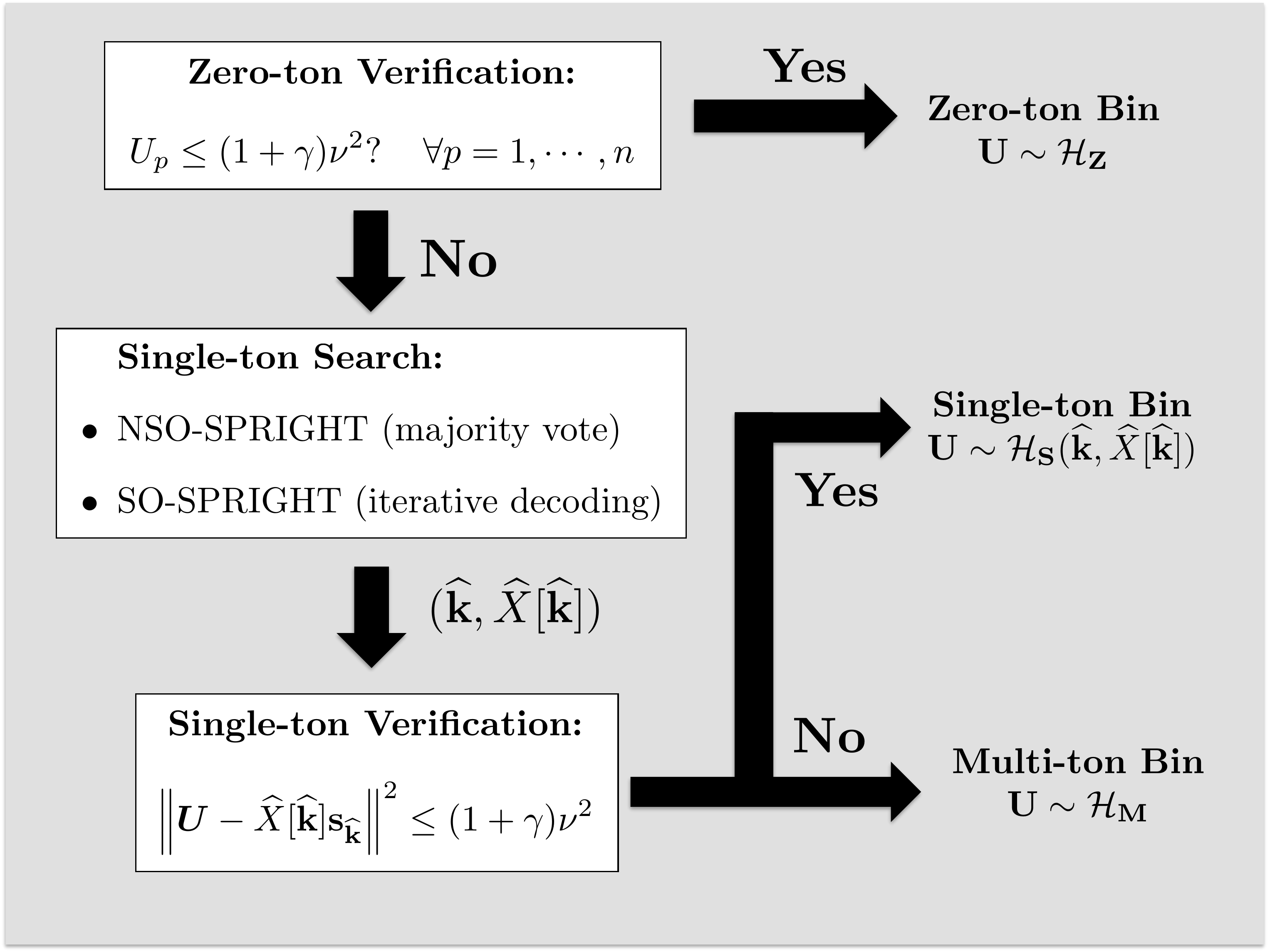}
\caption{A simplified flowchart of the bin detection routine $\psi$ for the noisy setting by choosing offsets according to Theorem \ref{peeling-decoder-RBI} for the NSO-SPRIGHT and the SO-SPRIGHT algorithm.}
\vspace{-1cm}
\label{fig:bindetectionroutine_noisy}
\end{center}
\end{minipage}
\end{wrapfigure}
From Proposition \ref{prop_BSC}, it can be seen that the sign vector of the bin observation vector $\bdsb{U}$ can be viewed as some potentially corrupted bits received over a binary symmetric channel (BSC). The design of the offset matrix $\mathbf{D}$ for reliable and fast decoding over the BSC is thus the key to achieving sub-linear complexity. 

In the following, we first present the sub-linear time {\it NSO-SPRIGHT Algorithm} that is easy to implement (i.e. a majority vote) and achieves a sub-linear complexity $T={O}(K\log^3 N)$ with a sample cost of $M={O}(K\log^2 N)$. Then, we present the sub-linear time {\it SO-SPRIGHT Algorithm} that maintains the optimal sample cost $M={O}(K\log N)$ and simultaneously achieves sub-linear complexity $T={O}(K\log N)$ using an iterative channel decoder.

\subsubsection{The NSO-SPRIGHT Algorithm}

Recall that the near-linear time design requires an exhaustive search due to the lack of structure of fully random offsets, which creates a bottleneck of the complexity. The key is to design a set of offsets that constitute a sufficiently good codebook to allow reliable transmissions of the $n$-bit index $\mathbf{k}$ over a BSC. In order to enable the bit-by-bit recovery of the binary representation of $\mathbf{k}$ as in the noiseless design, the first coding strategy we exploit is {\it repetition coding}, which is done by imposing structures on the random offsets for subsampling.

\begin{defi}\label{def:random_offset}
Let $P=P_1P_2$ with $P_1={O}(n)$ and $P_2= n$. The {\it NSO-SPRIGHT} algorithm requires $P_1$ random offsets $\{\mathbf{d}_{p}\}_{p\in[P_1]}$ chosen independently and uniformly over $\GF^n$ and $P_2$ {\it modulated offsets} $\{\mathbf{d}_{p,q}\}_{q\in[P_2]}$ such that
\begin{align}
	\mathbf{d}_{p,q}\oplus \mathbf{d}_{p} = \mathbf{e}_q,\quad q\in[P_2]
\end{align}
where $\mathbf{e}_q$ is the $q$-th column of the identity matrix.
\end{defi}

Given the offsets chosen as Definition \ref{def:random_offset}, we can identify the $q$-th bit of $\mathbf{k}$ by jointly considering $P_1$ observations associated with offsets $\mathbf{d}_{p,q}$ across $p\in[P_1]$. More specifically, 
\begin{align}
	\sgn{U_{p,q}}
	&= \ip{\mathbf{d}_{p,q}}{\mathbf{k}}\oplus\sgn{X[\mathbf{k}]}\oplus Z_{p,q}\\		
	\sgn{U_p}
	&= \ip{\mathbf{d}_p}{\mathbf{k}}\oplus\sgn{X[\mathbf{k}]}\oplus Z_p.
\end{align}
Since $\mathbf{d}_{p,q}\oplus\mathbf{d}_p=\mathbf{e}_q$, we have $P_1$ corrupted versions of $k[q]$:
\begin{align}\label{NSO_prob}
	\sgn{U_{p,q}}\oplus \sgn{U_p} = \ip{\mathbf{e}_q}{\mathbf{k}} \oplus Z_{p,q}' = k[q] \oplus Z_{p,q}',
\end{align}	
where $Z_{p,q}' = Z_p \oplus Z_{p,q}$ is another Bernoulli variable with $\theta = \Prob{Z_{p,q}' = 1}=2\Pe(1-\Pe)<1/2$.
%\begin{align}
%	\theta 
%	&= \Prob{Z_{p,q}=1|Z_p=0}\Prob{Z_p=0} + \Prob{Z_{p,q}=0|Z_p=1}\Prob{Z_p=1}\\
%	&= 2\Pe(1-\Pe).
%\end{align}
Then the MLE of $k[q]$ given observations $\{\sgn{U_{p,q}}\oplus \sgn{U_p}\}_{p=1}^{P_1}$ can be obtained as
\begin{align}
	\widehat{k}[q]
	= \arg\max_{a} \prod_{p=1}^{P_1} \theta^{\sgn{U_{p,q}}\oplus \sgn{U_p}\oplus a}(1-\theta)^{1-\sgn{U_{p,q}}\oplus \sgn{U_p}\oplus a}.
\end{align}
Using the fact that $\theta < 1/2$ such that $\log(\theta/1-\theta) < 0$, we can simplify the objective as
\begin{align}\label{MLE_kq}
	\widehat{k}[q]
	&= \arg\min_{a\in\GF}  \sum_{p=1}^{P_1} \sgn{U_{p,q}}\oplus \sgn{U_p}\oplus a.
\end{align}
In other words, the decoding scheme for the $q$-th bit of the index $\mathbf{k}$ becomes a simple {\it majority test} by accumulating $P_1={O}(n)$ random signs $\sgn{U_{p,q}}\oplus \sgn{U_p}$. Using the estimated bits $\{\widehat{k}[q]\}_{q\in[P_2]}$ together with $\mathbf{M}_c^T\mathbf{k}=\bdsb{j}$, the estimate $\widehat{\mathbf{k}}$ can be obtained accordingly. Finally, the value of the coefficient is obtained as \eqref{x_detection}. The zero-ton and single-ton verifications can be performed directly using the measurements associated with offsets $\mathbf{d}_p$ since there are $P_1={O}(n)={O}(\log N)$ such random offsets, which have been shown to achieve high probability of success in the near-linear time design.

From Definition \ref{def:random_offset}, we can see that there are a total of $P_1P_2 = {O}(n^2)$ offsets, and therefore each bin has ${O}(\log^2N)$ observations. As a result, the NSO-SPRIGHT algorithm leads to a sample cost of $M= C \eta K \log^2N = {O}(K\log^2 N)$. In terms of complexity, the majority vote requires ${O}(\log^2 N)$ operations for each bin, contributing to a total of ${O}(K\log^2N)$ operations across all ${O}(K)$ bins. However, this complexity is dominated by generating $P=P_1P_2=\log^2N$ basic observation sets from $B$-point WHTs, each imposing an extra complexity of ${O}(K\log K) = {O}(K\log N)$ because of $K={O}(N^\delta)$. As a result, this gives a total complexity of $T= {O}(K\log^3 N)$.

\subsubsection{The SO-SPRIGHT Algorithm}
While the NSO-SPRIGHT algorithm exploits repetition codes induced by the random offsets to robustify the noisy performance, we can further use better error correction codes to guide the choice of offsets. This is slightly more difficult to implement in practice since the decoding requires channel decoder instead of a simple majority vote, but the resulting sample complexity and computational complexity are order-optimal.

\begin{defi}\label{def:code_offset}
Let $P=\sum_{i=1}^3 P_i$ with $P_i={O}(n)$ for $i=1,2,3$. The {\it SO-SPRIGHT} algorithm requires $P_1$ {\it random offsets} $\mathbf{d}_p$ for $p=1, \cdots,P_1$ chosen independently and uniformly at random over $\GF^n$, and $P_2$ {\it zero offsets} $\mathbf{d}_p=\mathbf{0}$ for $p=P_1+1,\cdots,P_1+P_2$, and finally $P_3$ {\it coded offsets} $\mathbf{d}_p$ for $p=P_1+P_2+1,\cdots,P$ such that the offset matrix $\mathbf{G}=[\cdots;\mathbf{d}_p;\cdots;]\in \GF^{P_3\times n}$ constitutes a generator matrix of some linear block code with a minimum distance $\beta P_3$ with $\beta > \Pe$.
\end{defi}

Recall Proposition \ref{prop_BSC}, the observations associated with the coded offsets $\mathbf{G}$ can be written as
\begin{align}
	\begin{bmatrix}
		\sgn{U_{P_1+P_2+1}}\\
		\vdots\\
		\sgn{U_{P}}
	\end{bmatrix}
	&=
	\mathbf{G}\mathbf{k}
	\oplus
	\sgn{X[\mathbf{k}]}
	\oplus
	\begin{bmatrix}
		Z_{P_1+P_2+1}\\
		\vdots\\
		Z_{P}
	\end{bmatrix}.
\end{align}
Note that there is a nuisance sign $\sgn{X[\mathbf{k}]}$ which is unknown to the robust bin detector. To illustrate our scheme, we first assume that there is a genie that informs the decoder of the sign of the coefficient $\sgn{X[\mathbf{k}]}$, and then we discuss how to get rid of the genie.
\begin{itemize}
	\item {\it when $\sgn{X[\mathbf{k}]}$ is known a priori}: in this case, we can easily obtain
\begin{align}
	\begin{bmatrix}
		\sgn{U_{P_1+P_2+1}}\oplus\sgn{X[\mathbf{k}]}\\
		\vdots\\
		\sgn{U_{P}}\oplus\sgn{X[\mathbf{k}]}
	\end{bmatrix}
	=
	\mathbf{G}
	\mathbf{k}
	\oplus
	\begin{bmatrix}
		Z_{P_1+P_2+1}\\
		\vdots\\
		Z_{P}
	\end{bmatrix}.
\end{align}
	Since there are $n$ information bits in the index $\mathbf{k}$, then there exists some channel code (i.e. $\mathbf{G}$) with block length $P_3= n / R(\beta)$ that achieves a minimum distance of $\beta P_3$, where $R(\beta)$ is the rate of the code. As long as $\beta>\Pe$, it is obvious that the unknown $\mathbf{k}$ can be decoded with exponentially decaying probability of error. There exist many codes that satisfy the minimum distance properties, but the concern is the decoding time. It is desirable to have decoding time linear in the block length so that the sample complexity and computational complexity can be maintained at ${O}(n)$, same as the noiseless case. Excellent examples include the class of expander codes or LDPC codes that allow for linear time decoding.  
	\item {\it when $\sgn{X[\mathbf{k}]}$ is not known a priori}: we consider the observations associated with all the zero offsets $\mathbf{d}_p=\mathbf{0}$ for $p=P_1+1,\cdots,P_1+P_2$ 
\begin{align}
	\begin{bmatrix}
		\sgn{U_{P_1+1}}\\
		\vdots\\
		\sgn{U_{P_1+P_2}}
	\end{bmatrix}
	&=
	\sgn{X[\mathbf{k}]}	
	\oplus
	\begin{bmatrix}
		Z_{P_1+1}\\
		\vdots\\
		Z_{P_1+P_2}
	\end{bmatrix}
\end{align}	
which can recover the sign correctly $\sgnhat{X[\mathbf{k}]}=\sgn{X[\mathbf{k}]}$ with high probability using a majority test (assuming $\Pe\leq 1/2$). If $\Pe > 1/2$, the sign is obtained accordingly using a minority test. Then we can proceed as if the sign is known a priori:
\begin{align}
	\begin{bmatrix}
		\sgn{U_{P_1+P_2+1}}\oplus \sgnhat{X[\mathbf{k}]} \\
		\vdots\\
		\sgn{U_{P}}\oplus \sgnhat{X[\mathbf{k}]} 
	\end{bmatrix}
	=
	\mathbf{G}
	\mathbf{k}
	\oplus
	\begin{bmatrix}
		Z_{P_1+P_2+1}\\
		\vdots\\
		Z_{P}
	\end{bmatrix}.
\end{align}	
\end{itemize}
Finally, the value of the coefficient is obtained as \eqref{x_detection}. The zero-ton and single-ton verifications can be performed directly using the observations associated with offsets $\mathbf{d}_p$. Since there are $P_1={O}(n)={O}(\log N)$ such random offsets, which have been shown to achieve high probability of success in the near-linear time design.

Using the SO-SPRIGHT design, we can see that there are three sets of offsets, where one set includes $P_3={O}(n)$ offsets for the single-ton search, and the second set includes $P_2={O}(n)$ zero offsets for the sign reference, and $P_1={O}(n)$ random offsets for the zero-ton and single-ton verifications. Therefore, we have a total of $P = \sum_{i=1}^3 P_i = {O}(n) = {O}(\log N) $ offsets and each bin has ${O}(\log N)$ observations. As a result, the SO-SPRIGHT algorithm leads to a sample cost of $M= C \eta K P = {O}(K\log N)$, which is the same as the noiseless case \cite{scheibler2013fast}. In terms of complexity, if $\mathbf{G}$ is a properly chosen channel code generator matrix from the class of expander codes or LPDC codes, the decoding time for the index requires ${O}(n)={O}(\log N)$ operations for each bin. This contributes to a total of ${O}(K\log N)$ complexity across all ${O}(K)$ bins. However, this complexity is dominated by subsampling for generating $P$ basic observation sets from $B$-point WHTs, each imposing an extra complexity of ${O}(K\log K) = {O}(K\log N)$ because of $K={O}(N^\delta)$. As a result, this gives a total complexity of $T= {O}(K\log^2 N)$, which is also the same as the noiseless case \cite{scheibler2013fast}.

\section{Applications}\label{sec:applications}
In the following, we provide some machine learning concepts that can be cast as a WHT computation or expansion.

\begin{example}[Pseudo-Boolean Function and Sparse Polynomial]
An arbitrary pseudo-Boolean function can be represented uniquely by a multi-linear polynomial over the hypercube $(z_1,\cdots,z_n)\in\{-1,+1\}^n$:
\begin{align}\label{eq:fourier}
	f(z_1,\cdots,z_n) = \sum_{\mathcal{S}\subseteq [n]} \alpha_{\mathcal{S}} \prod_{i\in\mathcal{S}} z_i,~ \forall~ z_i \in \{-1,+1\}, 
\end{align}
where $\mathcal{S}$ is a subset of $[n]\defn \{1,\cdots,n\}$,  and $\alpha_{\mathcal{S}}$ is the Walsh (Fourier) coefficient associated with the monomial $\prod_{i\in\mathcal{S}} z_i$. If we replace $z_i$ by $(-1)^{m[i]}$ such that $z_i=-1$ when $m[i]=1$ and $z_i=1$ when $m[i]=0$, we have $x[\mathbf{m}] = f\left((-1)^{m[1]},\cdots,(-1)^{m[n]}\right)$ for
$\mathbf{m}\in\GF^n$ and $X[\mathbf{k}]=\sqrt{N}\alpha_{\mathcal{S}}$ such that $\supp{\mathbf{k}}=\mathcal{S}$. 
\end{example}
\begin{example}[Set Functions]
A set function is an arbitrary real-valued function $f: 2^{[n]}\rightarrow \mathbb{R}$ defined for every element in the power set $\mathcal{Z}\in2^{[n]}$, which has a Walsh expansion given by
\begin{align}\label{set_function}
	f(\mathcal{Z}) = \frac{1}{\sqrt{N}} \sum_{\mathcal{S}\in2^{[n]}} \hat{f}(\mathcal{S}) (-1)^{|\mathcal{S}\cap\mathcal{Z}|},
\end{align}
where $\hat{f}(\mathcal{S})$ is the Walsh (Fourier) coefficient. Clearly, a set function can also be viewed as a $n$-ary pseudo-Boolean function in \eqref{eq:fourier} such that $f(\mathcal{Z})=f(z_1,\cdots,z_n)$ as long as $z_i=-1$ if $i\in\mathcal{Z}$ and $z_i=1$ if $i\notin\mathcal{Z}$. Therefore, each function value $f(\mathcal{Z})$ can be regarded as a sample $x[\mathbf{m}]  = f(\supp{\mathbf{m}})$, where the Walsh coefficient satisfies $X[\mathbf{k}]=\hat{f}(\mathcal{S})$ as long as $\supp{\mathbf{k}}=\mathcal{S}$.

\end{example}
\begin{example}[Decision Tree Learning] 
Decision trees are machine-learning methods for constructing prediction models from data, whose goal is to predict the value of a target label $f$ based on $n$ input variables $z_i\in\{\pm 1\}$ for $i\in[n]$. More specifically, this includes classification trees (discrete-valued outcome $f\in\mathbb{Z}$) and regression trees (real-valued outcome $f\in\mathbb{R}$). Decision tree models are usually constructed from top-down starting at the root node, by choosing a certain variable $z_i$ for some $i$ at each step that optimally splits the set of training data with respect to some measure of goodness. Hence, for each set of input variables $(z_1,\cdots,z_n)\in\{-1,+1\}^n$, there is a unique leaf node in the tree that assigns the target label $f$. This is mathematically equivalent to learning a (pseudo)-Boolean function, which can be cast as a problem of computing the WHT of $f$.
\end{example}

It has been found that many instances of the examples above exhibit sparsity in the Walsh spectrum. In general, our SPRIGHT framework can be applied to learning $K$-sparse pseudo-Boolean polynomials $f: \{\pm 1\}^n \rightarrow \mathbb{R}$ with $n$ variables. A concrete example is in decision tree learning, where the underlying (pseudo)-Boolean function has a sparse spectrum if the decision tree has few leaf nodes with short depth. An extreme case would be when the underlying function only depends on few input variables, which is also referred to as the juntas problem in Boolean analysis\footnote{It is well-known that learning juntas using {\it random} samples is NP-hard. Our framework tackles the juntas problem using {\it specifically chosen} samples, and hence we can achieve sub-linear sample cost and run-time. This is not a contradiction.}. Therefore, if the $K$-sparse $N$-point WHT can be computed efficiently, these machine learning applications can benefit greatly from the reductions in both the sample complexity and computational complexity. In the following, we present a specific machine learning application in graph sketching.

\subsection{Applications in Hypergraph Sketching}
A \emph{hypergraph}, denoted by $\mathcal{G}=(\mathcal{V},\mathcal{E})$, is a generalized notion of graphs where each edge $e\in\mathcal{E}$, called the \emph{hyperedges}, can connect more than two nodes in the node set $\mathcal{V}$. Hypergraph sketching here refers to the procedure of identifying the unknown hypergraph structure from cut queries. Hypergraphs have been very useful in relational learning, which has received extensive attentions in recent years since many real-world data are organized by the relations between entities. Some of the interesting problems involved in relational learning include the discovery of communities, classification, and predictions of possible new relations.

We describe the hypergraph sketching application through an example depicted in \figref{fig:graph_sketching_example}. Consider a scenario where there are $n$ books from a certain provider (e.g. Amazon) and each book is characterized by a node in the graph. There are numerous transactions taking place in which each customer buys a few books. In this setting, the relationship between books in each transaction can be captured by a hyperedge, which connects the subset of books bought in the same transaction. A cartoon illustration is depicted in \figref{fig:graph_sketching_example}, where there are $3$ distinct sets of books bought in different transactions with each set coded in different colors. Then, the hypergraph sketching problem is equivalent to solving the following problem under a the following {\it query} model:
\begin{itemize}
	\item Pick an arbitrary partition $(S,\bar{S})$ of $n$ books such that $S\cup\bar{S}=\mathcal{V}$ (see \figref{arbitrary_cut}). 
	\item One can query the following: i) are there any transactions that include books from both sets $(S,\bar{S})$? and ii) if there are, what is the total number of transactions that satisfy this requirement? For example in \figref{arbitrary_cut_value}, the resulting query would return $1$ since there is only $1$ transaction that includes books from both sets. 
	\item How many such queries are needed to fully learn all the unknown distinct subsets of books that are bought in different transactions?
\end{itemize}
Note that the query requested here is in fact the number of hyperedges that cross over the two sets $(S,\bar{S})$, which is defined as the {\it cut value} of the graph. As shown next, this can be mathematically established as a sparse WHT computation problem, where our SPRIGHT framework is found to be useful.

\begin{figure}[h]
\begin{center}
\subfigure[Hidden graph of $n$ books: there are a few purchase patterns, where each corresponds to a hyperedge]{
\includegraphics[width=0.3\linewidth]{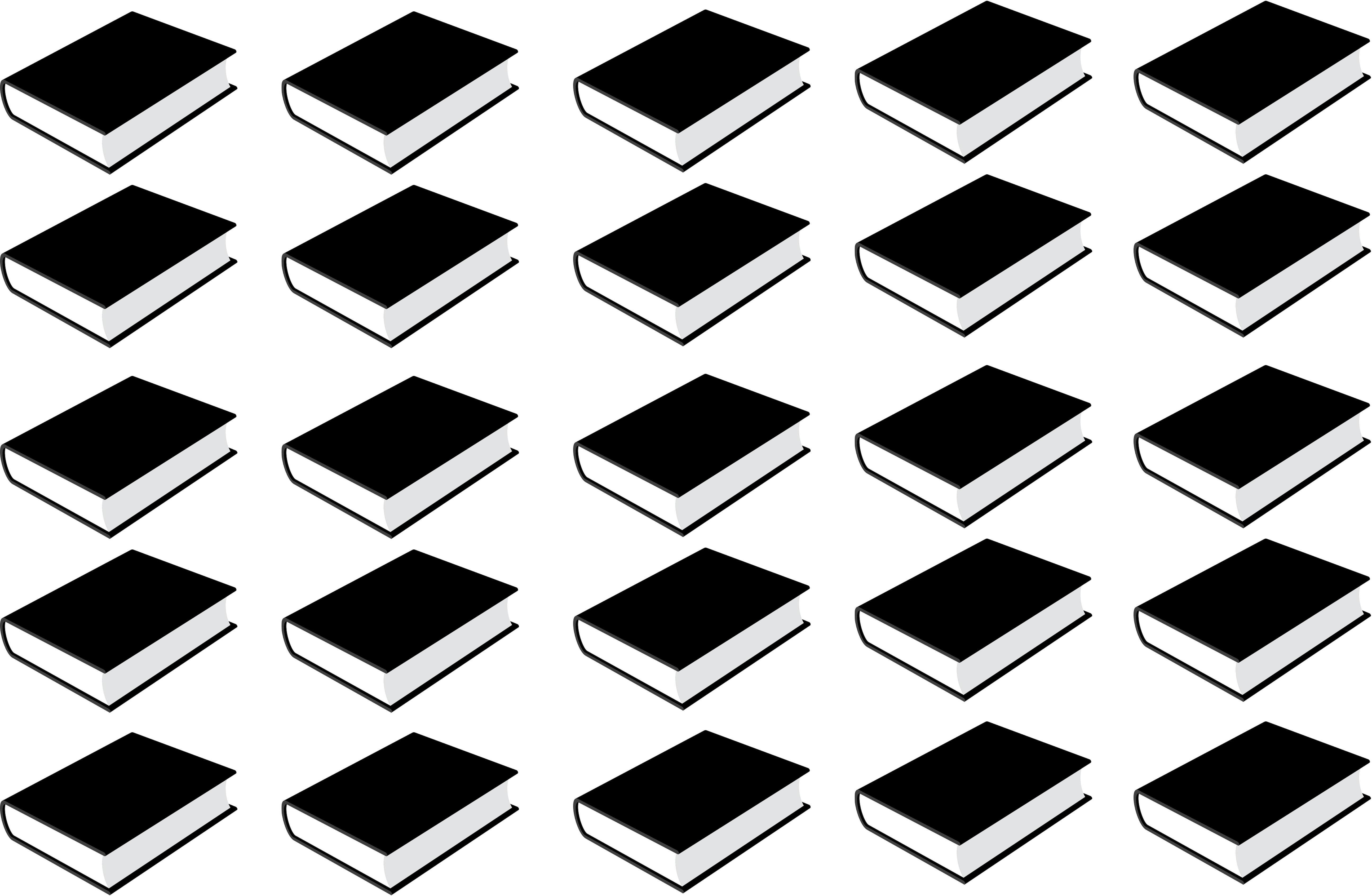}
}
~~~
\subfigure[Pick some partition $(S, \bar{S})$: how many transactions include books from both sets $(S,\bar{S})$?]{
\includegraphics[width=0.3\linewidth]{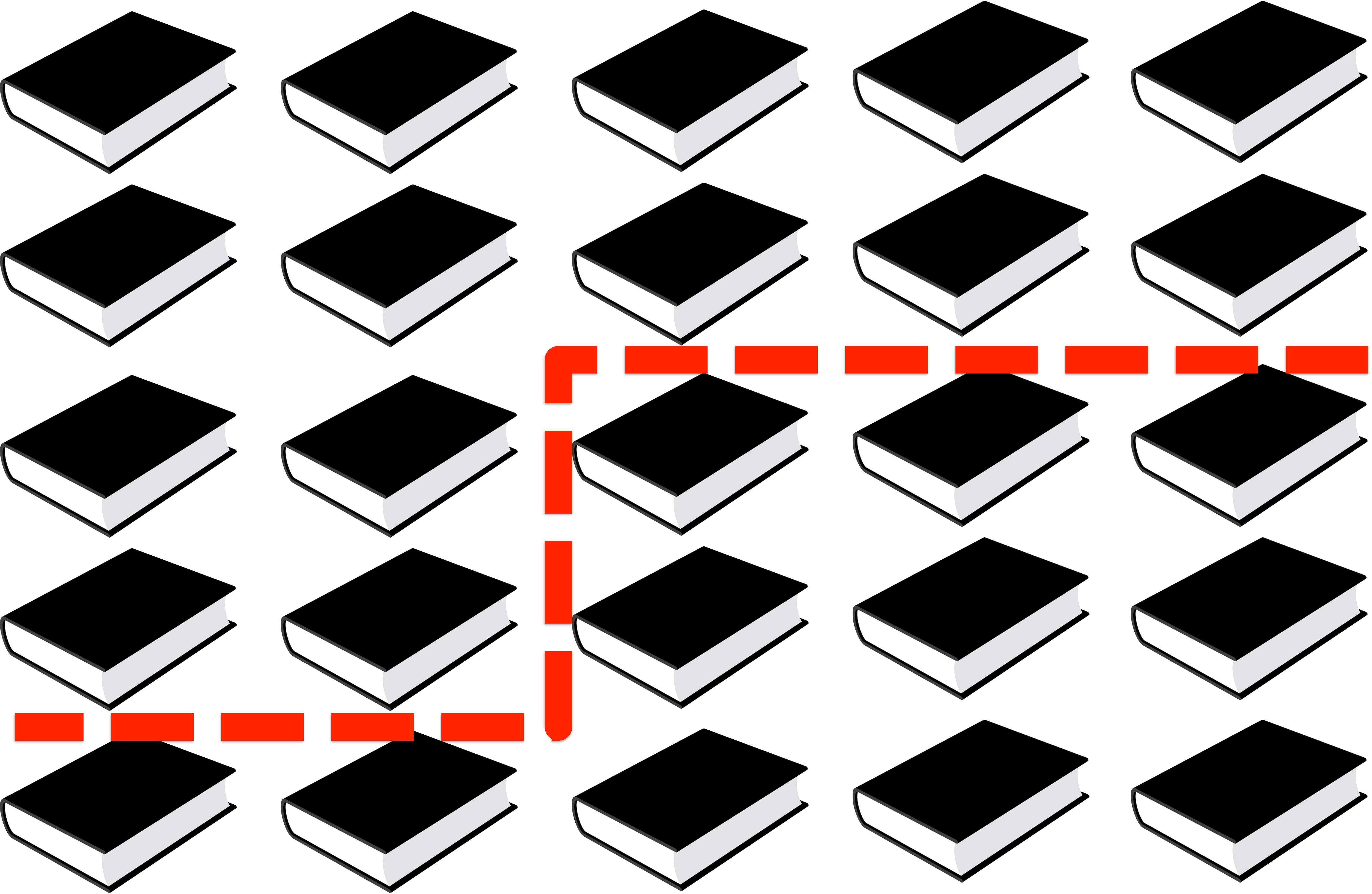}
\label{arbitrary_cut}
}
~~~
\subfigure[Query: in this example, the query result for this partition is $1$ and the graph has $3$ distinct subsets.]{
\includegraphics[width=0.3\linewidth]{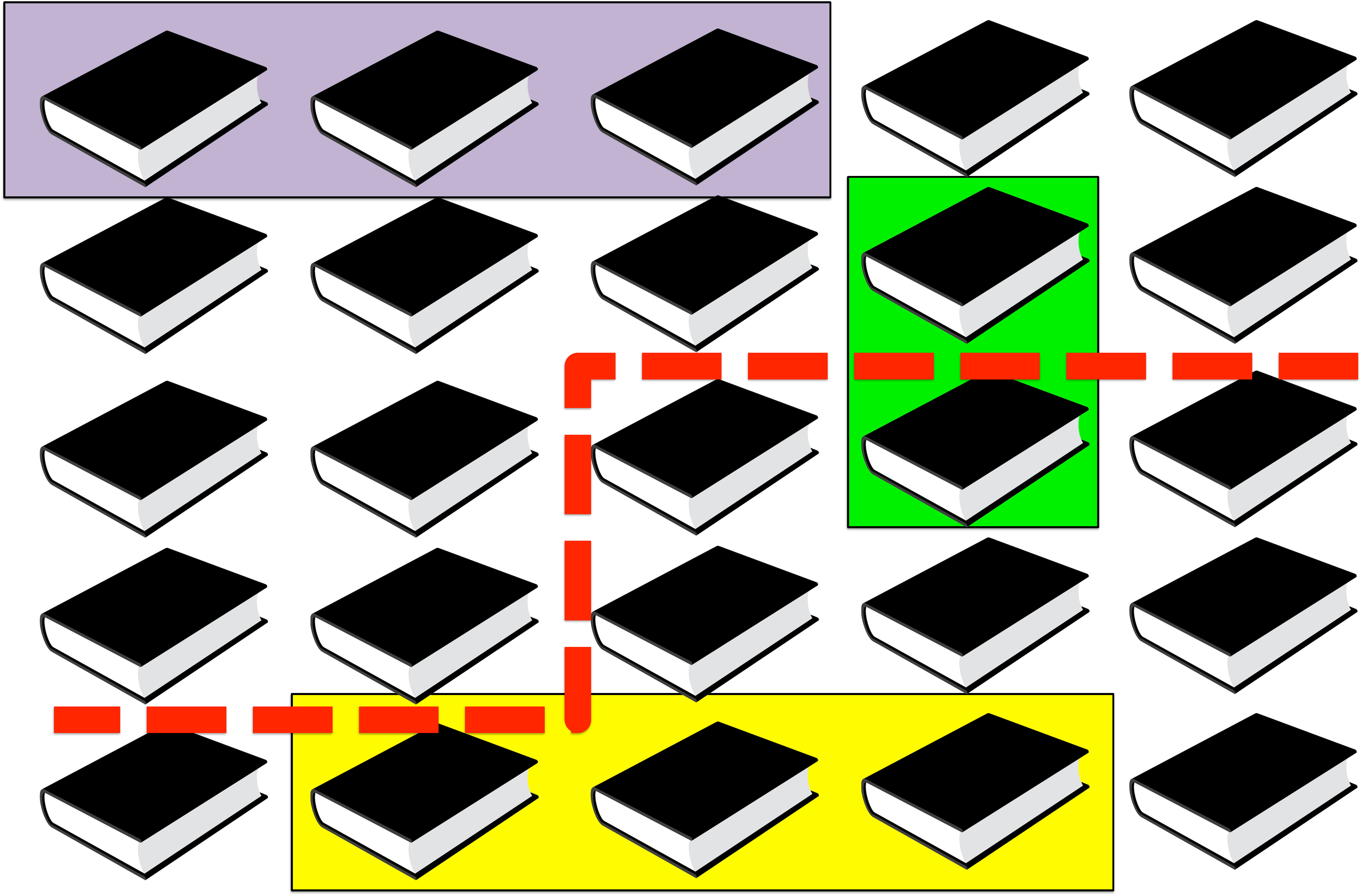}
\label{arbitrary_cut_value}
}
\end{center}
\vspace{-0.7cm}
\caption{Given a set of $n$ books, infer the graph structure by querying graph cuts.}\label{fig:graph_sketching_example}
\end{figure}

Let $| \mathcal{V} | = n$ and $|\mathcal{E}|=s$. A cut $\mathcal{S} \subseteq \mathcal{V}$ is a set of selected vertices, denoted by the binary $n$-tuple $\mathbf{m}=[m[1],\cdots,m[n]]$ over $\GF^n$, where $m[i]=1$ if $i\in\mathcal{S}$ and $m[i]=0$ if $i\notin\mathcal{S}$. The cut value $x[\mathbf{m}]$ for a specific cut $\mathbf{m}$ in the hypergraph is defined as $x[\mathbf{m}] =\left |\{e \in \mathcal{E}: e \cap \mathcal{S} \neq \varnothing,~e \cap \bar{S} \neq \varnothing \} \right |$, where $\bar{S}=\mathcal{V}/\mathcal{S}$. In other words, the cut value corresponds to the number of hyperedges that crosses between the two sets $(S,\bar{S})$. Given a partition $\mathbf{m}\in \GF^n$, for some edge $e\in \mathcal{E}$, we define the following function to indicate whether it crosses over two sets $(S,\bar{S})$:
\begin{align}\label{1e_m}
	1_e[\mathbf{m}] = \prod_{i\in e} \frac{\left(1+(-1)^{m[i]}\right)}{2}+\prod_{i\in e} \frac{\left(1-(-1)^{m[i]}\right)}{2}.
\end{align}
For example, if all the nodes connected through this particular hyperedge $i\in e$ is on the same side of the partition $(S,\bar{S})$, which implies that either $m[i]=0$ or $m[i]=1$ for all $i\in e$, this indicator $1_e[\mathbf{m}]=1$ is $1$. This suggests that when the edge $e$ does {\it not} cross over the two sets $(S,\bar{S})$, the indicator takes the value $1$. Therefore, the total count of edges that do cross over can be obtained accordingly as 
\begin{align}
	x[\mathbf{m}]  =\sum_{e\in\mathcal{E}} \left(1-1_e[\mathbf{m}]\right).
\end{align}
By substituting $1_e[\mathbf{m}]$ with \eqref{1e_m}, it can be equivalently written as a WHT expansion as follows:
\begin{align}
	x[\mathbf{m}]
	&=\sum_{\mathbf{k}\in\GF^n}X[\mathbf{k}] (-1)^{\ip{\mathbf{k}}{\mathbf{m}}},
\end{align}
where the coefficient $X[\mathbf{k}]$ is a scaled WHT coefficient such that $X[\mathbf{0}] = \left(s - \sum_{e\in\mathcal{E}}\frac{1}{2^{|e|-1}}\right)$ and
\begin{align}
	X[\mathbf{k}] &=  
	\begin{cases}
		\frac{1}{2^{|e|-1}}, &\textrm{if}~\supp{\mathbf{k}}\in e~\textrm{and}~\left|\supp{\mathbf{k}}\right|~\textrm{is even}\\
		0, & \textrm{otherwise}
	\end{cases}
\end{align}
Clearly, if the number of hyperedges is small $s\ll 2^n$ and the maximum size of each hyperedge is small, the coefficients $X[\mathbf{k}]$'s are sparse. For example, if the hyperedge size can be universally bounded by $d$, the sparsity can be well upper bounded by $K\leq s2^{d-1}$. 

\subsection{Simple Experiment}
Here we consider the noiseless scenario as a proof of concept, we use our SO-SPRIGHT algorithm for hypergraph sketching, which requires ${O}(Kn)$ queries for interpolating the total $2^n$ cut values with run-time ${O}(Kn^2)$. In this experiment, we randomly generate hypergraphs with $n=50$ to $400$ nodes with $s=3,6,9$ edges, where each edge does not connect more than $d=6$ nodes. As can be seen, our SPRIGHT framework computes the sparse coefficients $X[\mathbf{k}]$ in time ${\Theta}(K \log K n) = {\Theta}(Kn^2)$ from only ${\Theta}(Kn)$ cut queries.

\begin{figure}[h]
\begin{center}
\subfigure[Query cost scaling with the graph size $n$]{
\includegraphics[width=0.48\linewidth]{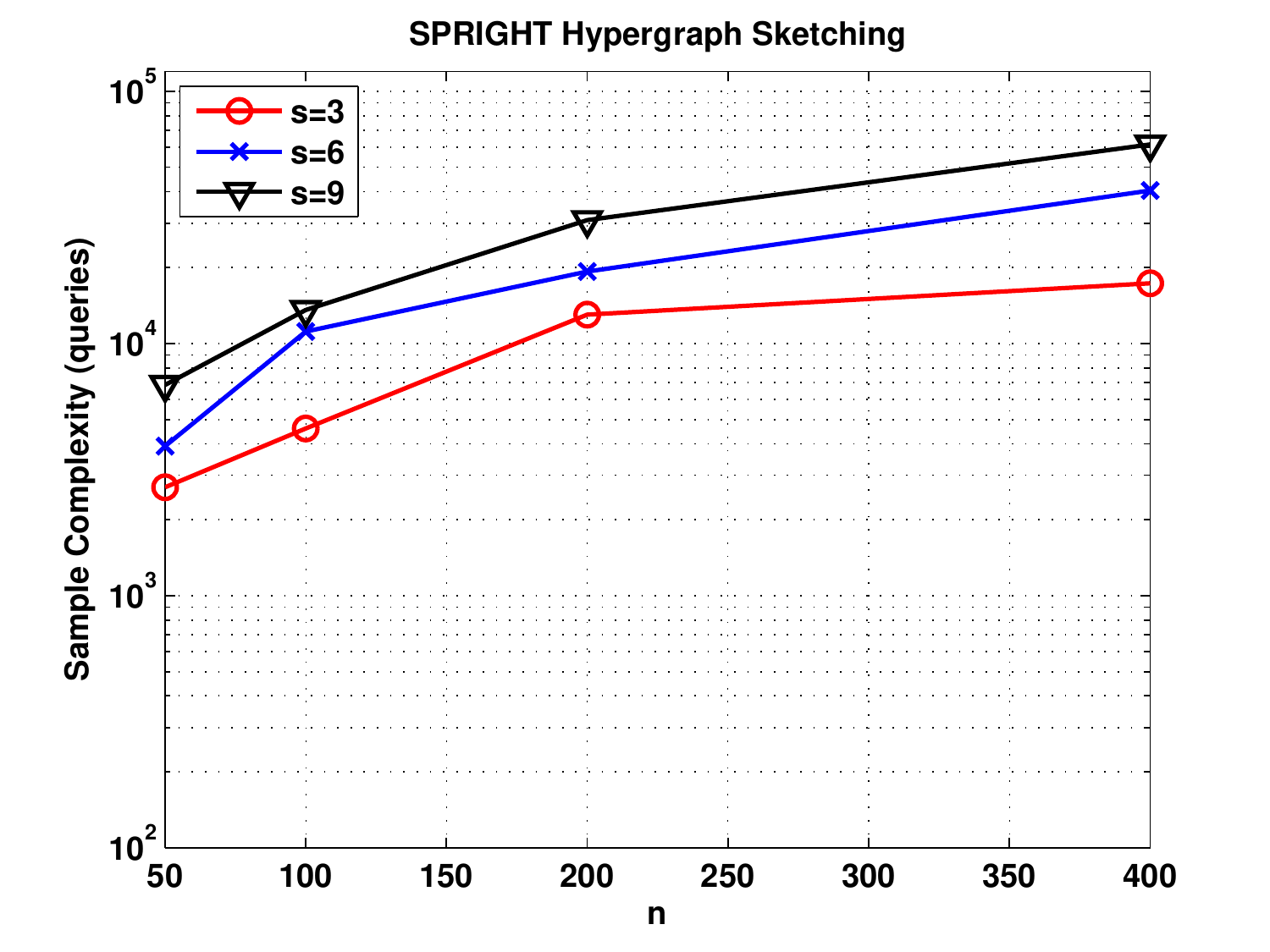}
\label{fig:hypergraph_sketching_meas}
}
~
\subfigure[Run-time scaling with the graph size $n$]{
\includegraphics[width=0.48\linewidth]{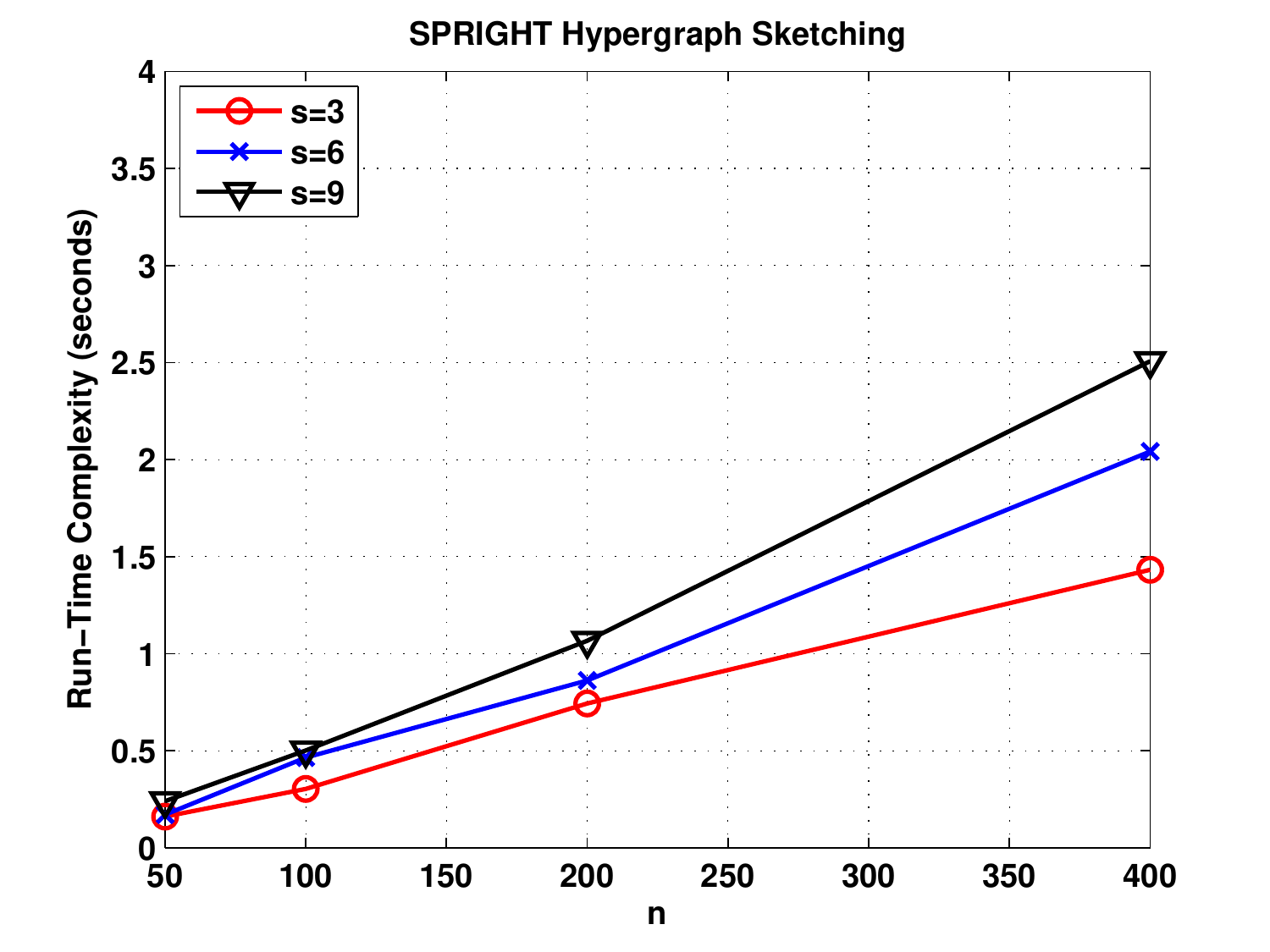}
\label{fig:hypergraph_sketching_runtime}
}
\end{center}

\end{figure}

\section{Numerical Experiments}\label{sec:simulations}

In this section, we test the NSO-SPRIGHT algorithm and SO-SPRIGHT algorithm respectively. We first showcase the performances in many settings by varying the signal length $N=2^n$, sparsity and SNR. Then, we demonstrate possible applications of our SPRIGHT framework in machine learning domains such as hypergraph sketching and decision tree learning over large datasets.

\subsection{Performance of the SPRIGHT Framework}
Here, we synthetically generate time domains samples $\mathbf{x}$ from a $K$-sparse WHT signal $\mathbf{{X}}$ of length $N=2^n$ with $K$ randomly positioned non-zero coefficients of magnitude $\pm\rho$.  The setup of our experiments is given below: 
\begin{itemize}
	\item {\it subsampling parameters}: we fix the number of groups to $C=3$ and the number of bins in each group is $B=2^b$ where $b=\lceil \log_2(K) \rceil$. Note that in this case $B\approx K$ and thus $\eta \approx 1$.
	\vspace{-0.2cm}
	\item {\it NSO-SPRIGHT algorithm parameters}: we choose $P_1=2n$ random offsets and $P_2=n$ modulated offsets. Thus the sample cost is $M_{\rm NSO} = 2CBn^2 \approx 6Kn^2$ and the complexity is $T_{\rm NSO}={O}(Kn^3)$.
	\vspace{-0.2cm}
	\item {\it SO-SPRIGHT algorithm parameters}: we choose $P_1=2n$ coded offsets for the single-ton search, $P_2=n$ zero offsets and $P_3=n$ random offsets for the zero-ton and single-ton verifications. For the single-ton search, the $P_1=2n$ coded offsets are chosen to induce a $(3,6)$-regular LDPC code, where the search utilizes the Gallager's bit flipping algorithm for decoding, which imposes linear run-time ${O}(n)$. The sample cost is $M_{\rm SO} = 4CBn \approx 12Kn$ and the complexity is $T_{\rm SO}={O}(Kn^2)$.
\end{itemize}

\subsubsection{Noise Robustness}
In this subsection, we compare the noise robustness of the NSO-SPRIGHT and SO-SPRIGHT algorithms. The experiment settings are given below:
\begin{itemize}
	\item {\it input profile}: we generate a sparse WHT vector or length $N=2^n$ with $n=14$ and $K = 10, 20, 40$ non-zero coefficients respectively. Therefore, the signal dimension is $N=16384$. The non-zero WHT coefficients are chosen with uniformly random support and random amplitudes $\{\pm 1\}$. The input signal samples $\mathbf{x}$ is obtained by taking the inverse WHT of the sparse WHT vector and adding i.i.d. Gaussian noise samples with variance $\sigma^2$ determined by the range of $\mathsf{SNR}=[-5:5:20]$ dB .	
\end{itemize}

\begin{figure}[h]
\vspace{-0.5cm}
\begin{center}
\subfigure[Probability of success versus SNR]{
\includegraphics[width=0.48\linewidth]{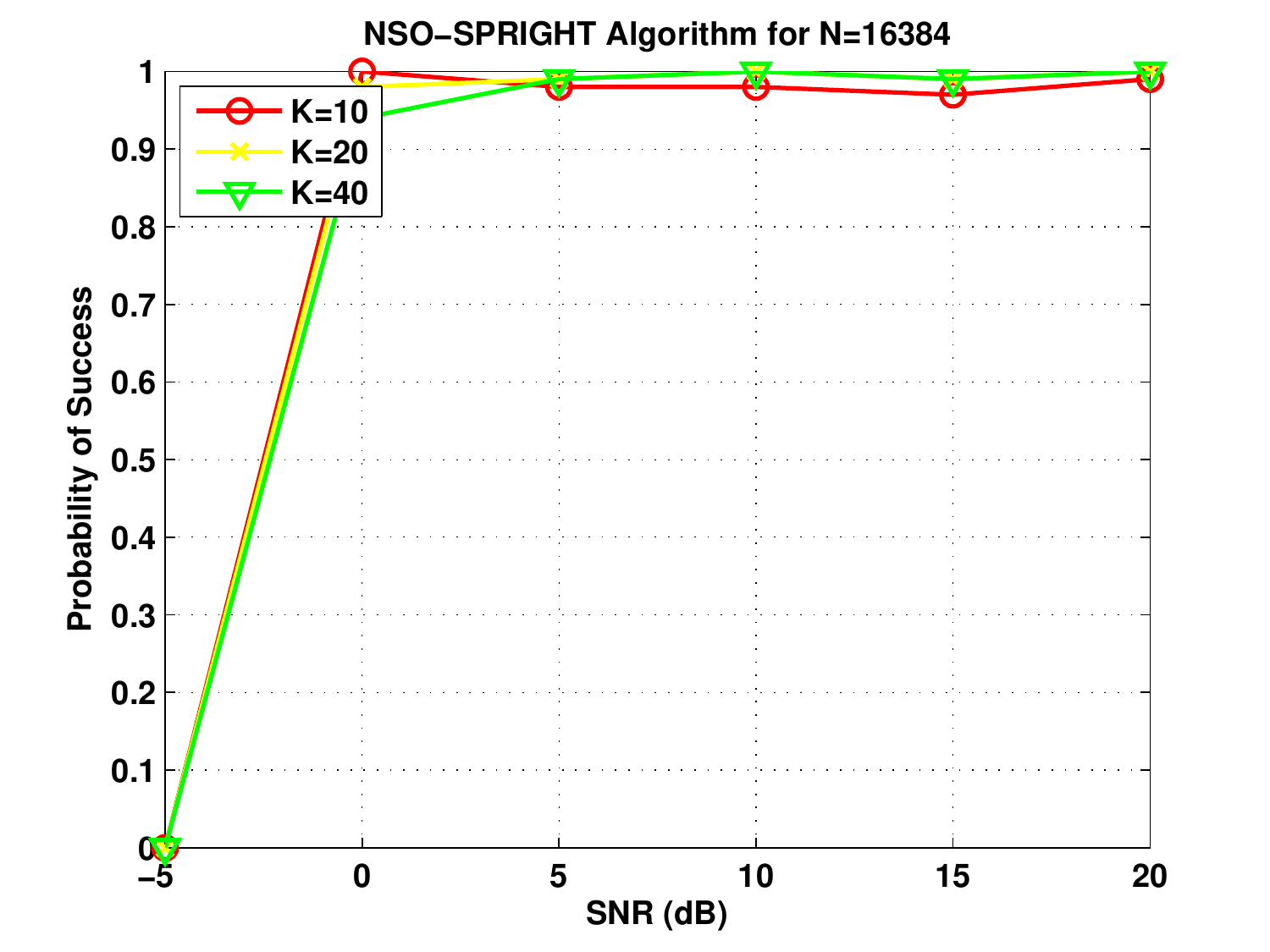}
\label{fig:NSO_prob}
}
\subfigure[Probability of success versus SNR]{
\includegraphics[width=0.48\linewidth]{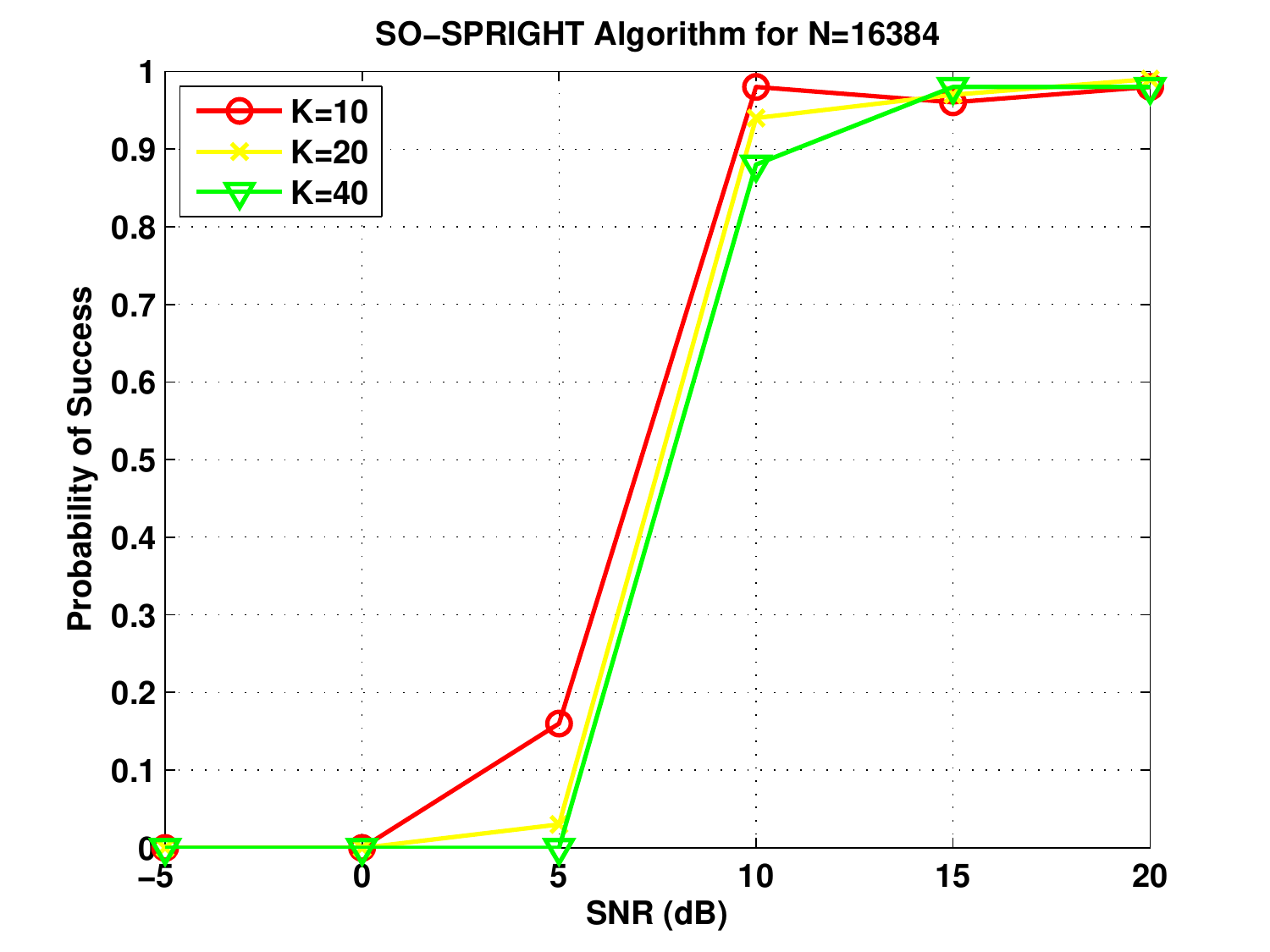}
\label{fig:SO_prob}
}
\end{center}
\vspace{-0.5cm}
\end{figure}

Note that the sample complexity of the NSO-SPRIGHT algorithm is approximately a factor of $n$ more than the SO-SPRIGHT algorithm, and thus the recovery performance is better under the same experiment setup. However, this is due to our simple choice of $(3,6)$-regular LDPC codes for inducing the offsets in the SO-SPRIGHT algorithm, which is far from capacity-achieving. Potentially one can use better LDPC code ensembles or even spatially coupled LDPC codes to provide better performance at the low SNR regime. Here the $(3,6)$-regular ensemble is simply an example to showcase the algorithm.

\subsubsection{Sample Complexity and Run-time Performance}
In this subsection, we compare the sample complexity and run-time performance of the NSO-SPRIGHT and SO-SPRIGHT algorithms. The experiment settings are given below:
\vspace{-0.1cm}
\begin{itemize}
	\item {\it input profile}: we generate a sparse WHT vector or length $N=2^n$ with $K = 10, 20, 40$ non-zero coefficients respectively and vary $n$ from $n=7$ to $n=17$. Therefore, the signal dimension spans from $N=128\approx 10^2$ to $131072\approx 0.1\times 10^6$. The non-zero WHT coefficients are chosen with uniformly random support and random amplitudes $\{\pm 1\}$. The input signal samples $\mathbf{x}$ is obtained by taking the inverse WHT of the sparse WHT vector and adding i.i.d. Gaussian noise samples with variance $\sigma^2$ determined by the $\mathsf{SNR}=10$ dB.	
	\vspace{-0.2cm}
	\item {\it benchmark}: as the signal length $N=2^n$ varies, the algorithm parameters are fixed over $200$ random experiments. We record a data point  only when the success probability exceeds $0.95$.
\end{itemize}
%\vspace{-0.1cm}
%
\begin{figure}[h]
\vspace{-0.3cm}
\begin{center}
\subfigure[NSO-SPRIGHT : signal length $N=2^n$ increases by $1000$ fold while the sample complexity increases by $5$ fold.]{
\includegraphics[width=0.47\linewidth]{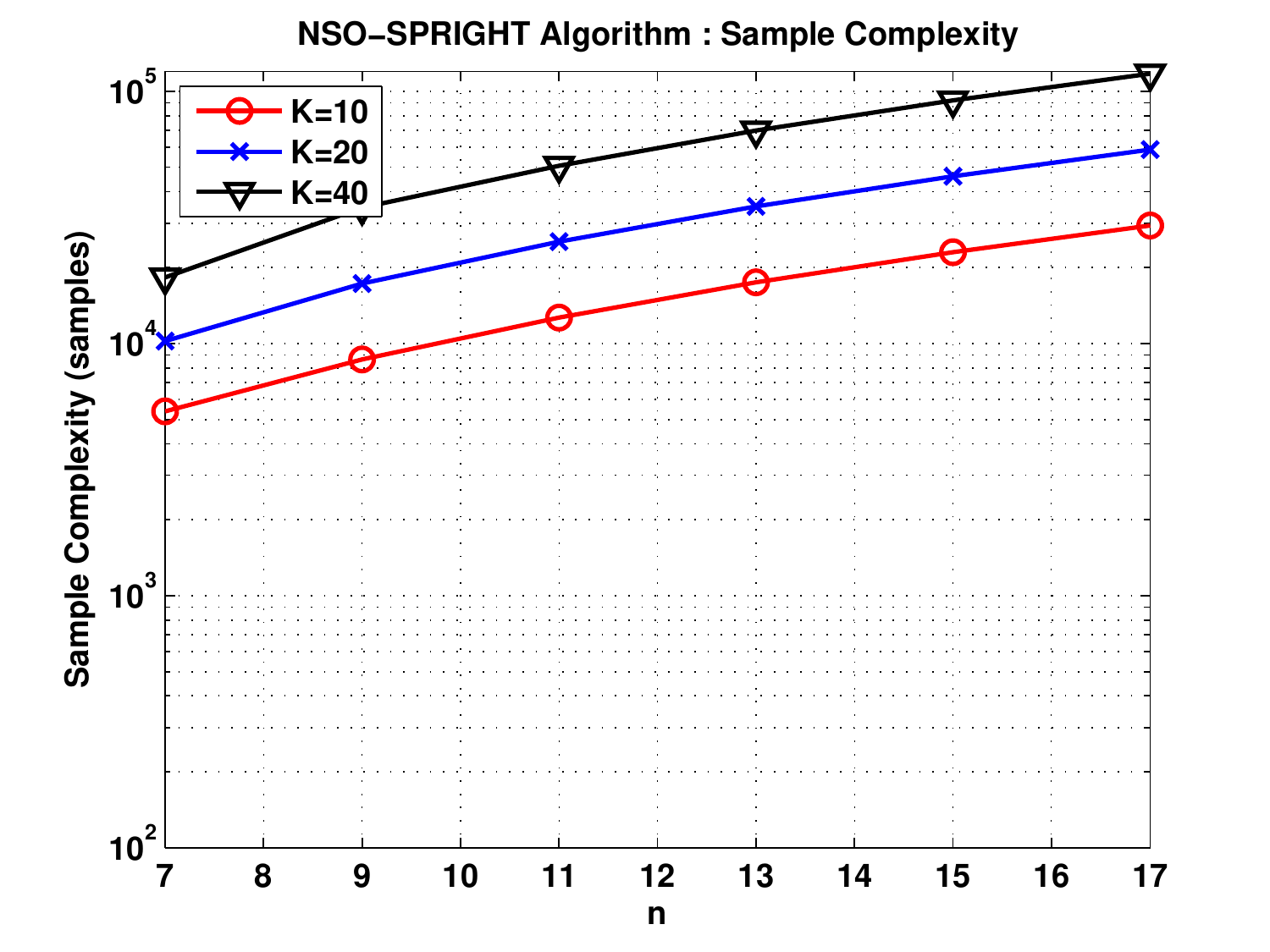}
\label{fig:NSO_meas}
}
~~
\subfigure[SO-SPRIGHT : signal length $N=2^n$ increases by $1000$ fold while the sample complexity increases by $3$ fold.]{
\includegraphics[width=0.47\linewidth]{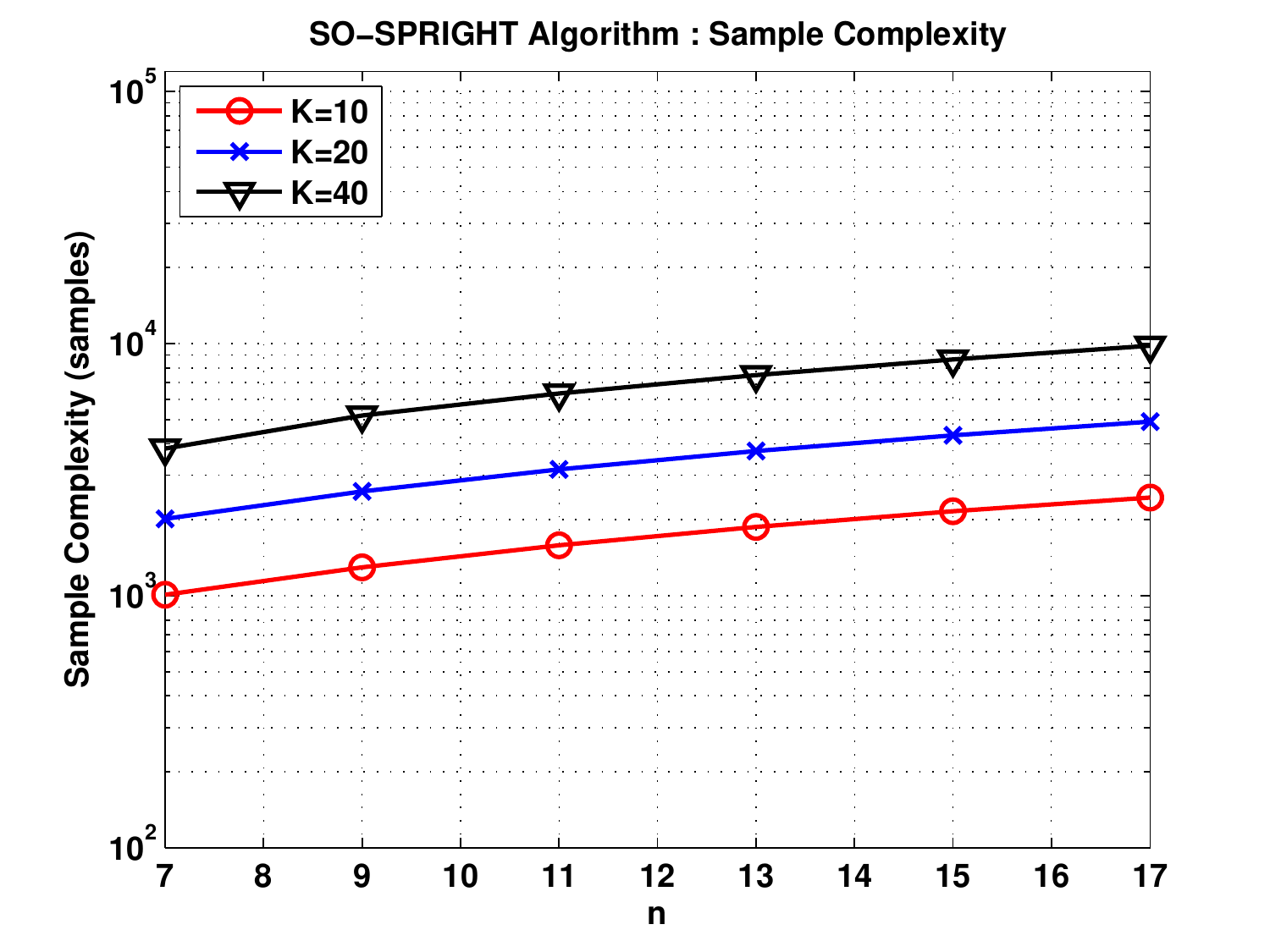}
\label{fig:SO_meas}
}\\
~\\
\subfigure[NSO-SPRIGHT : signal length $N=2^n$ increases by $1000$ fold while the run-time increases by at most $6$ fold.]{
\includegraphics[width=0.47\linewidth]{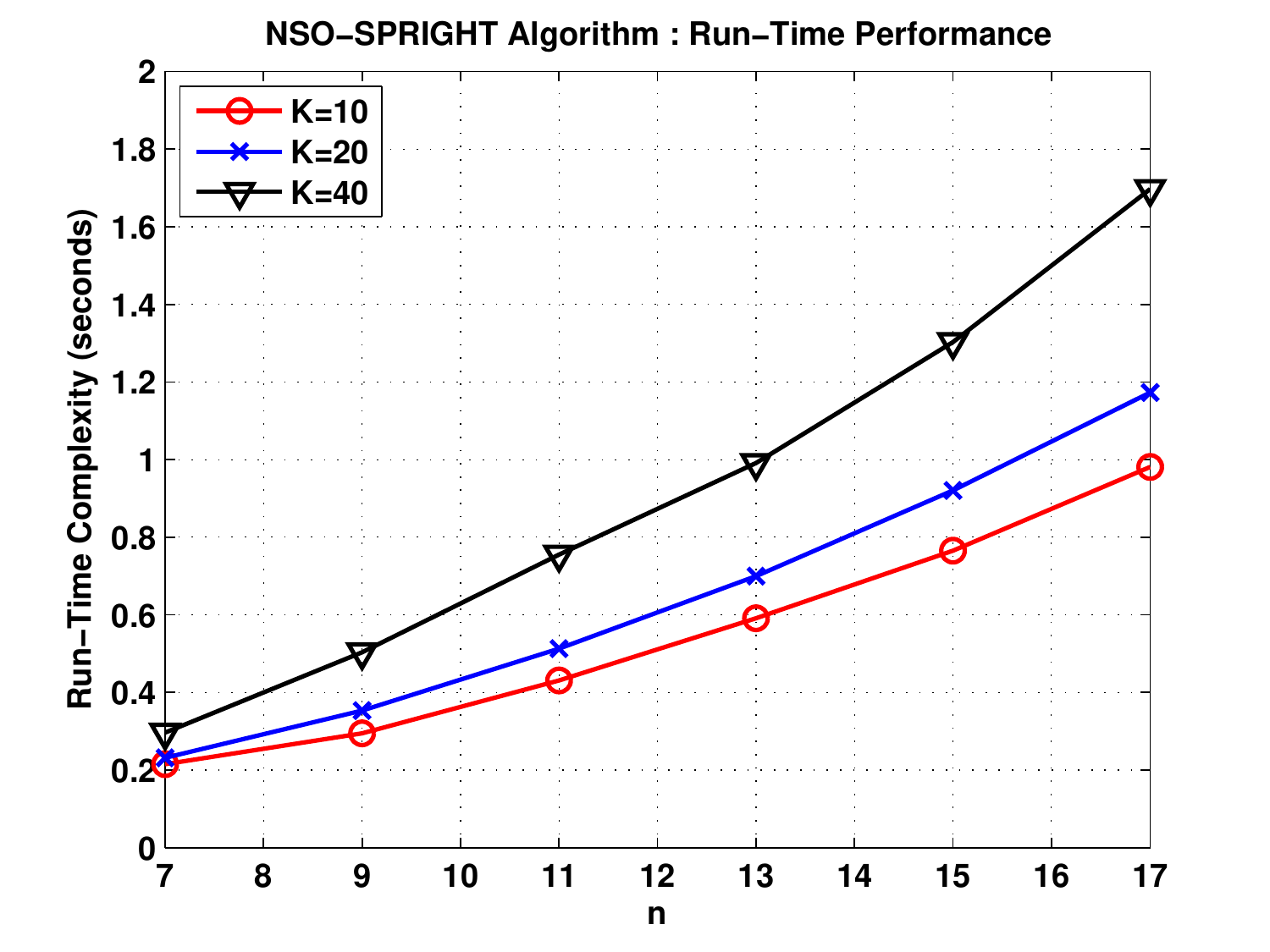}
\label{fig:NSO_time}
}
~~
\subfigure[SO-SPRIGHT : signal length $N=2^n$ increases by $1000$ fold while the run-time increases by at most $2$ fold.]{
\includegraphics[width=0.47\linewidth]{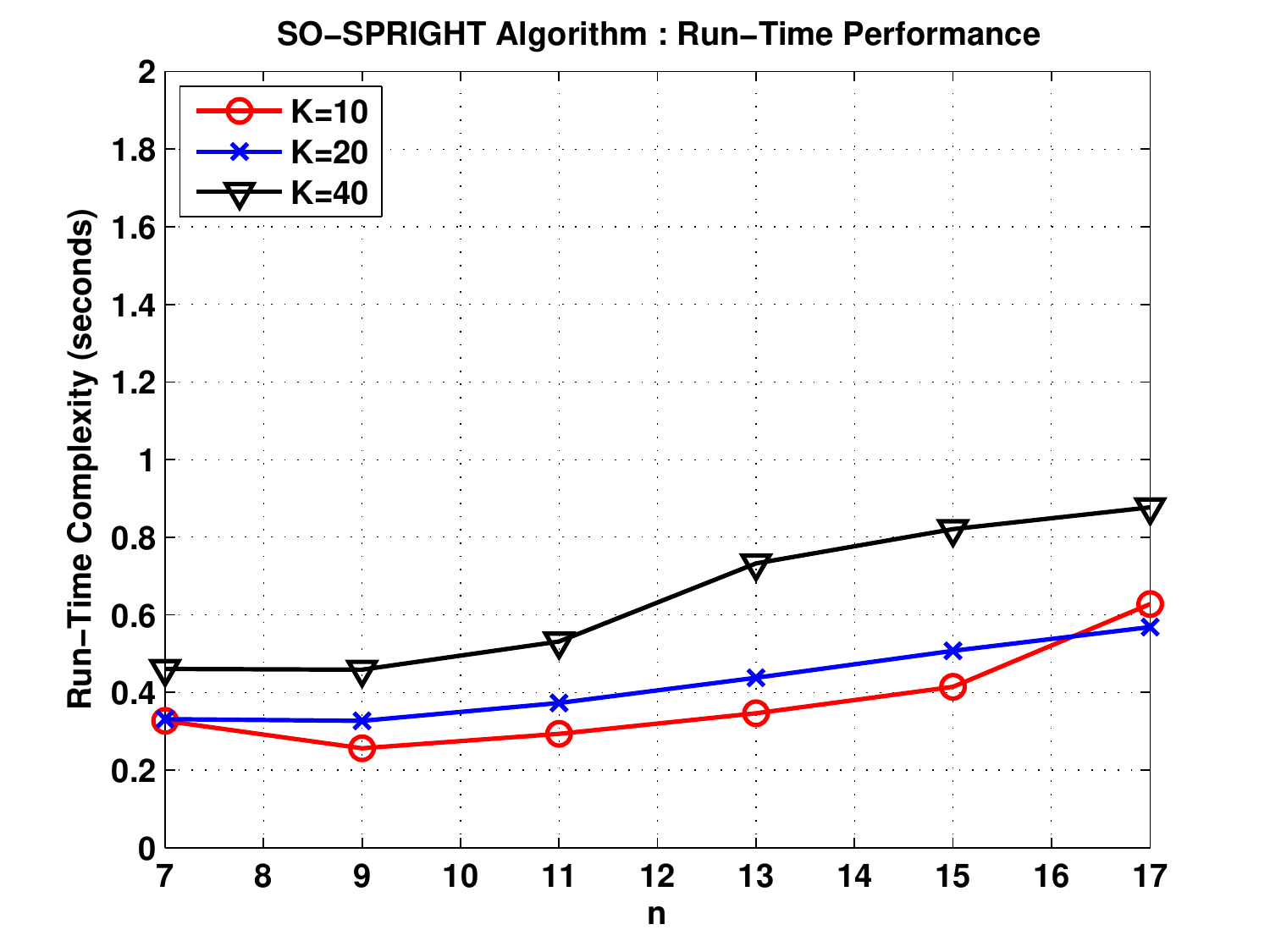}
\label{fig:SO_time}
}
\end{center}
\vspace{-0.4cm}
\caption{The plot shows the scaling of the sample complexity and run-time of the NSO-SPRIGHT and SO-SPRIGHT algorithms for inputs with varying dimensions $N=2^n$. With probability of success exceeding $0.95$ and sparsity $K=10,20,40$ at a constant SNR of $10$ dB, both the sample complexity and the run-time of the NSO-SPRIGHT and SO-SPRIGHT algorithms scale sub-linearly in $N$ (i.e. linear in $n^2$).}
\end{figure}

%\subsection{Applications in Decision Tree Learning}

%%%%%%%%%%%%%%%%%%%%%%%%%%%%%%%%%%%%%%%%%%%%%%%%%%%%%%%%%%%%
\section{Conclusions}

In this paper, we have proposed the SPRIGHT framework to compute a $K$-sparse $N$-point WHT, where the NSO-SPRIGHT algorithm uses ${O}(K \log^2 N) $ samples and ${O}(K \log^3 N)$ operations while the SO-SPRIGHT algorithm maintains the optimal sample scaling ${O}(K\log N)$ and complexity ${O}(K\log^2N)$ as that of the noiseless case. Our approach is based on strategic subsampling of the input noisy samples using a small set of randomly shifted patterns that are carefully designed, which achieves a vanishing failure probability.

\newpage
\bibliographystyle{IEEEtran}
\bibliography{SWHT_journal,SWHT_ISIT,Noisy_Hybrid,learning_sparse_polynomial}

\newpage
\appendix

\section*{Appendices}

\section{Proof of Theorem \ref{thm_main.result_sublinear}}\label{main.results.fast}
From Theorem \ref{thm_peeling_decoder_general}, it is shown that as long as $C\leq 8$ groups and $B={O}(K)$, the oracle-based peeling decoder succeeds with probability at least $1-{O}(1/K)$ for $0<\delta<1$. In Theorem \ref{peeling-decoder-RBI}, it is further shown that with the proposed bin detection routine using $P$ observation sets (chosen differently) in each group, the peeling decoder continues to succeed with probability at least $1-{O}(1/K)$ in the presence of noise. Therefore, the sample complexity is $M=CBP={O}(KP)$. On the other hand, the computational complexities stem from two sources:
\begin{itemize}
	\item The computation of $B$-point WHTs for subsampling: there are $P$ observations sets in each group, where each observation set requires a $B$-point WHT. Thus the total complexity is ${O}(P B\log B) = {O}(PK\log N)$, where $K = {O}(N^\delta)$ has been used;
	\item The bin detection routine in each peeling iteration for decoding:
	In the NSO-SPRIGHT scheme it is a majority vote, which leads to a complexity of ${O}(P)$. In the SO-SPRIGHT scheme it requires the decoding of a linear code formed by the $P$ offsets. As mentioned, one can potentially use (spatially coupled) LDPC or expander codes to achieve linear-time decoding ${O}(P)$, where $P$ is the block length of the code. Therefore, both sub-linear detection schemes result in a total complexity of ${O}(KP)$ throughout the ${O}(K)$ peeling iterations.
\end{itemize}
Clearly, the complexity is dominated by the subsampling $T={O}(P K\log N)$. Substituting the corresponding $P$ required by the sub-linear bin detection routines in the NSO-SPRIGHT and the SO-SPRIGHT schemes, we arrive at our stated results.

\section{Proof of Theorem \ref{thm_peeling_decoder_general} : Oracle-based Peeling Decoder Analysis}\label{sec:analysis_peeling_decoder}

\subsection{Design and Analysis for the Very Sparse Regime $0 < \delta \leq 1/3$}\label{sec:very_sparse}
To keep our discussions general, we choose $C$ subsampling groups and $B=2^b$ with $b=\delta n$ such that $B=\eta K$ for some $\eta>0$ and the subsampling matrices
\begin{align}\label{Psi_c_very_sparse}
	\mathbf{M}_c
	&=[\mathbf{0}_{(c-1)\times b}^T, \mathbf{I}_{b\times b}^T, \mathbf{0}_{(n-cb)\times b}^T]^T,\quad c \in [C],
\end{align}
which freezes a $(n-b)$-bit segment of the time domain indices $\mathbf{m}\in\GF^n$ to all zeros\footnote{The reason for $\delta = 1/3$ to be the separation point between the very sparse regime and the less sparse regime will become clear in Proposition \ref{thm_peeling_decoder} in the following section, where $C\geq 3$ is proven necessary for successful decoding with high probability. With the requirement $C\geq 3$ and 
the constraint $Cb\leq n$ due to the choice of $\mathbf{M}_c$, we have $b=\delta n$ and therefore $\delta \leq 1/3$.}. Then, each left node labeled $\mathbf{k}\in\GF^n$ is connected to a right node labeled $\bdsb{j}\in\GF^b$ determined by the aliasing pattern $\mathbf{M}_c^T\mathbf{k}=\bdsb{j}$. Therefore, the graph ensemble $\mathcal{G}(K,\eta,C,\{\mathbf{M}_c\}_{c\in[C]})$ in Definition \ref{def:graph_ensemble} is consistent with the ``{\it balls-and-bins}'' model, where the $\mathbf{k}$-th ball (i.e. left node $\mathbf{k}$) is thrown to bin $\bdsb{j}_c=\mathcal{H}_c(\mathbf{k})$ in group $c$. Now we show that given the uniform support distribution, the graph ensemble is further consistent with the random ``balls-and-bins'' model in each group.

We divide the index $\mathbf{k}$ into $C+1$ segments as $\mathbf{k}=[\bdsb{k}_1^T,\bdsb{k}_2^T,\cdots,\bdsb{k}_{C-1}^T,\bdsb{k}_C^T,\bdsb{k}_{C+1}^T]^T$, where each of the first $C$ segments $\bdsb{k}_c= [k[cb],\cdots,k[(c-1)b+1]]^T$ for $c\in[C]$ contains $b$ bits while the last segment $\bdsb{k}_{C+1}= [k[n],\cdots,k[Cb+1]]^T$ contains the remaining $(n-Cb)$ bits. 
Then, the hash functions associated with the subsampling matrices in \eqref{Psi_c_very_sparse} are $\mathcal{H}_c(\mathbf{k}) = \mathbf{M}_c^T\mathbf{k} = \bdsb{k}_c$, which sifts out the $b$-bit segment $\bdsb{k}_c$ independently out of $n$ bits from the index $\mathbf{k}$ in group $c$. We call the output of the hash function in each group the {\it bit segmentation}. Clearly, these {\it bit segmentations} can be chosen differently according to the choice of subsampling matrices $\{\mathbf{M}_c\}_{c\in[C]}$. For example, the {\it bit segmentations} in the first $3$ groups are
\begin{align}
	\bdsb{j}_1 =
	\begin{bmatrix}
		k[1]\\
		\vdots\\
		k[b]
	\end{bmatrix},
	\quad
	\bdsb{j}_2 =
	\begin{bmatrix}
		k[b+1]\\
		\vdots\\
		k[2b]
	\end{bmatrix},
	\quad	
	\bdsb{j}_3 =
	\begin{bmatrix}
		k[2b+1]\\
		\vdots\\
		k[3b]
	\end{bmatrix}.
\end{align}
Since each element $\mathbf{k}$ of the support set $\mathcal{K}$ is chosen independently and uniformly at random from $\GF^n$ by Assumption \ref{random_support_assumption}, each {\it bit segmentation} $\bdsb{j}_c = \mathcal{H}_c(\mathbf{k})$ is independently and uniformly chosen from $\{0,1\}$ for each ball. Therefore, each left ball is thrown independently into the bins on the right, which suggests that the edges from each left node to each right node are connected independently. Further, the bin index in each group $\bdsb{j}_c$ contains bit segments in $\mathbf{k}$ that are uniformly distributed, and hence each ball is thrown uniformly at random to one of the $B$ right nodes in that group.

In the following, we show that if the redundancy parameter $\eta=B/K$ is chosen appropriately for the graph ensemble $\mathcal{G}(K,\eta,C,\{\mathbf{M}_c\}_{c\in[C]})$ with $C$ subsampling groups and $\mathbf{M}_c$ chosen as \eqref{Psi_c_very_sparse}, then given the oracle, all the edges of the  graph can be peeled off in ${O}(K)$ peeling iterations with high probability. 

\begin{prop}[\bf Oracle-based Peeling Decoder Performance for $0<\delta\leq 1/3$]\label{thm_peeling_decoder}
If we use $C = 3$ groups with the set size $B=0.4073K$, where the subsampling matrices $\mathbf{M}_c$ for each group are chosen as in \eqref{Psi_c_very_sparse}, the induced  graph ensemble $\mathcal{G}(K,\eta,C,\{\mathbf{M}_c\}_{c\in[C]})$ guarantees that the oracle-based peeling decoder peels off all the edges in ${O}(K)$ iterations with probability at least $1-{O}(1/K)$.
\end{prop}
\begin{proof}
	The proof is given in the following subsections.
\end{proof}

\subsubsection{Density Evolution}\label{sec:DE}

\begin{wrapfigure}{r}{0.56\textwidth}
\vspace{-0.5cm}
\begin{minipage}{0.56\textwidth}
\begin{center}
\includegraphics[scale=.5]{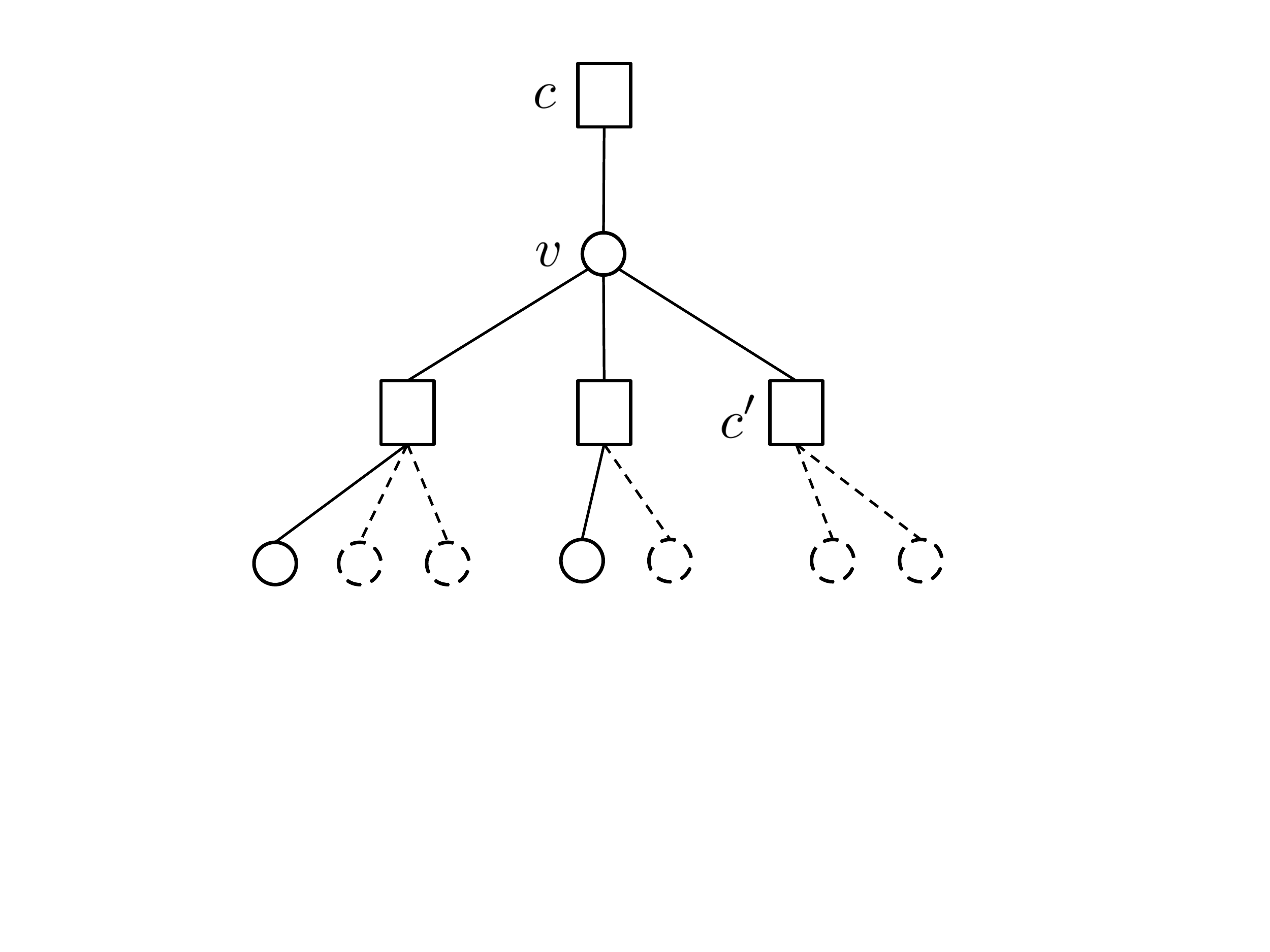}
\caption{Directed neighborhood of depth $2$ of an edge $\vec{e} = (v,c)$. The dashed lines correspond to nodes/edges removed at the end of iteration $i$. The edge between $v$ and $c$ can be potentially removed at iteration $i+1$ as one of the check nodes $c'$ is a singleton (it has no more variable nodes remaining at the end of iteration $i$).}
\vspace{-0.3cm}
\label{fig:localtree}
\end{center}
\end{minipage}
\vspace{-0.2cm}
\end{wrapfigure}

Density evolution, a powerful tool in modern coding theory, tracks the average density of remaining edges that are not decoded after a fixed number of peeling iteration $i>0$. We introduce the concept of {\it directed neighborhood} of a certain edge in the bipartite graph up to depth $\ell=2i$. This concept is important in the density evolution analysis since the peeling of an edge in the $i$-th iteration depends solely on the removal of the edges from this neighborhood in the previous $i-1$ iterations. The {\it directed neighborhood} $\mathcal{N}_{\textrm{e}}^\ell$ at depth $\ell$ of a certain edge $e = (v, c)$ is defined as the induced sub-graph containing all the edges and nodes on paths $e_1,\cdots, e_\ell$ starting at a variable node $v$ (left node) such that $e_1 \neq e$. An example of a directed neighborhood of depth $\ell=2$ is given in Fig. \ref{fig:localtree}. 

To analyze the performance of the peeling decoder over the bipartite graph, we need to understand the edge degree distributions on the left and right of the bipartite graph. Since the left edge degree distribution is already known due to the regularity of the graph ensemble induced by subsampling, next we study the right edge degree distribution.
\begin{lem}
Let $\rho_{j}$ be the fraction of edges in the bipartite graph connecting to right nodes with degree ${j}$. In the very sparse regime $0<\delta\leq 1/3$, if we use $C$ subsampling groups with subsampling matrices $\{\mathbf{M}_c\}_{c\in[C]}$ chosen as \eqref{Psi_c_very_sparse}, the edge degree sequence $\rho_j$ of the graph ensemble $\mathcal{G}(K,\eta,C,\{\mathbf{M}_c\}_{c\in[C]})$ is obtained as
\begin{align}\label{rho_d}
	\rho_{j} 
	= \frac{(1/\eta)^{{j}-1} e^{-1/\eta}}{({j}-1)!}.
\end{align}
\end{lem}
\begin{proof}
	See Appendix \ref{proof_right_edge_deg}.
\end{proof}

Now let us consider the local neighborhood $\mathcal{N}_{\textrm{e}}^{2i}$ of an arbitrary edge $e=(v,c)$  with a left regular degree ${d}$ and right degree distribution given by $\{\rho_{j}\}_{j=1}^{K}$. If the sub-graph corresponding to the neighborhood $\mathcal{N}_{\textrm{e}}^{2i}$ of the edge $e=(v,c)$ is a {\it tree} or namely {\it cycle-free}, then the peeling procedures over different bins in the first $i$ iterations are independent, which can greatly simplify our analysis. Density evolution analysis is based on the assumption that this neighborhood is cycle-free (tree-like), and we will prove later (in the next subsection) that all graphs in the regular ensemble behave like a tree when $N$ and $K$ are large and hence the actual density evolution concentrates well around the density evolution result.

Let $p_i$ be the probability of this edge being present in the bipartite graph after $i>0$ peeling iterations. If the neighborhood is a tree as in \figref{fig:localtree}, the probability $p_i$ can be written with respect to the probability $p_{i-1}$ at the previous depth in a recursive manner $p_i = \left(1-\sum_{j} \rho_{j} (1-{p}_{i-1})^{{j}-1}\right)^{C-1}$ for $i = 1,2,3,\cdots$. The term $\sum_{j} \rho_{j} (1-{p}_{i-1})^{{j}-1}$ can be approximated using the right degree generating polynomial
\begin{align}
	\rho(x) \defn \sum_{j} \rho_{j} x^{{j}-1} = e^{-(1-x)\frac{1}{\eta}},
\end{align}
where we have used \eqref{rho_d} to derive the second expression. Therefore, the density evolution equation for our peeling decoder can be obtained as
\begin{align}\label{density_evolution}
	p_i = \left(1-e^{-\frac{1}{\eta}p_{i-1}}\right)^{C-1}, i= 1,2,3,\cdots
\end{align}
Clearly, the probability $p_i$ can be made arbitrarily small for a sufficiently large but finite $i>0$ as long as $C$ and $\eta$ are chosen properly. One can find the minimum value $\eta$ for a given $C$ to guarantee $p_i<p_{i-1}$, which is shown in Table \ref{Table_beta}. Due to lack of space we only show up to $C=6$.
\begin{table}[h]
\begin{center}
\begin{tabular}{|c|c|c|c|c|c|c|c|c|}
  \hline
  $C$ & 2 & 3& 4 & 5 & 6\\% &7 \\
  \hline
  $\eta$ &1.0000 &   0.4073   & 0.3237 &   0.2850  &  0.2616  \\% & 0.2456  \\
  \hline
  $C\eta$ & 2.0000 &   1.2219 &   1.2948 &   1.4250  &  1.5696\\ %  &  1.7192\\
  \hline
\end{tabular}
\vspace{-0.4cm}
\end{center}
\caption{Minimum value for $\eta$ given the number of groups $C$}\label{Table_beta}
\vspace{-0.4cm}
\end{table}

\begin{lem}[\bf Density evolution $0<\delta\leq 1/3$]\label{lem:DE_verysparse}
Let $\mathcal{G}(K,\eta,C,\{\mathbf{M}_c\}_{c\in[C]})$ be the graph ensemble induced by subsampling with $C$ subsampling groups using subsampling matrices $\{\mathbf{M}_c\}_{c\in[C]}$ in \eqref{Psi_c_very_sparse} in the very sparse regime $0<\delta\leq 1/3$, where the number of groups $C$ and the redundancy parameter $\eta$ chosen from Table \ref{Table_beta}. Denote by $\mathcal{T}_{i}$ the event where the local $2i$-neighorhood $\mathcal{N}_{\textrm{e}}^{2{i}}$ of every edge in the graph is tree-like and let $Z_{i}$ be the total number of edges that are not decoded after $i$ (an arbitrarily large but fixed) peeling iterations. For any $\varepsilon>0$, there exists a finite number of iteration $i>0$ such that 
	\begin{align}\label{mean_analysis_result}
		\mathbb{E}[Z_{i}|\mathcal{T}_{i}] = K C \varepsilon/4,
	\end{align}
where the expectation is taken with respect to the random graph ensemble $\mathcal{G}(K,\eta,C,\{\mathbf{M}_c\}_{c\in[C]})$.
\end{lem}
\begin{proof}
	See Appendix \ref{sec:density_evolution}.	
\end{proof}	
Based on this lemma, we can see that if the bipartite graph has a local neighborhood that is tree-like up to depth $2i$ for every edge, the peeling decoder on average peels off all but an arbitrarily small fraction of the edges.

\subsubsection{Convergence to Density Evolution}\label{lem:convergence2DE}
Given the mean performance analysis (in terms of the number of undecoded edges) over cycle-free graphs, now we provide a {\it concentration analysis} on the number of the undecoded edges $Z_{i}$ for {\it any graph from the ensemble $\mathcal{G}(K,\eta,C,\{\mathbf{M}_c\}_{c\in[C]})$} at the $i$-th iteration, by showing that $Z_{i}$ converges to the density evolution.
\begin{lem}[\bf Convergence to density evolution for $0<\delta\leq 1/3$]\label{lem:convergence2DE}
Over the probability space of all graphs from $\mathcal{G}(K,\eta,C,\{\mathbf{M}_c\}_{c\in[C]})$,  let $p_i$ be as given in the density evolution \eqref{density_evolution}. Given any $\varepsilon>0$ and a sufficiently large $K$, there exists a constant $c>0$ such that
\begin{align}
	&{\tt (i)}~\quad\quad \mathbb{E}[Z_{i}] < K C \varepsilon/2 \label{mean_on_general_graph}\\
	&{\tt (ii)} \quad ~~~  \Prob{\left|Z_{i}-\mathbb{E}[Z_{i}]\right|>K C \varepsilon/2} 
	\leq
	2\exp\left(-c \varepsilon^2 K^{\frac{1}{4i+1}}\right).\label{concentration_DE}
\end{align}
\end{lem}
\begin{proof}
	We provide a {\it concentration analysis} in Appendix \ref{sec:concentration_analysis} on the number of the remaining edges for an {\it arbitrary graph from the ensemble} by showing that $Z_{i}$ converges to the mean analysis result. Here is a sketch of the proof:
\begin{itemize}
	\item {\it Mean analysis on general graphs from ensembles}: first, we use a counting argument similar to \cite{pedarsani2014phasecode} to show that any random graph from the ensemble $\mathcal{G}(K,\eta,C,\{\mathbf{M}_c\}_{c\in[C]})$ behaves like a {\it tree} with high probability. Therefore, the expected number of remaining edges can be made arbitrarily close to the mean analysis $|\mathbb{E}[Z_{i}]-\mathbb{E}[Z_{i}|\mathcal{T}_{i}] |<K C \varepsilon/4$ such that $\mathbb{E}[Z_{i}]<K C \varepsilon/2$ if $N$ and $K$ are greater than some constants.
	\item {\it Concentration to mean by large deviation analysis}: we use a Doob martingale argument as in \cite{richardson2001capacity} to show that the actual number of remaining edges $Z_{i}$ well concentrates around its mean $\mathbb{E}[Z_{i}]$ with an exponential tail in $K$ such that $\Prob{\left|Z_{i}-\mathbb{E}[Z_{i}]\right|>K C \varepsilon/2} \leq 2\exp\left(-c_4 \varepsilon^2 K^{\frac{1}{4i+1}}\right)$ for some constant $c_4>0$.
\end{itemize}
\end{proof}

\subsubsection{Complete Decoding through Graph Expanders}\label{sec:expander}
From previous analyses, it has already been established that with high probability, our peeling decoder terminates with an arbitrarily small fraction of edges undecoded
	\begin{align}\label{undecoded_edges}
		Z_{i} &< K C \varepsilon,\quad \forall \varepsilon>0,
	\end{align}
where ${d}$ is the left degree. In this section, we show that the all the undecoded edges can be completely decoded if the sub-graph consisting of the remaining undecoded edges is a ``good-expander''. Since there are many notions of ``graph expanders'', we introduce the concept of graph expander with respect to the graph ensemble $\mathcal{G}(K,\eta,C,\{\mathbf{M}_c\}_{c\in[C]})$ in this paper, which is induced by subsampling.
\begin{defi}[Graph Expander]\label{def:graph_expander}
A $C$-regular graph with $K$ left nodes and $C$ subsampling groups of $B=\eta K$ right nodes is called a $(\varepsilon,1/2,C)$-expander if for all subsets $\mathcal{S}$ of left nodes with $|\mathcal{S}|\leq \varepsilon K$, there exists a right neighborhood in some group $c$, denoted by $\mathcal{N}_c(\mathcal{S})$, that satisfies $|\mathcal{N}_c(\mathcal{S})| > |\mathcal{S}|/2$ for some $c\in[C]$.
\end{defi}

\begin{lem}[\bf Graph expansion property for $0<\delta\leq 1/3$]\label{lem_graph_expander}
In the very sparse regime $0<\delta\leq 1/3$, if we use $C\geq 3$ groups with subsampling matrices $\{\mathbf{M}_c\}_{c\in[C]}$ chosen as \eqref{Psi_c_very_sparse}, then any graph from the ensemble $\mathcal{G}(K,\eta,C,\{\mathbf{M}_c\}_{c\in[C]})$ is a $(\varepsilon,1/2,C)$-expander with probability at least $1-{O}(1/K)$ for some sufficiently small but constant $\varepsilon>0$.
\end{lem}
\begin{proof}
	See Appendix \ref{sec:expander_graph}.
\end{proof}
Without loss of generality, let the $Z_{i}$ undecoded edges be connected to a set of left nodes $\mathcal{S}$. Since each left node has degree $C$, it is obvious from \eqref{undecoded_edges} that $|\mathcal{S}|\leq K\varepsilon$ with high probability. Note that our peeling decoder fails to decode the set $\mathcal{S}$ of left nodes if and only if there are no more single-ton right nodes in the neighborhood of $\mathcal{S}$. A sufficient condition for all the right nodes in at least one group $\mathcal{N}_c(\mathcal{S})$ to have at least one single-ton is that the corresponding average degree is less than $2$, which implies that $|\mathcal{S}|/|\mathcal{N}_c(\mathcal{S})| \leq 2$ and hence $|\mathcal{N}_c(\mathcal{S})|\geq |\mathcal{S}|/2$. Since we have shown in Lemma \ref{lem_graph_expander} that any graph from the regular ensemble $\mathcal{G}(K,\eta,C,\{\mathbf{M}_c\}_{c\in[C]})$ is a $(\varepsilon,1/2,C)$-expander with high probability such that there is at least one group $|\mathcal{N}_c(\mathcal{S})|\geq |\mathcal{S}|/2$ for some $c$, there will be sufficient single-tons to peel off all the remaining edges.

\subsection{Design and Analysis of a Specific Less Sparse Regime $\delta = 1-1/C$}\label{sec:less_sparse}
From now on, we address the design and analysis for the less sparse regime $1/3<\delta < 1$. For convenience,  we start by discussing the case $\delta = 1-1/C$ where $C$ is the number of subsampling groups. Then, we generalize our design in Section \ref{sec:less_sparse_general} to tackle arbitrary sparsities $\delta\in(1/3,1)$ using the basic constructions for sparsity $\delta = 1-1/C$.  We let $t=n/C$ such that $B=2^b$ with $b=(C-1)t$ and $B=\eta K$ for some $\eta>0$. The subsampling matrices are chosen differently by 
\begin{align}\label{Psi_c_less_sparse}
	\mathbf{M}_c
	=
	\begin{bmatrix}
		\mathbf{I}_{(c-1)t \times (C-c)t} & \mathbf{0}_{(c-1)t \times (c-1)t}\\
		\mathbf{0}_{t \times (C-c)t} & \mathbf{0}_{t \times (c-1)t}\\
		\mathbf{0}_{(C-c)t \times (C-c)t} & \mathbf{I}_{(C-c)t \times (c-1)t}
	\end{bmatrix},\quad c\in[C],
\end{align}
which freezes a $t$-bit segment of the time domain indices $\mathbf{m}\in\GF^n$ to all zeros. 

\subsubsection{Random Graph Ensemble in the Less Sparse Regime $\delta = 1-1/C$}
For convenience, we divide $\mathbf{k}=[\bdsb{k}_1^T,\cdots,\bdsb{k}_{C-1}^T,\bdsb{k}_C^T]^T$ into $C$ pieces of $n/C$-bit segments with 
\begin{align}
	\bdsb{k}_c=[k[cn/C],\cdots, k[(c-1)n/C+1]]^T. 
\end{align}	
Then in this regime, the hash functions associated with \eqref{Psi_c_less_sparse} are defined as
\begin{align}
	\mathcal{H}_c(\mathbf{k})  = \mathbf{M}_c^T\mathbf{k} = [\bdsb{k}_1^T,\cdots,\bdsb{k}_{c-1}^T,\bdsb{k}_{c+1}^T,\cdots,\bdsb{k}_C^T]^T,\quad c\in[C],
\end{align}
which produces a {\it bit segmentation} that sifts out all but one segment $\bdsb{k}_c$ cyclically. Using this set of subsampling matrices (i.e. hash functions), the graph ensemble $\mathcal{G}(K,\eta,C,\{\mathbf{M}_c\}_{c\in[C]})$ in Definition \ref{def:graph_ensemble} is also consistent with the ``balls-and-bins'' model. For example, when $C=3$ and $\delta = 2/3$ such that $t=n/3$, the subsampling matrices are chosen as
\begin{align}
	\mathbf{M}_1
	=
	\begin{bmatrix}
		\mathbf{0}_{\frac{n}{3} \times \frac{n}{3}} & \mathbf{0}_{\frac{n}{3} \times \frac{n}{3}}\\	
		\mathbf{I}_{\frac{n}{3} \times \frac{n}{3}} & \mathbf{0}_{\frac{n}{3} \times \frac{n}{3}}\\
		\mathbf{0}_{\frac{n}{3} \times \frac{n}{3}} & \mathbf{I}_{\frac{n}{3} \times \frac{n}{3}}
	\end{bmatrix},
	\quad
	\mathbf{M}_2
	=
	\begin{bmatrix}
		\mathbf{I}_{\frac{n}{3} \times \frac{n}{3}} & \mathbf{0}_{\frac{n}{3} \times \frac{n}{3}}\\
		\mathbf{0}_{\frac{n}{3} \times \frac{n}{3}} & \mathbf{0}_{\frac{n}{3} \times \frac{n}{3}}\\
		\mathbf{0}_{\frac{n}{3} \times \frac{n}{3}} & \mathbf{I}_{\frac{n}{3} \times \frac{n}{3}}
	\end{bmatrix},
	\quad
	\mathbf{M}_3
	=
	\begin{bmatrix}		
		\mathbf{I}_{\frac{n}{3} \times \frac{n}{3}} & \mathbf{0}_{\frac{n}{3} \times \frac{n}{3}}\\
		\mathbf{0}_{\frac{n}{3} \times \frac{n}{3}} & \mathbf{I}_{\frac{n}{3} \times \frac{n}{3}}\\
		\mathbf{0}_{\frac{n}{3} \times \frac{n}{3}} & \mathbf{0}_{\frac{n}{3} \times \frac{n}{3}}		
	\end{bmatrix}.
\end{align}
and the bin indices corresponding to the ball $\mathbf{k}$ in the $3$ groups are given by
\begin{align}
	\bdsb{j}_1 =
	\begin{bmatrix}
		k[n/3+1]\\
		\vdots\\
		k[n]
	\end{bmatrix},
	\quad
	\bdsb{j}_2 =
	\begin{bmatrix}
		k[1]\\
		\vdots\\
		k[n/3]\\
		k[2n/3+1]\\
		\vdots\\
		k[n]
	\end{bmatrix},
	\quad	
	\bdsb{j}_3 =
	\begin{bmatrix}
		k[1]\\
		\vdots\\
		k[2n/3]
	\end{bmatrix}.
\end{align}
Same as the very sparse case, since each {\it bit segmentation} $\bdsb{j}_c = \mathcal{H}_c(\mathbf{k})$ is independently and uniformly at random from $\GF^n$ by Assumption \ref{random_support_assumption}, the bit patterns $k[i]$ for $i\in[n]$ are independently and uniformly chosen from $\{0,1\}$ for each ball. Therefore, each left ball is thrown independently into the bins on the right, which suggests that the edges from each left node to each right node are connected independently. Further, the bin index in each group $\bdsb{j}_c$ contains bit segments in $\mathbf{k}$ that are uniformly distributed, and hence each ball is thrown uniformly at random to one of the $B$ right nodes in that group. Therefore, due to the independence and uniformity of the support distribution $\mathbf{k}$, the graph ensemble is consistent with the random ``balls-and-bins'' model in each group.

\subsubsection{Peeling Decoder over the Graph Ensemble in the Less Sparse Regime $\delta= 1-1/C$}
The analysis of the peeling decoder in the less sparse regime depends on the graphs induced by subsampling. Note that the key difference of the graphs associated with the less sparse case from the very sparse case is that for each ball $\mathbf{k}$, although the edges are connected uniformly and independently to $B$ bins in each group, they are no longer connected independently across different groups. However, since the graph ensemble is consistent with the ``balls-and-bins'' model in each group, it can be easily shown that the density evolution analysis and concentration analysis carry over to the less sparse regime based on the analysis in Section \ref{sec:very_sparse}. However, there are some key differences in the graph expansion properties due to the lack of independence across different groups. In this section, we focus on proving the graph expansion properties for the graph ensemble in the less sparse regime.

\begin{lem}[\bf Graph expansion property for $\delta = 1-1/C$]\label{lem_graph_expander_less_sparse}
In the less sparse regime $\delta = 1-1/C$, if we use $C\geq 3$ groups with subsampling matrices $\{\mathbf{M}_c\}_{c\in[C]}$ chosen as \eqref{Psi_c_less_sparse} and $B=\eta K$ chosen with respect to the number of groups $C$ according to Table \ref{Table_beta}, then any graph from the ensemble $\mathcal{G}(K,\eta,C,\{\mathbf{M}_c\}_{c\in[C]})$ is a $(\varepsilon,1/2,C)$-expander with probability at least $1-{O}\left(\frac{1}{K^{(2^C-2C)/(C-1)}}\right)$ for some sufficiently small but constant $\varepsilon>0$.
\end{lem}
\begin{proof}

To show that the graph ensemble in the less sparse regime is a $(\varepsilon,1/2,C)$ expander defined in Definition \ref{def:graph_expander}, we need to show that irrespective of the inter-dependence of the edges across different groups, any subset $\mathcal{S}$ of left nodes has at least one right neighborhood in one group such that $\max_{c\in[C]}|\mathcal{N}_c(\mathcal{S})|\geq |\mathcal{S}|/2$. Since it has been shown in the very sparse regime in Lemma \ref{lem_graph_expander} that the bottleneck event of graph expander is when the size of the set is constant $|\mathcal{S}|={O}(1)$. Therefore in the following, we show that for any given subset $\mathcal{S}$ of left nodes with size $|\mathcal{S}|=s={O}(1)$, their right neighborhoods will not be multi-tons with high probability.

Given an arbitrary left node with the following bit segments
\begin{align}
	\mathbf{k} 
	&= [k[n],\cdots,k[1]]^T
	= [\bdsb{k}_C^T,\cdots,\bdsb{k}_1^T]^T,
	\quad \bdsb{k}_c \defn [k[ct],\cdots,k[(c-1)t+1]]^T,\quad c\in[C],
\end{align}	
its right neighbors are all multi-tons if and only if there exists at least another left node labeled $\mathbf{k}'$ in each group $c\in[C]$ such that $\mathcal{H}_c(\mathbf{k})=\mathcal{H}_c(\mathbf{k}')$. For a pathological set $\mathcal{S}$ where $\mathcal{H}_c(\mathbf{k})=\mathcal{H}_c(\mathbf{k}')$ for any distinct pair $\mathbf{k}\neq\mathbf{k}'\in\mathcal{S}$, the left node labels $\mathbf{k}$ and $\mathbf{k}'$ differ with each other only in one segment:
\begin{align}\label{pathological_set}
	\bdsb{k}_{c_\star} &\neq \bdsb{k}_{c_\star}',\quad \textrm{for some}~c_\star\in[C]\\
	\bdsb{k}_c &= \bdsb{k}_c', \quad ~\,\textrm{for all}~c\neq c_\star. 
\end{align}

Since there are at least $2$ such nodes for each group $c\in[C]$ to form multi-tons, the size of the pathological set $|\mathcal{S}|=s$ is satisfies $s\geq 2^C$. Let us consider the augmented worst case scenario where there are $s/2^{C-1}$ left nodes satisfying the pathological set requirements in \eqref{pathological_set} in one group (assuming there are only $2$ such nodes in other $C-1$ groups). For all the nodes $\mathbf{k}\in\mathcal{S}$, the total possible number of left nodes that can differ in one segment $\bdsb{k}_c$ for some $c\in[C]$ is $2^t$, and therefore the probability of having $s/2^{C-1}$ nodes from that space is $\frac{s}{2^{C-1}}/{2^t}$. In order for an arbitrary set of $s/2^{C-1}$ left nodes to land in the same bin on the right in all $C$ subsampling groups, the probability can be obtained as
\begin{align}
	\prod_{c=1}^C {2^t \choose s/2^{C-1}}\left(\frac{s}{2^{C-1}2^t}\right)^s.
\end{align}
Let ${\tt F}=2^t$, then the probability of this event can be obtained readily for any size $s$ as 
\begin{align}
	\Prob{{\mathcal{S}}} 
	&\leq {K \choose s} \prod_{c=1}^C {{\tt F} \choose s/2^{C-1}}\left(\frac{s}{2^{C-1} {\tt F} }\right)^s\\
	& = {K \choose s}  { {\tt F} \choose s/2^{C-1}}^C \left(\frac{s}{2^{C-1} {\tt F}}\right)^{Cs}.
\end{align}
Using the inequality ${a \choose b}\leq (a e/b)^b$, we have
\begin{align}
	\Prob{{\mathcal{S}}} 
	= {O}\left(\left(\frac{s}{{\tt F}}\right)^{\frac{s}{2^C}\left(2^C-2C\right)}\right)
\end{align}
Since the pathological set satisfies $s\geq 2^C$ and $K={O}({\tt F}^{C-1})$, we can further bound the probability as
\begin{align}
	\Prob{{\mathcal{S}}} 
	= {O}\left(\frac{1}{{\tt F}^{2^C-2C}}\right)= {O}\left(\frac{1}{K^{(2^C-2C)/(C-1)}}\right).
\end{align}
\end{proof}

\subsection{Generalized Design to Arbitrary Sparsity Regime $0 \leq \delta<1$}\label{sec:less_sparse_general}

As of now, we have presented the subsampling design for the very sparse regime $0<\delta\leq 1/3$ and partly for the less sparse regime $\delta = 1-1/C$ for $\delta = 2/3, 3/4, 4/5, 5/6, \cdots$ for all $C>0$. However, it does not generalize to any sparsity $0 < \delta < 1$. In this section, we continue to show that using the basic constructions above, we can achieve any sparsity regime.  The main idea of extending our subsampling design to an arbitrary sparsity is by the following:
\begin{itemize}
	\item {\it Hash with Common Prefix}: for example, we want to design the subsampling pattern for sparsity $\delta = (1+a)/(3+a)$ for some $a>0$. Clearly, by varying $a\in(0,\infty)$, one can obtain an arbitrary sparsity $\delta \in (1/3, 1)$. However, we hereby note that this construction is not universal since beyond some $a_\star$, the sparse bipartite graph constructed by this design fails to work with high probability. We will show later how to determine such threshold $a_\star$ and how to achieve sparsity beyond that point. In the following, we will proceed with this example.
	 
	We divide the bin index $\mathbf{k}$ into $4$ segments $\mathbf{k}=[\bdsb{k}_1^T,\bdsb{k}_2^T,\bdsb{k}_3^T,\bdsb{k}_4^T]^T$, where $\bdsb{k}_1$, $\bdsb{k}_2$ and $\bdsb{k}_3$ are of equal length containing $b_c = n/(3+a)$ bits for $c=1,2,3$, while $\bdsb{k}_4$ contains $b_4 = an/(3+a)$ bits. The hash function in each group is then designed with the following bit segmentation:
	\begin{align}
		\mathcal{H}_c(\mathbf{k}) = 
		\begin{bmatrix}
			\bdsb{k}_c\\
			\bdsb{k}_4
		\end{bmatrix},
		\quad
		c=1,2,3.
	\end{align}
	In this way, the output of the hash has $b$ bits with 
	\begin{align}
		b=b_c+b_4 = \frac{n}{3+a}+\frac{an}{3+a} = \frac{1+a}{3+a} n = \delta n,
	\end{align}		
	and hence we have $B=2^b=\eta K = {O}(N^{\delta})$ for some appropriately chosen $\eta$. We refer to this generalized hash design as {\it common-prefix} since the hash outputs start with the same segment $\bdsb{k}_4$.
	\item {\it Union of Disjoint Sparse Bipartite Graphs}: using the generalized hash designed above, the sparse bipartite graph is still consistent with the balls-and-bins model, where there are $B=2^{(b_c+b_4)}$ right nodes and $K$ left nodes. Furthermore, since the right node of the graph is indexed by two segments $(\bdsb{k}_4,\bdsb{k}_c)$, the resulting bipartite graph can be viewed as $2^{b_4}$ disjoint unions of sparse bipartite graphs with $K/2^{b_4}$ left nodes and $B/2^{b_4}=2^{b_c}$ right nodes. In other words, we have $2^{b_4}$ disjoint unions of graphs from the random graph ensemble $\mathcal{G}(K/2^{b_4},0.4073,3,\{\mathbf{M}_c\}_{c=1,2,3})$, the decoding of which fails with probability ${O}(1/K/2^{b_4})={O}(1/2^{b_c})$. Therefore, by a union bound, the failure probability of peeling decoding over the bipartite graphs given by this design  is
	\begin{align}
		{O}\left( \frac{1}{2^{b_c}}\right) \times 2^{b_4} = {O}\left( \frac{1}{2^{b_c-b_4}}\right) = {O}\left( \frac{1}{2^{\frac{1-a}{3+a}n}}\right) = {O}\left( \frac{1}{2^{\frac{1+a}{3+a}n \times \left(\frac{1-a}{1+a}\right)}}\right) = {O}\left( \frac{1}{K^{\left(\frac{1-a}{1+a}\right)}}\right).
	\end{align}
Clearly, it is required that $a<a_\star=1$ such that the failure probability approaches zero asymptotically in $K$. This implies a sparsity regime $\delta = (1+a)/(3+a) < (1+a_\star)/(3+a_\star) = 1/2$. Therefore, this example only works for sparsity $1/3<\delta<1/2$. In the following, we provide specific constructions that cover the entire sparsity regime $0<\delta<1$.
\end{itemize}

\subsubsection{Achieving Intermediate Sparsity $0 < \delta \leq 1/3$}\label{sec:delta1}
The design in Section \ref{sec:very_sparse} can be used directly and hence we omit the discussions here.

\subsubsection{Achieving Intermediate Sparsity $1/3<\delta\leq 0.73$}\label{sec:delta2}
Here we target sparsity $\delta=(2+a)/(6+a)$, which starts from $\delta=1/3$ with $a=0$ and ends at $\delta = 0.73$ with $a=8.81$. To achieve such sparsity, we divide the bin index $\mathbf{k}$ into $7$ segments 
\begin{align}
	\mathbf{k}=[\bdsb{k}_1^T,\bdsb{k}_2^T,\bdsb{k}_3^T,\bdsb{k}_4^T,\bdsb{k}_5^T,\bdsb{k}_6^T,\bdsb{k}_7^T]^T, 
\end{align}
where $\bdsb{k}_1$, $\bdsb{k}_2$, $\bdsb{k}_3$, $\bdsb{k}_4$, $\bdsb{k}_5$ and $\bdsb{k}_6$ are of equal length containing $b_c = n/(6+a)$ bits for $c=1,2,\cdots,6$, while $\bdsb{k}_7$ contains $b_7 = an/(6+a)$ bits. Therefore, we have $C=6$ groups for subsampling, and the hash function in each group is designed with the following bit segmentation:
	\begin{align}
		\mathcal{H}_1(\mathbf{k}) &= 
		\begin{bmatrix}
			\bdsb{k}_2\\
			\bdsb{k}_3\\
			\bdsb{k}_7
		\end{bmatrix},
		\quad
		\mathcal{H}_2(\mathbf{k}) = 
		\begin{bmatrix}
			\bdsb{k}_1\\
			\bdsb{k}_3\\
			\bdsb{k}_7
		\end{bmatrix},
		\quad
		\mathcal{H}_3(\mathbf{k}) = 
		\begin{bmatrix}
			\bdsb{k}_1\\
			\bdsb{k}_2\\
			\bdsb{k}_7
		\end{bmatrix}\\		
		\mathcal{H}_4(\mathbf{k}) &= 
		\begin{bmatrix}
			\bdsb{k}_5\\
			\bdsb{k}_6\\
			\bdsb{k}_7
		\end{bmatrix},
		\quad
		\mathcal{H}_5(\mathbf{k}) = 
		\begin{bmatrix}
			\bdsb{k}_4\\
			\bdsb{k}_6\\
			\bdsb{k}_7
		\end{bmatrix},
		\quad		
		\mathcal{H}_6(\mathbf{k}) = 
		\begin{bmatrix}
			\bdsb{k}_4\\
			\bdsb{k}_5\\
			\bdsb{k}_7
		\end{bmatrix}.
	\end{align}
	In this way, the output of the hash has $b$ bits with 
	\begin{align}
		b=2b_c+b_7 = \frac{2n}{6+a}+\frac{an}{6+a} = \frac{2+a}{6+a} n = \delta n.
	\end{align}
	According to Table \ref{Table_beta}, we need to choose $B=0.2616K$. Using the same analysis outlined before, we can show that the peeling decoder works with probability at least $1-{O}(1/K)$.

\subsubsection{Achieving Intermediate Sparsity $0.73<\delta\leq 7/8$}\label{sec:delta3}
Here we target sparsity $\delta=(3+a)/(8+a)$, which starts from $\delta=0.73$ with $a=10.52$ and ends at $\delta = 7/8$ with $a=32$. To achieve such sparsity, we divide the bin index $\mathbf{k}$ into $9$ segments 
\begin{align}
	\mathbf{k}=[\bdsb{k}_1^T,\bdsb{k}_2^T,\bdsb{k}_3^T,\bdsb{k}_4^T,\bdsb{k}_5^T,\bdsb{k}_6^T,\bdsb{k}_7^T,\bdsb{k}_8^T,\bdsb{k}_9^T]^T, 
\end{align}
where $\bdsb{k}_1$, $\bdsb{k}_2$, $\bdsb{k}_3$, $\bdsb{k}_4$, $\bdsb{k}_5$, $\bdsb{k}_6$, $\bdsb{k}_7$ and $\bdsb{k}_8$ are of equal length containing $b_c = n/(8+a)$ bits for $c=1,2,\cdots,8$, while $\bdsb{k}_9$ contains $b_9 = an/(8+a)$ bits. Therefore, we have $C=8$ groups for subsampling, and the hash function in each group is designed with the following bit segmentation:
	\begin{align}
		\mathcal{H}_1(\mathbf{k}) &= 
		\begin{bmatrix}
			\bdsb{k}_2\\
			\bdsb{k}_3\\
			\bdsb{k}_4\\
			\bdsb{k}_9
		\end{bmatrix},
		\quad
		\mathcal{H}_2(\mathbf{k}) = 
		\begin{bmatrix}
			\bdsb{k}_1\\
			\bdsb{k}_3\\
			\bdsb{k}_4\\
			\bdsb{k}_9
		\end{bmatrix},
		\quad
		\mathcal{H}_3(\mathbf{k}) = 
		\begin{bmatrix}
			\bdsb{k}_1\\
			\bdsb{k}_2\\
			\bdsb{k}_4\\
			\bdsb{k}_9
		\end{bmatrix},
		\quad
		\mathcal{H}_4(\mathbf{k}) = 
		\begin{bmatrix}
			\bdsb{k}_1\\
			\bdsb{k}_2\\
			\bdsb{k}_3\\
			\bdsb{k}_9
		\end{bmatrix}\\		
		\mathcal{H}_5(\mathbf{k}) &= 
		\begin{bmatrix}
			\bdsb{k}_6\\
			\bdsb{k}_7\\
			\bdsb{k}_8\\
			\bdsb{k}_9
		\end{bmatrix},
		\quad
		\mathcal{H}_6(\mathbf{k}) = 
		\begin{bmatrix}
			\bdsb{k}_5\\
			\bdsb{k}_7\\
			\bdsb{k}_8\\
			\bdsb{k}_9
		\end{bmatrix},
		\quad		
		\mathcal{H}_7(\mathbf{k}) = 
		\begin{bmatrix}
			\bdsb{k}_5\\
			\bdsb{k}_6\\
			\bdsb{k}_8\\
			\bdsb{k}_9
		\end{bmatrix},
		\quad
		\mathcal{H}_8(\mathbf{k}) = 
		\begin{bmatrix}
			\bdsb{k}_5\\
			\bdsb{k}_6\\
			\bdsb{k}_7\\
			\bdsb{k}_9
		\end{bmatrix}.
	\end{align}
	In this way, the output of the hash has $b$ bits with 
	\begin{align}
		b=3b_c+b_9 = \frac{3n}{8+a}+\frac{an}{8+a} = \frac{3+a}{8+a} n = \delta n.
	\end{align}
	According to Table \ref{Table_beta}, we need to choose $B=0.2336K$. Using the same analysis outlined before, we can show that the peeling decoder works with probability at least $1-{O}(1/K^{0.9})$.

\subsubsection{Achieving Intermediate Sparsity $7/8<\delta <1$}\label{sec:delta4}
The sparsity index $\delta$ in the range $0.875 < \delta < 1$ can be achieved by the
combination of designs proposed in the less sparse regime for increasing (but constant) number of groups $C$ as dictated by $\delta=1-1/C$. For example, we can target the sparsity setting $\delta=(7+a)/(8+a)$ starting from $\delta=0.875$ with $a=0$ and until $\delta=0.99$. In this construction, we divide the bin index $\mathbf{k}$ into $9$ segments 
\begin{align}
	\mathbf{k}=[\bdsb{k}_1^T,\bdsb{k}_2^T,\bdsb{k}_3^T,\bdsb{k}_4^T,\bdsb{k}_5^T,\bdsb{k}_6^T,\bdsb{k}_7^T,\bdsb{k}_8^T,\bdsb{k}_9^T]^T, 
\end{align}
where $\bdsb{k}_1$, $\bdsb{k}_2$, $\bdsb{k}_3$, $\bdsb{k}_4$, $\bdsb{k}_5$, $\bdsb{k}_6$, $\bdsb{k}_7$ and $\bdsb{k}_8$ are of equal length containing $b_c = n/(8+a)$ bits for $c=1,2,\cdots,8$, while $\bdsb{k}_9$ contains $b_9 = an/(8+a)$ bits. The hash function in each group is then designed with the following bit segmentation:
	\begin{align}
		\mathcal{H}_c(\mathbf{k}) = 
		\begin{bmatrix}
			\bdsb{k}_1\\
			\bdsb{k}_2\\
			\vdots\\
			\bdsb{k}_{c-1}\\
			\bdsb{k}_{c+1}\\
			\vdots\\
			\bdsb{k}_9
		\end{bmatrix},
		\quad
		c=1,2,\cdots,8.
	\end{align}
In this way, the output of the hash has $b$ bits with 
	\begin{align}
		b=7b_c+b_9 = \frac{7n}{8+a}+\frac{an}{8+a} = \frac{7+a}{8+a} n = \delta n.
	\end{align}
	According to Table \ref{Table_beta}, we need to choose $B=0.2336K$. Using the same analysis outlined before, we can show that the peeling decoder works with probability at least $1-{O}(1/K)$.

\subsection{Right Edge Degree Distribution}\label{proof_right_edge_deg}
Clearly, the total number of edges is $K C$ in the bipartite graph since there are $K$ left nodes in the bipartite graph and each left node has degree $C$. Therefore, since the expected number of right nodes with degree ${j}$ can be obtained as $\Prob{\textrm{a right node has degree}~{j}} CB {j}$, the fraction $\rho_{j}$ can be obtained as
\begin{align}
	\rho_{j} 
	= \frac{\Prob{\textrm{a right node has degree}~{j}} CB{j} }{KC} = {j}\eta \Prob{\textrm{a right node has degree}~{j}},
\end{align}
where we have used $B=\eta K$ and $\eta$ is the redundancy parameter. According to the ``balls-and-bins'' model, the degree of a right node follows the binomial distribution $B(1/\eta K, K)$ and can be well approximated by a Poisson variable as
\begin{align}
	\Prob{\textrm{a right node has degree}~{j}} \approx \frac{(1/\eta)^{j} e^{-1/\eta}}{{j}!}.
\end{align}
As a result, the fraction $\rho_{j}$ of edges connected to right nodes having degree ${j}$ is obtained as \eqref{rho_d}.

\subsection{Proof of Mean Performance}\label{sec:density_evolution}
Let $Z_{i}^{\textrm{e}}\in\{0,1\}$ be the random variable denoting the presence of edge $e$ after $i$  iterations, thus 
\begin{align}\label{sum_edges}
	Z_{i} = \sum_{e=1}^{KC} Z_{i}^{\textrm{e}}.
\end{align}
Since each edge is peeled off independently given the event $\mathcal{T}_{i}$, the expected number of remaining edges over cycle-free graphs can be obtained as
\begin{align}\label{mean_analysis}
	\mathbb{E}\left[Z_{i}|\mathcal{T}_{i}\right] = \sum_{e=1}^{KC} \mathbb{E}\left[Z_{i}^{\textrm{e}}|\mathcal{T}_{i}\right]  = K C p_i,
\end{align}
where by definition $p_i = \Prob{Z_{i}^{\textrm{e}}=1|\mathcal{T}_{i}}$ is the {\it conditional probability} of an edge in the $i$-th peeling iteration conditioned on the event $\mathcal{T}_{i}$ studied in the density evolution equation \eqref{density_evolution}. We are interested in the evolution of such probability $p_i$. In the following, we prove that for any given $\varepsilon>0$, there exists a finite number of iterations $i>0$ such that $p_i \leq \varepsilon/4$, which leads to our desired result in \eqref{mean_analysis_result}.

\subsection{Concentration Analysis}\label{sec:concentration_analysis}

\subsubsection{Proof of Mean Analysis on General Graphs from Ensembles}
From \eqref{sum_edges}, we have
\begin{align}\label{mean_analysis}
	\mathbb{E}\left[Z_{i}\right] = \sum_{e=1}^{K C} \mathbb{E}\left[Z_{i}^{\textrm{e}}\right]  = K \bar{d} \mathbb{E}\left[Z_{i}^{\textrm{e}}\right].
\end{align}
From basic probability laws, we have
\begin{align*}
	\mathbb{E}\left[Z_{i}^{\textrm{e}}\right]
	&=
	\mathbb{E}\left[Z_{i}^{\textrm{e}}|\mathcal{T}_{i}\right]\Prob{\mathcal{T}_{i}}
	+\mathbb{E}\left[Z_{i}^{\textrm{e}}|\mathcal{T}_{i}^c\right]\Prob{\mathcal{T}_{i}^c}.
\end{align*}
Recall from the density evolution analysis that $\mathbb{E}\left[Z_{i}^{\textrm{e}}|\mathcal{T}_{i}\right] = p_i$, we have
\begin{align}
	\Prob{\mathcal{T}_{i}} \leq 1,\quad
	\mathbb{E}\left[Z_{\textrm{e}}|\mathcal{T}_{i}^c\right] \leq 1
\end{align}
and therefore the following holds:
\begin{align}
	{p}_{i} - \Prob{\mathcal{T}_{i}^c}
	\leq
	\mathbb{E}\left[Z_{i}^{\textrm{e}}\right] 
	\leq
	{p}_{i} + \Prob{\mathcal{T}_{i}^c}.
\end{align}
If the probability of a general graph not behaving like a tree can be made arbitrarily small for any $\varepsilon>0$,
\begin{align}\label{prob_not_tree}
	 \Prob{\mathcal{T}_{i}^c} < \frac{\varepsilon}{4},
\end{align}
then we can obtain the result in \eqref{mean_on_general_graph} by letting $p_i=\varepsilon/4$ in the density evolution analysis. Next, we show that \eqref{prob_not_tree} holds for sufficiently large $K$.
\begin{lem}
For any given constant $\varepsilon>0$ and iteration $i>0$, there exists some absolute constant $K_0>0$ such that 
\begin{align}
	\Prob{\mathcal{T}_{i}^c}
	<
	c_0 \frac{\log^i K}{K}
\end{align}
for some constant $c_0>0$ as long as $K>K_0$.
\end{lem}
From this lemma, we can see that for an arbitrary $\varepsilon>0$, the result follows as long as $K>K_0$ where $K_0$ is the smallest constant that satisfies $K_0/\log^i K_0 > 4c_0/\varepsilon$ given $\varepsilon$ and $i$. In the following we give the proof of the lemma.

\begin{proof}
Let $C_i$ be the number of check nodes and $V_i$ be the number of variable nodes in the neighborhood $\mathcal{N}_{\textrm{e}}^{2i}$. Because the graph ensemble $\mathcal{G}(K,\eta,C,\{\mathbf{M}_c\}_{c\in[C]})$ follows Poisson distributions on the right, the results in \cite{richardson2001capacity} are not readily applied here. Therefore, the key idea is to prove that the size of the tree neighborhood is bounded by ${O}(\log^i K)$ with high probability, which is intuitively clear since  Poisson distributions have very light tails due to the exponential decay. 

To show this, we unfold the neighborhood of an edge $e$ up to level $i_\star$, and at each level $i$ we upper bound the probability that the size of the tree grows larger than ${O}(\log^i K)$. Specifically, from the law of total probability, we upper bound the probability of not having a tree as follows for some $\kappa_1>0$
\begin{align}
	\Prob{\mathcal{T}_{i}^c}
	&\leq 
	\Prob{V_i > \kappa_1 \log^i K}
	+ \Prob{C_i > \kappa_1 \log^i K}\\
	&~~~ + \Prob{\mathcal{T}_{i}^c | V_i < \kappa_1 \log^i K,C_i < \kappa_1 \log^i K}.
\end{align}
Denoting $\alpha_i = \Prob{V_i > \kappa_1 \log^i K}$, we bound the first term using the total law of probability as follows
\begin{align*}
			\alpha_i
			&\leq
			\alpha_{i-1} +
			\Prob{V_i > \kappa_1 \log^i K | V_{i-1} < \kappa_1  \log^{i-1} K}.
\end{align*}
Given $V_{i-1} < \kappa_1  \log^{i-1} K$, we have $C_i<\kappa_2 \log^{i-1} K$ at depth $i$ for some $\kappa_2>0$ since the left degree of any graph from both ensembles is upper bounded by ${d}$ and $D$ respectively, which are both constants. Therefore, the second term in the above recursion can be bounded as
\begin{align}
	\Prob{V_i > \kappa_1 \log^i K | V_{i-1} < \kappa_1  \log^{i-1} K}
	&\leq
	\Prob{V_i > \kappa_1 \log^i K | C_i < \kappa_2 \log^{i-1} K}.
\end{align}
Now let the number of check nodes at exactly depth $i$ be $M_i$ and let $D_i$ be the degrees of each of these check nodes, the right hand side can be evaluated as
\begin{align}
	&\Prob{V_i > \kappa_1 \log^i K | C_i < \kappa_2 \log^{i-1} K}
	\leq
	\Prob{\sum_{i=1}^{M_i} D_i \geq \kappa_3 \log^i K}
\end{align}
for some $\kappa_3>0$. Since the check node degrees are Poisson variables with rate $1/\eta$ and the number of check nodes at depth $i$ is less than the total number of check nodes up to depth $i$ such that $M_i\leq C_i < \kappa_1 \log^i K$, then the probability can be upper bounded with $\Prob{D_i\geq x}\leq (e\lambda/x)^x$ as
\begin{align}
	\Prob{\sum_{i=1}^{M_i} D_i \geq \kappa_3 \log^i K}
	&\leq \left(\frac{e M_i/\eta}{\kappa_3 \log^i K}\right)^{\kappa_3\log^i K}
	\leq \left(\frac{\kappa_4}{\log K}\right)^{\kappa_3\log^i K} \leq \frac{\kappa_5}{K}	
\end{align}
for some $\kappa_4>0$ and $\kappa_5>0$. Therefore we have
\begin{align}
	\alpha_i \leq \alpha_{i-1} + \frac{\kappa_5}{K}
\end{align}
and thus the number of variable nodes exposed until the $i$-th iteration can be bounded by $\log^i K$ with high probability $\Prob{V_i > \kappa_1 \log^i K} \leq {O}\left(\frac{1}{K}\right)$. Similar technique can be used to show that the tail bound for the check nodes is $\Prob{C_i > \kappa_1 \log^i K} \leq {O}\left(\frac{1}{K}\right)$.

It has been shown that the number of nodes is well bounded by ${O}(\log^i K)$, now we proceed to show the tree-like neighborhood of our graph ensemble by induction. Assuming that  the neighborhood $\mathcal{N}_{\textrm{e}}^{2i}$ at the $i$-th iteration ($i<i_\star$) is tree-like, we prove that $\mathcal{N}_{\textrm{e}}^{2(i+1)}$ is tree-like with high probability. 

First of all, we examine the neighborhood $\mathcal{N}_{\textrm{e}}^{2i+1}$. Assume that $t$ additional edges have been revealed at this level without forming a cycle. The probability that the next edge from a variable node does not create a cycle is the probability that it is connected to one of the check nodes that are not already included in the tree, which is lower bound by $1-C_{i_\star}/(\eta K)$. Therefore, given that $\mathcal{N}_{\textrm{e}}^{2i}$ is tree-like, the probability that $\mathcal{N}_{\textrm{e}}^{2i+1}$ is tree-like is lower bounded by 
\begin{align}
	\left(1-\frac{C_{i_\star}}{\eta K}\right)^{C_{i+1}-C_i}.
\end{align}
By similar reasoning, given that $\mathcal{N}_{\textrm{e}}^{2i+1}$ is tree-like, the probability that $\mathcal{N}_{\textrm{e}}^{2(i+1)}$ is tree-like is lower bounded by 
\begin{align}
	\left(1-\frac{V_{i_\star}}{K}\right)^{V_{i+1}-V_i}.
\end{align}
Therefore, the probability that $\mathcal{N}_{\textrm{e}}^{2(i+1)}$ is tree-like is lower bounded by 
\begin{align*}
	\left(1-\frac{C_{i_\star}}{\eta K}\right)^{C_{i_\star}}\left(1-\frac{V_{i_\star}}{K}\right)^{V_{i_\star}}
	&\geq 1- \left(\frac{V_{i_\star}^2}{K}+\frac{C_{i_\star}^2}{\eta K}\right)
	\geq 1- {O}\left(\frac{\log^i K}{K}\right).
\end{align*}
Therefore the probability of not being tree-like is upper bounded by 
\begin{align}
	\Prob{\mathcal{T}_{i}^c}< c_0 \frac{\log^i K}{K}
\end{align}
for some absolute constant $c_0>0$.
\end{proof}

\subsubsection{Proof of Concentration to Mean by Large Deviation Analysis}
Now it remains to show the concentration of $Z_{i}$ around its mean $\mathbb{E}[Z_{i}]$. According to \eqref{sum_edges}, the number of remaining edges is a sum of random variables $Z_{i}=\sum_{e=1}^{K C} Z_{\rm e}^{(i)}$ while summands $Z_{\rm e}^{(i)}$ are not independent with each other. Therefore, to show the concentration, we use a standard martingale argument and Azuma's inequality provided in \cite{richardson2001capacity} with some modifications to account for the irregular degrees of the right nodes. 

Suppose that we expose the whole set of $E=K C$ edges of the graph one at a time. We let
\begin{align}
	Y_\ell=\mathbb{E}\left[Z_{i}|Z_1^{(i)},\cdots,Z_\ell^{(i)}\right],\quad \ell=1,\cdots, K C.
\end{align}	 
By definition, $Y_0,Y_1,\cdots,Y_{K C}$ are a Doob's martingale process, where $Y_0=\mathbb{E}[Z_{i}]$ and $Y_{K C}=Z_{i}$. To use Azuma's inequality, it is required that $|Y_{\ell+1}-Y_\ell|\leq \Delta_\ell$ for some $\Delta_\ell>0$.  If the variable node has a regular degree $d_V$ and the check node has a regular degree $d_C$, then \cite{richardson2001capacity} shows that $\Delta_\ell = 8(d_V d_C)^i$ with $i$ being the number of peeling iterations. However, although we have a regular left degree $d_V=C$ in our graph ensemble $\mathcal{G}(K,\eta,C,\{\mathbf{M}_c\}_{c\in[C]})$, the check node degree is not regular with degree $d_C$ and therefore requires further analysis.

\subsubsection*{Proof of Finite Difference $\Delta_\ell$}
To prove that the difference $\Delta_\ell$ is finite for check node degrees with Poisson distributions, we first prove that the degree of all the check nodes can be upper bounded by $d_C\leq {O}(K^{\frac{2}{4i+1}})$ with probability\footnote{Let $X$ be a Poisson variable with parameter $\lambda$, then the following holds
\begin{align*}
	\Prob{X>c K^{\frac{2}{4i+1}} }
	&\leq
	\left(\frac{e\lambda}{c K^{\frac{2}{4i+1}}} \right)^{c K^{\frac{2}{4i+1}}}
	\leq c_1 \exp\left(-c_2 K^{\frac{2}{4i+1}} \right)
\end{align*}
for some $c_1$ and $c_2$.} at least $$c_1K\exp\left(-c_2 K^{\frac{2}{4i+1}} \right)$$ for some constants $c_1$ and $c_2$. Let $\mathcal{B}$ be the event that at least one check node has more than ${O}\left(K^{\frac{2}{4i +1}}\right)$ edges, then for some $c_3>0$ we have
\begin{align}
	\Prob{\mathcal{B}} < c_3K\exp\left(-c_2 K^{\frac{2}{4i+1}} \right).
\end{align}
by applying a union bound on all the $R=\eta K$ check nodes of the graphs from $\mathcal{G}(K,\eta,C,\{\mathbf{M}_c\}_{c\in[C]})$. As a result, under the complement event $\mathcal{B}^c$, we have
\begin{align}
	\Delta_\ell^2 = {O}\left(K^{\frac{4i}{4i+1}}\right).
\end{align}

\subsubsection*{Large Deviation by Azuma's Inequality}
For any given $\varepsilon>0$, the tail probability of the event $Z_{i}>K C\varepsilon$ can be computed as
\begin{align*}
	\Prob{\left|Z_{i}-\mathbb{E}[Z_{i}]\right| > \frac{K C \varepsilon}{2}}
	&\leq
	\Prob{\left|Z_{i}-\mathbb{E}[Z_{i}]\right| > \frac{K C \varepsilon}{2}\Big |\mathcal{B}^c} + \Prob{\mathcal{B}}\\
	&\leq 2\exp\left(-\frac{K^2 C^2 \varepsilon^2/4}{2\sum_{\ell=1}^{K C} \Delta_\ell^2 }\right) + c_3K\exp\left(-c_2 K^{\frac{2}{4i+1}} \right)\\
	&\leq 2\exp\left(-c_4 \varepsilon^2 K^{\frac{1}{4i+1}}\right),
\end{align*}
where $c_4$ is some constant depending on $C$, $\eta$ and all the other constants $c_1,c_2,c_3$. This concludes our proof for \eqref{concentration_DE}.

\subsection{Proof of Expander Graphs}\label{sec:expander_graph}
Given an arbitrary subset of left nodes $\mathcal{S}$ of size $|\mathcal{S}|=s$ with less than $s/2$ neighbors for all $C$ subsampling groups. The probability of this event can be obtained readily for any size $v$ as 
\begin{align}
	\Prob{{\mathcal{S}}} 
	&\leq {K \choose s} \prod_{c=1}^C { \eta K \choose s/2} \left(\frac{s}{2 \eta K}\right)^{s}\\
	& = {K \choose s}  { \eta K \choose s/2}^C \left(\frac{s}{2 \eta K}\right)^{Cs},
\end{align}
where we have used the fact that the number of check nodes is $\eta K$. Using the inequality ${a \choose b}\leq (a e/b)^b$, we have
\begin{align}
	\Prob{{\mathcal{S}}} 
	&\leq \left(\frac{Ke}{s}\right)^s\left(\frac{\eta K e}{s/2}\right)^{Cs/2}  \left(\frac{s}{2 \eta K}\right)^{Cs}
	=\left(\frac{s}{K}\right)^{s\left(\frac{C}{2}-1\right)} c^s,
\end{align}
where $c=e(e/\eta)^{C/2}$ is some constant. Clearly, as long as $C/2-1\geq 1/2$ such that $C\geq 3$, we can further bound the probability as
\begin{align}
	\Prob{{\mathcal{S}}} 
	&\leq \left(\frac{s c^2}{K}\right)^{s/2}.
\end{align}
It can be seen that the probability of not forming an expander depends on the size of the remaining subset $|\mathcal{S}|=s$. Now we examine two extremes with $s={O}(K)$ and $s={O}(1)$, and obtain the following:
\begin{align}
	\Prob{{\mathcal{S}}} 
	=
	\begin{cases}
		e^{-K \log\left(\frac{1}{\varepsilon c^2}\right)}, & s=\varepsilon K~\textrm{with}~\varepsilon <1/(2c^2)\\
		{O}\left(\frac{1}{K}\right), & s={O}(1).
	\end{cases}
\end{align}
Clearly, the bottleneck event is when the graph is left with $s={O}(1)$ variable nodes, which happens also with probability approaching zero asymptotically in $K$. Therefore, the random graphs from the ensemble are good expanders with probability at least $1-{O}(1/K)$.	

\newpage
\section{Proof of Proposition \ref{prop_BSC}}\label{proof_prop_BSC}
Given a single-ton bin with an index-value pair $(\mathbf{k},X[\mathbf{k}])$, 
\begin{align}
	U_p 
	&= |X[\mathbf{k}]|(-1)^{\ip{\mathbf{d}_p}{\mathbf{k}}\oplus\sgn{X[\mathbf{k}]}} + W_p,\quad p \in [P],
\end{align}
it is clear that $\sgn{U_p}=\ip{\mathbf{d}_p}{\mathbf{k}}\oplus\sgn{X[\mathbf{k}]}\oplus 1$ whenever the noise $W_p$ is sufficiently large such that it crosses over $X[\mathbf{k}](-1)^{\ip{\mathbf{d}_p}{\mathbf{k}}}$. Clearly, this is a random event and we can model it with some Bernoulli variable $Z_p\in\{0,1\}$ with some probability $p_Z$
\begin{align}
	\sgn{U_p}
	&= \ip{\mathbf{d}_p}{\mathbf{k}}\oplus\sgn{X[\mathbf{k}]}\oplus Z_p.
\end{align}
The exact parameter $p_Z$ of the Bernoulli random variable $Z_p$ can be found by studying the tail events that trigger the flipping, but here for simplicity we directly upper bound it as follows
\begin{align}
	p_Z \leq \Pe \defn \Prob{|W_p|>|X[\mathbf{k}]|} \leq e^{-\frac{|X[\mathbf{k}]|^2}{2 N/B\sigma^2}} = e^{-\frac{\eta}{2}\mathrm{SNR}}.
\end{align}

\section{Proof of Theorem \ref{peeling-decoder-RBI}: Peeling Decoder using a Robust Bin Detector}\label{sec:RBI_perf_analysis}

Let ${E}_{\rm bin}$ be the event where the robust bin detector makes a mistake in the ${O}(K)$ peeling iterations. If the error probability of the robust bin detector described in Section \ref{sec:robust_bin_detection} satisfies
\begin{align}\label{P_RBI_upperbound}
	\Prob{{E}_{\rm bin}} = {O}\left(\frac{1}{K}\right),
\end{align}
then the result directly follows from the Bayes rule:
\begin{align*}
	\Pf &= \Prob{\supp{\widehat{\mathbf{X}}}\neq\supp{\mathbf{X}}\big|{E}_{\rm bin}^c}\Prob{{E}_{\rm bin}^c}
	+\Prob{\supp{\widehat{\mathbf{X}}}\neq\supp{\mathbf{X}}\big|{E}_{\rm bin}}\Prob{{E}_{\rm bin}}\\
	&\leq \Prob{\supp{\widehat{\mathbf{X}}}\neq\supp{\mathbf{X}}\big|{E}_{\rm bin}^c}+\Prob{{E}_{\rm bin}} = O\left(1/K\right),
\end{align*}
where the first term in the last inequality is obtained from Theorem \ref{thm_peeling_decoder} for the peeling decoder with an oracle such that the event $E_0^c$ holds. Therefore, it remains to show that \eqref{P_RBI_upperbound} holds. The main idea is to analyze the error probability of making at least an error on any bin observation, followed by a union bound on all the bin observation. Let the error event in any bin $j$ as $E_j$, then we have the following union bound across $\eta K$ bin observation vectors as well as $C K$ iterations\footnote{The number of iterations is taken to be the worst case where at each iteration only one edge is peeled off.}
\begin{align}
	\Prob{{E}_{\rm bin}} \leq CK\bigcup_{j=1}^{\eta K} \Prob{E_j},
\end{align}
where $C$ is the left degree of the regular ensemble $\mathcal{G}(K,\eta,C,\{\mathbf{M}_c\}_{c\in[C]})$ . Without loss of generality, we drop the bin index $j$ and use a union bound over all bins such that
\begin{align}
	\Prob{{E}_{\rm bin}} \leq \eta C K^2\Prob{ E},
\end{align}
where $\Prob{ E}$ is the error probability for an arbitrary bin. It can be seen that due to the union bounds, it is required that $\Prob{ E}\leq {O}(1/K^3)$ such that $\Prob{{E}_{\rm bin}}\leq {O}(1/K)$. 

In the following, we prove that $\Prob{ E}\leq {O}(1/K^3)$ holds using the generic model in Proposition \ref{prop_meas.bin.model}. Since there are different types of errors, thus in the following analysis $\bdsb{\alpha}$ in \eqref{equiv.model.bin_simplified} is fixed as a zero-ton, single-ton or multi-ton respectively for each class of errors. 

\begin{defi}\label{def_error_category}
The error probability $\Prob{ E}$ for an arbitrary bin can be upper bounded as
\begin{align}
	\Prob{ E} 
	&\leq \sum_{\mathcal{F}\in\{\mathcal{H}_{\textrm{Z}},\mathcal{H}_{\textrm{M}}\}}\Prob{\mathcal{F}\leftarrow\mathcal{H}_{\textrm{S}}(\mathbf{k},X[\mathbf{k}])}+ \sum_{\mathcal{F}\in\{\mathcal{H}_{\textrm{Z}},\mathcal{H}_{\textrm{M}}\}}\Prob{\mathcal{H}_{\textrm{S}}(\mathbf{k},X[\mathbf{k}])\leftarrow\mathcal{F}}\\
	&~~~~~~+\Prob{\mathcal{H}_{\textrm{S}}(\widehat{\mathbf{k}},\widehat{X}[\widehat{\mathbf{k}}])\leftarrow\mathcal{H}_{\textrm{S}}(\mathbf{k},X[\mathbf{k}])}
\end{align}
where $\mathcal{F}$ is either a zero-ton $\mathcal{H}_{\textrm{Z}}$  or a multi-ton $\mathcal{H}_{\textrm{M}}$ and
\begin{enumerate}
	\item $\Prob{\mathcal{F}\leftarrow\mathcal{H}_{\textrm{S}}(\mathbf{k},X[\mathbf{k}])}$ is called the {\it missed verification} rate in which the single-ton verification fails when the ground truth is in fact a single-ton $\mathcal{H}=\mathcal{H}_{\textrm{S}}(\mathbf{k},X[\mathbf{k}])$ for some $\mathbf{k}\in\GF^n$ and $X[\mathbf{k}]$.  
	\item $\Prob{\mathcal{H}_{\textrm{S}}(\mathbf{k},X[\mathbf{k}])\leftarrow\mathcal{F}}$ is called the {\it false verification} rate in which the single-ton verification is passed for some single-ton $\mathcal{H}=\mathcal{H}_{\textrm{S}}(\widehat{\mathbf{k}},\widehat{X}[\widehat{\mathbf{k}}])$ with an index-value pair $(\widehat{\mathbf{k}},\widehat{X}[\widehat{\mathbf{k}}])$ when the ground truth is $\mathcal{F}\in\{\mathcal{H}_{\textrm{Z}},\mathcal{H}_{\textrm{M}}\}$.
	\item $\Prob{\mathcal{H}_{\textrm{S}}(\widehat{\mathbf{k}},\widehat{X}[\widehat{\mathbf{k}}])\leftarrow\mathcal{H}_{\textrm{S}}(\mathbf{k},X[\mathbf{k}])}$ is called the {\it crossed verification} rate in which a single-ton with a wrong index-value pair $\widehat{\mathbf{k}}\neq \mathbf{k}, \widehat{X}[\widehat{\mathbf{k}}] \neq X[\mathbf{k}]$ passes the single-ton verification when the ground truth is a single-ton with an index-value pair $\mathcal{H}=\mathcal{H}_{\textrm{S}}(\mathbf{k},X[\mathbf{k}])$ for some $k\neq \widehat{k}$.
\end{enumerate}
\end{defi}
The false verification rate, missed verification rate and crossed verification rate for the near-linear time and sub-linear time recovery schemes are given in the following propositions. 

\begin{prop}[False Verification Rate]\label{prop_false_detection}
For any $0<\gamma<\mathsf{SNR}/2$, the false verification rate for each bin hypothesis can be upper bounded as follows:
\begin{align*}
		&\Prob{\mathcal{H}_{\textrm{S}}(\widehat{\mathbf{k}},\widehat{X}[\widehat{\mathbf{k}}])\leftarrow\mathcal{H}_{\textrm{Z}}}<
		e^{-\frac{P_1}{4}\left(\sqrt{1+2\gamma}-1\right)^2}\\
		&\Prob{\mathcal{H}_{\textrm{S}}(\widehat{\mathbf{k}},\widehat{X}[\widehat{\mathbf{k}}])\leftarrow\mathcal{H}_{\textrm{M}}}< 
		e^{-\frac{P_1}{4}\frac{\gamma^2}{1+4\gamma}} + Ne^{-\frac{\varepsilon}{4}\left(1-\frac{2\gamma\nu^2}{\rho^2}\right)P_1},
\end{align*}
where $P_1$ is the number of the random offsets in the NSO-SPRIGHT and the SO-SPRIGHT algorithm.
\end{prop}
\begin{proof}
	See Appendix \ref{proof_prop_false_detection}.
\end{proof}

\begin{prop}[Missed Verification Rate]\label{prop_miss_detection}
For any $0<\gamma<\mathsf{SNR}/2$, the missed verification rate for each bin hypothesis can be upper bounded as follows:
\begin{align*}
		&\Prob{\mathcal{H}_{\textrm{Z}}\leftarrow\mathcal{H}_{\textrm{S}}(\mathbf{k},X[\mathbf{k}])}
		<e^{-\frac{P_1}{4}\frac{\left(\rho^2/\nu^2-\gamma\right)^2}{1+2\rho^2/\nu^2}}\\
		&\Prob{\mathcal{H}_{\textrm{M}}\leftarrow\mathcal{H}_{\textrm{S}}(\mathbf{k},X[\mathbf{k}])}
		<
		2e^{-\frac{\rho^2}{2\nu^2}P_1}
		+
		\begin{cases}
		2n e^{-\frac{(1-2\theta)^2}{8}P_1}, & \textrm{NSO-SPRIGHT}\\
		2e^{-\frac{(\beta/\Pe-1)^2}{3}P_3}+2e^{-\frac{(1-2\Pe)^2}{8}P_2}, & \textrm{SO-SPRIGHT}.
		\end{cases}
\end{align*}
where $P_1$ is the number of random offsets in the NSO-SPRIGHT and the SO-SPRIGHT algorithm, while $P_2$ and $P_3$ are the numbers of the zero offsets and coded offsets in the SO-SPRIGHT algorithm.
\end{prop}
\begin{proof}
	See Appendix \ref{proof_prop_miss_detection}.
\end{proof}

\begin{prop}[Crossed Verification Rate]\label{prop_cross_detection}
For any $0<\gamma<\mathsf{SNR}/2$, the false verification rate for each bin hypothesis can be upper bounded as follows:
\begin{align*}
		&\Prob{\mathcal{H}_{\textrm{S}}(\widehat{\mathbf{k}},\widehat{X}[\widehat{\mathbf{k}}])\leftarrow\mathcal{H}_{\textrm{S}}(\mathbf{k},X[\mathbf{k}])}
		<
		e^{-\frac{P_1}{4}\frac{\gamma^2}{1+4\gamma}}
		+
		2Ne^{-\frac{1}{8}\left(1-\frac{\gamma \nu^2}{\rho^2}\right)^2P_1},
\end{align*}
where $P_1$ is the number of random offsets in the NSO-SPRIGHT and the SO-SPRIGHT algorithm.
\end{prop}
\begin{proof}
	See Appendix \ref{proof_prop_cross_detection}.
\end{proof}

%%%%%%%%%%%%%%%%%%%%%%%%%%%%%%%%%%%%%%%%%%%%%

Since all the error probabilities decay exponentially with respect to $\{P_i\}_{i=1}^3$, it is now clear that if $P_i$ is chosen as $P_i={O}(n)={O}(\log N)$, the probability can be bounded as $\Prob{E} = {O}(1/N^3)$ such that $\Prob{{E}_{\rm bin}} = {O}(1/K)$.

%%%%%%%%%%%%%%%
%%%%%%%%%%%%%%%

\section{Proof of False Verification Rates in Proposition \ref{prop_false_detection}}\label{proof_prop_false_detection}
The false verification events occur when the ground truth is not a single-ton, and therefore, the probabilities can be obtained using the bin observation model
\begin{align}
	\bdsb{U} = \mathbf{S}\bdsb{\alpha} +\bdsb{W}
\end{align}
with $\bdsb{\alpha}$ being a zero-ton $\bdsb{\alpha}=\mathbf{0}$ or a multi-ton $\left|\supp{\bdsb{\alpha}}\right| > 1$. With a slight abuse of notation, here $\mathbf{S}\in\{\pm 1\}^{P_1\times N}$ is the codebook associated with the $P_1$ fully random offsets in the NSO-SPRIGHT and SO-SPRIGHT algorithm.

\subsection{Detecting a Zero-ton as a Single-ton}
By definition, the probability of detecting a zero-ton as a single-ton can be upper bounded by the probability of a zero-ton failing the zero-ton verification:
\begin{align*}
	&\Prob{\mathcal{H}_{\textrm{S}}(\widehat{\mathbf{k}},\widehat{X}[\widehat{\mathbf{k}}])\leftarrow\mathcal{H}_{\textrm{Z}}}
	\leq \Prob{\frac{1}{P_1}\left\|\bdsb{W}\right\|^2\geq (1+\gamma)\nu^2}.
\end{align*}
Since $\bdsb{W}\sim\mathcal{N}(\mathbf{0},\nu^2\mathbf{I})$, we can bound this probability using Lemma \ref{general_tail}:
\begin{align*}
	\Prob{\frac{1}{P_1}\left\|\bdsb{W}\right\|^2\geq (1+\gamma)\nu^2}
	\leq
	e^{-\frac{P_1}{4}\left(\sqrt{1+2\gamma}-1\right)^2}.
\end{align*}

\subsection{Detecting a Multi-ton as a Single-ton}
By definition, the error probability can be evaluated under the multi-ton model when it passes the single-ton verification step for some index-value pair $(\widehat{\mathbf{k}},\widehat{X}[\widehat{\mathbf{k}}])$
\begin{align*}
	&\Prob{\mathcal{H}_{\textrm{S}}(\widehat{\mathbf{k}},\widehat{X}[\widehat{\mathbf{k}}])\leftarrow\mathcal{H}_{\textrm{M}}}
	=\Prob{\frac{1}{P_1}\left\|\bdsb{U}-\widehat{X}[\widehat{\mathbf{k}}]\mathbf{s}_{\widehat{\mathbf{k}}}\right\|^2 \leq  (1+\gamma)\nu^2}
\end{align*}
given some multi-ton observation $\bdsb{U}=\mathbf{S}\bdsb{\alpha}+\bdsb{W}$. Letting $\bdsb{g}=\mathbf{S}(\bdsb{\alpha}-\widehat{X}[\widehat{\mathbf{k}}]\mathbf{e}_{\widehat{\mathbf{k}}})$ and $\bdsb{v}=\bdsb{W}$, we compute this probability according to the total probability law as follows
\begin{align}\label{inequality_three}
	&\Prob{\frac{1}{P_1}\left\|\bdsb{g}+\bdsb{v}\right\|^2\leq (1+\gamma)\nu^2}\\
	&=\Prob{\frac{1}{P_1}\left\|\bdsb{g}+\bdsb{v}\right\|^2 \leq  (1+\gamma)\nu^2\Big|\frac{\left\|\bdsb{g}\right\|^2}{P_1}\geq 2\gamma\nu^2}
	\times\Prob{\frac{\left\|\bdsb{g}\right\|^2}{P_1}\geq 2\gamma\nu^2}\nonumber\\
	&~~~+\Prob{\frac{1}{P_1}\left\|\bdsb{g}+\bdsb{v}\right\|^2 \leq  (1+\gamma)\nu^2\Big|\frac{\left\|\bdsb{g}\right\|^2}{P_1}\leq 2\gamma\nu^2}
	\times\Prob{\frac{\left\|\bdsb{g}\right\|^2}{P_1}\leq 2\gamma\nu^2}\nonumber\\
	&\leq \Prob{\frac{1}{P_1}\left\|\bdsb{g}+\bdsb{v}\right\|^2 \leq  (1+\gamma)\nu^2\Big|\frac{\left\|\bdsb{g}\right\|^2}{P_1}\geq 2\gamma\nu^2}
	+\Prob{\frac{\left\|\bdsb{g}\right\|^2}{P_1}\leq 2\gamma\nu^2},\nonumber
\end{align}
where the first term is basically the single-ton verification error rate when the multi-ton has sufficiently large energy while the second term is the probability of any multi-ton not having sufficiently large energy. In the following, we bound these two probabilities separately with exponential tails.

We start from the single-ton verification error rate when the multi-ton has sufficiently large energy, or namely $\Prob{\frac{1}{P_1}\left\|\bdsb{g}+\bdsb{v}\right\|^2 \leq  (1+\gamma)\nu^2\Big|\frac{\left\|\bdsb{g}\right\|^2}{P_1}\geq 2\gamma\nu^2}$. Lemma \ref{general_tail} can be directly used here by letting $\tau_2=(1+\gamma)\nu^2$. Note that the first term is conditioned on the event where ${\left\|\bdsb{g}\right\|^2}/{P_1} \geq 2\gamma\nu^2$, therefore the minimum normalized non-centrality parameter can be obtained as $\nu_{\min} = \min_{\bdsb{g}}{\left\|\bdsb{g}\right\|^2}/{P_1\nu^2} = 2\gamma$. Clearly, the condition for the threshold in \eqref{tau_requirement2} holds for Corollary \ref{general_tail_worst}, and thus the first term can be bounded accordingly as
\begin{align}\label{inequality_one}
	\Prob{\frac{1}{P_1}\left\|\bdsb{g}+\bdsb{v}\right\|^2 \leq  (1+\gamma)\nu^2\Big|\frac{\left\|\bdsb{g}\right\|^2}{P_1}\geq 2\gamma\nu^2}
	&\leq e^{-\frac{P_1}{4}\frac{\gamma^2}{1+4\gamma}}.
\end{align}

Now we examine the probability of a multi-ton not having sufficiently large energy, or namely $\Prob{\frac{1}{P}\left\|\bdsb{g}\right\|^2\leq 2\gamma\nu^2}$. Letting $\bdsb{\beta}=\bdsb{\alpha} -\widehat{X}[\widehat{\mathbf{k}}]\mathbf{e}_{\widehat{\mathbf{k}}}$, we have $\bdsb{g}=\mathbf{S}\bdsb{\beta}$ and thus
\begin{align}\label{sum_of_random_var}
	\Prob{\frac{\left\|\bdsb{g}\right\|^2}{P_1}\leq 2\gamma\nu^2}
	=\Prob{\frac{\left\|\mathbf{S}\bdsb{\beta}\right\|^2}{P_1}\leq 2\gamma\nu^2}.
	%=\Prob{\frac{1}{P} \sum_{p\in[P]} |s[p]|^2 \leq  2\gamma\nu^2}.
\end{align}
Denoting the support of $\mathcal{L} \defn \supp{\bdsb{\beta}}$, we bound this probability with respect to the following two multi-ton scenarios:
\begin{itemize}
	\item $\left|\mathcal{L}\right|=L={O}(1)$ where the multi-ton size is a constant. Note that $\left\|\mathbf{S}\bdsb{\beta}\right\|^2=\bdsb{\beta}_{\mathcal{L}}^T\mathbf{S}_{\mathcal{L}}^T\mathbf{S}_{\mathcal{L}}\bdsb{\beta}_{\mathcal{L}}$ where $\mathbf{S}_{\mathcal{L}}$ is the sub-matrix consisting of the columns $\mathbf{k}\in\mathcal{L}$ and $\bdsb{\beta}_{\mathcal{L}}$ is the sub-vector containing the elements in the set $\mathbf{k}\in\mathcal{L}$. Then, we have
	\begin{align}
		\lambda_{\min}\left(\mathbf{S}_{\mathcal{L}}^T\mathbf{S}_{\mathcal{L}}\right)\left\|\bdsb{\beta}_{\mathcal{L}}\right\|^2
		\leq \bdsb{\beta}_{\mathcal{L}}^T\mathbf{S}_{\mathcal{L}}^T\mathbf{S}_{\mathcal{L}}\bdsb{\beta}_{\mathcal{L}}
		\leq \lambda_{\max}\left(\mathbf{S}_{\mathcal{L}}^T\mathbf{S}_{\mathcal{L}}\right)\left\|\bdsb{\beta}_{\mathcal{L}}\right\|^2.
	\end{align}
	Using $\left\|\bdsb{\beta}_{\mathcal{L}}\right\|^2\geq L\rho^2$, the probability can be bounded as
	\begin{align}\label{prob_using_quadform}
		\Prob{\frac{\left\|\bdsb{g}\right\|^2}{P_1}\leq 2\gamma\nu^2}
		&=\Prob{\frac{\left\|\mathbf{S}_{\mathcal{L}}\bdsb{\beta}_{\mathcal{L}}\right\|^2}{P_1}\leq 2\gamma\nu^2}\\
		&\leq \Prob{ \lambda_{\min}\left(\frac{1}{P_1}\mathbf{S}_{\mathcal{L}}^T\mathbf{S}_{\mathcal{L}}\right) \leq  \frac{2 \gamma\nu^2}{\left\|\bdsb{\beta}_{\mathcal{L}}\right\|^2}}\\
		&= \Prob{ \lambda_{\min}\left(\frac{1}{P_1}\mathbf{S}_{\mathcal{L}}^T\mathbf{S}_{\mathcal{L}}\right) \leq  \frac{2 \gamma\nu^2}{L\rho^2}}.
	\end{align}	
	\begin{lem}\label{lem_mc}
	Denote the mutual coherence of the codebook $\mathbf{S}$ by $\mu \defn \max_{\mathbf{k}\neq\mathbf{m}} \frac{1}{P_1}\left|\mathbf{s}_{\mathbf{k}}^T\mathbf{s}_{\mathbf{m}}\right|$. Then for some given $\mu_0>0$, we have $\Prob{\mu\geq \mu_0} \leq 2Ne^{-\frac{\mu_0^2}{2}P_1}$.
	\end{lem}
	\begin{proof}
	Since $\mathbf{S}$ contains i.i.d. Rademacher entries, the result follows by a simple Hoeffding bound.
	\end{proof}
	According to the Gershgorin Circle Theorem
	\begin{align}
		\lambda_{\min}\left(\frac{1}{P_1}\mathbf{S}_{\mathcal{L}}^T\mathbf{S}_{\mathcal{L}}\right)\geq 1-L\mu
	\end{align}
	we have the following bound
	\begin{align}\label{prob_using_MC}
		\Prob{ \lambda_{\min}\left(\frac{1}{P_1}\mathbf{S}_{\mathcal{L}}^T\mathbf{S}_{\mathcal{L}}\right) \leq  \frac{2 \gamma\nu^2}{L\rho^2}}
		&\leq
		\Prob{ 1-L\mu \leq  \frac{2 \gamma\nu^2}{L\rho^2}}\\
		&=
		\Prob{ \mu \geq \frac{1}{L} \left(1- \frac{2 \gamma\nu^2}{L\rho^2}\right)}.
	\end{align}	
	By letting $\mu_0=\frac{1}{L} \left(1- \frac{2 \gamma\nu^2}{L\rho^2}\right)$, we can upper bound this probability using Lemma \ref{lem_mc} as
	\begin{align}\label{mc_example}
		\Prob{\frac{\left\|\bdsb{g}\right\|^2}{P_1}\leq 2\gamma\nu^2}
		\leq
		\Prob{ \mu \geq \frac{1}{L} \left(1- \frac{2 \gamma\nu^2}{L\rho^2}\right)}
		\leq 
		2Ne^{-\frac{1}{2L^2}\left(1-\frac{2\gamma \nu^2}{L\rho^2}\right)^2P_1},
	\end{align}
	which holds if $\gamma<L\rho^2/2\nu^2$.
	
	\item $\left|\mathcal{L}\right|=L={\omega}(1)$ where the multi-ton size is not a constant and grows asymptotically with respect to $K$. As a result, the vector of random variables $\bdsb{g}=\mathbf{S}_{\mathcal{L}}\bdsb{\beta}_{\mathcal{L}}$ becomes asymptotically Gaussian due to the central limit theorem with zero mean and a covariance
\begin{align}
	\mathbb{E}\left[\bdsb{g}\bdsb{g}^T\right] 
	&= \mathbb{E}\left[\mathbf{S}_{\mathcal{L}}\bdsb{\beta}_{\mathcal{L}}\bdsb{\beta}_{\mathcal{L}}^T\mathbf{S}_{\mathcal{L}}^T\right]
	= L\rho^2\mathbf{I}.
\end{align}
Therefore, from Lemma \ref{general_tail} and Corollary \ref{general_tail_worst} we have
\begin{align*}
	\Prob{\frac{\left\|\bdsb{g}\right\|^2}{P}\leq 2\gamma\nu^2}
	&\leq e^{-\frac{1}{4}\left(1-\frac{2\gamma\nu^2}{L\rho^2}\right)P_1},
\end{align*}
which holds if $\gamma<L\rho^2/2\nu^2$.
\end{itemize}
Finally, as long as $0<\gamma<\rho^2/2\nu^2$, for any multi-ton there exists some constant $\varepsilon>0$ such that
\begin{align*}
	\Prob{\frac{\left\|\bdsb{g}\right\|^2}{P_1}\leq 2\gamma\nu^2}
	&\leq Ne^{-\frac{\varepsilon}{4}\left(1-\frac{2\gamma\nu^2}{\rho^2}\right)P_1}.
\end{align*}

%%%%%%%%%%%%%%%
%%%%%%%%%%%%%%%

\section{Proof of Missed Verification Rates in Proposition \ref{prop_miss_detection}}\label{proof_prop_miss_detection}
The missed verification events occur when the ground truth is a single-ton, and therefore, the probabilities is obtained using the bin observation model with some index-value pair $(\mathbf{k},X[\mathbf{k}])$
\begin{align}
	\bdsb{U} = X[\mathbf{k}] \mathbf{s}_{\mathbf{k}} +\bdsb{W}.
\end{align}
With a slight abuse of notation, here $\mathbf{S}$ is the codebook associated with the fully random offsets in our designs.

\subsection{Detecting a Single-ton as a Zero-ton}
By definition, the probability of detecting a single-ton as a zero-ton can be upper bounded by the probability of a single-ton passing the zero-ton verification:
\begin{align*}
	&\Prob{\mathcal{H}_{\textrm{Z}}\leftarrow\mathcal{H}_{\textrm{S}}(\mathbf{k},X[\mathbf{k}])}
	\leq \Prob{\frac{1}{P}\left\|X[\mathbf{k}] \mathbf{s}_{\mathbf{k}}+\bdsb{W}\right\|^2\leq (1+\gamma)\nu^2}.
\end{align*}
Since $\bdsb{W}\sim\mathcal{N}(\mathbf{0},\nu^2\mathbf{I})$, we can bound this probability using Lemma \ref{general_tail} by letting $\bdsb{g}=X[\mathbf{k}]\mathbf{s}_{\mathbf{k}}$ and $\bdsb{v}=\bdsb{W}$:
\begin{align*}
	\Prob{\frac{1}{P}\left\|X[\mathbf{k}] \mathbf{s}_{\mathbf{k}}+\bdsb{W}\right\|^2
	\leq (1+\gamma)\nu^2}
	\leq
	e^{-\frac{P_1}{4}\frac{\left(\rho^2/\nu^2-\gamma\right)^2}{1+2\rho^2/\nu^2}},
\end{align*}
which holds as long as $\gamma<\rho^2/\nu^2$.

\subsection{Detecting a Single-ton as a Multi-ton}
By definition, the error probability can be evaluated under the single-ton model when it fails the single-ton verification step for some index-value pair $(\widehat{\mathbf{k}},\widehat{X}[\widehat{\mathbf{k}}])$
\begin{align*}
	&\Prob{\mathcal{H}_{\textrm{M}}\leftarrow\mathcal{H}_{\textrm{S}}(\mathbf{k}, X[\mathbf{k}])}
	=\Prob{\frac{1}{P}\left\|\bdsb{U}-\widehat{X}[\widehat{\mathbf{k}}]\mathbf{s}_{\widehat{\mathbf{k}}}\right\|^2 \geq  (1+\gamma)\nu^2}
\end{align*}
given some single-ton observation $\bdsb{U}= X[\mathbf{k}]\mathbf{s}_{\mathbf{k}} +\bdsb{W}$. Since the estimated index-value pair $(\widehat{\mathbf{k}},\widehat{X}[\widehat{\mathbf{k}}])$ may or may not be correct, the above probability can be bounded as:
\begin{align*}
	&\Prob{\frac{1}{P_1}\left\|\bdsb{U}-\widehat{X}[\widehat{\mathbf{k}}]\mathbf{s}_{\widehat{\mathbf{k}}}\right\|^2 \geq  (1+\gamma)\nu^2}\\
	&=\Prob{\frac{1}{P_1}\left\|\bdsb{U}-\widehat{X}[\widehat{\mathbf{k}}]\mathbf{s}_{\widehat{\mathbf{k}}}\right\|^2 \geq  (1+\gamma)\nu^2 \Big|\widehat{X}[\widehat{\mathbf{k}}] \neq X[\mathbf{k}]~\textrm{or}~\widehat{\mathbf{k}}\neq\mathbf{k}}\Prob{\widehat{X}[\widehat{\mathbf{k}}] \neq X[\mathbf{k}]~\textrm{or}~\widehat{\mathbf{k}}\neq\mathbf{k}}\\
	&~~~~+
	\Prob{\frac{1}{P_1}\left\|\bdsb{U}-\widehat{X}[\widehat{\mathbf{k}}]\mathbf{s}_{\widehat{\mathbf{k}}}\right\|^2 \geq  (1+\gamma)\nu^2 \Big|\widehat{X}[\widehat{\mathbf{k}}] = X[\mathbf{k}]~\textrm{and}~\widehat{\mathbf{k}} = \mathbf{k}}\Prob{\widehat{X}[\widehat{\mathbf{k}}] = X[\mathbf{k}]~\textrm{and}~\widehat{\mathbf{k}} = \mathbf{k}}\\
	&\leq
	\Prob{\widehat{X}[\widehat{\mathbf{k}}] \neq X[\mathbf{k}]~\textrm{or}~\widehat{\mathbf{k}}\neq\mathbf{k}} +
	\Prob{\frac{1}{P_1}\left\|\bdsb{U}-\widehat{X}[\widehat{\mathbf{k}}]\mathbf{s}_{\widehat{\mathbf{k}}}\right\|^2 \geq  (1+\gamma)\nu^2 \Big|\widehat{X}[\widehat{\mathbf{k}}] = X[\mathbf{k}]~\textrm{and}~\widehat{\mathbf{k}} = \mathbf{k}}.
\end{align*}
It is clear that 
\begin{align}
	&\Prob{\frac{1}{P_1}\left\|\bdsb{U}-\widehat{X}[\widehat{\mathbf{k}}]\mathbf{s}_{\widehat{\mathbf{k}}}\right\|^2 \geq  (1+\gamma)\nu^2 \Big|\widehat{X}[\widehat{\mathbf{k}}] = X[\mathbf{k}]~\textrm{and}~\widehat{\mathbf{k}} = \mathbf{k}}\\
	&=
	\Prob{\frac{1}{P_1}\left\|\bdsb{W}\right\|^2\geq (1+\gamma)\nu^2}
	\leq
	e^{-\frac{P_1}{4}\left(\sqrt{1+2\gamma}-1\right)^2},	
\end{align}
therefore we focus on bounding the first term $\Prob{\widehat{X}[\widehat{\mathbf{k}}] \neq X[\mathbf{k}]~\textrm{or}~\widehat{\mathbf{k}}\neq\mathbf{k}}$. From basic probability laws we have
\begin{align}
	&\Prob{\widehat{X}[\widehat{\mathbf{k}}] \neq X[\mathbf{k}]~\textrm{or}~\widehat{\mathbf{k}}\neq\mathbf{k}}\\
	&\leq
	\Prob{\widehat{X}[\widehat{\mathbf{k}}] \neq X[\mathbf{k}]}+\Prob{\widehat{\mathbf{k}}\neq\mathbf{k}}\\
	&=
	\Prob{\widehat{X}[\widehat{\mathbf{k}}] \neq X[\mathbf{k}]\Big|\widehat{\mathbf{k}}\neq\mathbf{k}}\Prob{\widehat{\mathbf{k}}\neq\mathbf{k}}
	+\Prob{\widehat{X}[\widehat{\mathbf{k}}] \neq X[\mathbf{k}]\Big|\widehat{\mathbf{k}}=\mathbf{k}}\Prob{\widehat{\mathbf{k}}=\mathbf{k}}
	+\Prob{\widehat{\mathbf{k}}\neq\mathbf{k}}\\
	&\leq
	\Prob{\widehat{X}[\widehat{\mathbf{k}}] \neq X[\mathbf{k}]\Big|\widehat{\mathbf{k}}=\mathbf{k}}
	+2\Prob{\widehat{\mathbf{k}}\neq\mathbf{k}}.
\end{align}
The first term is the detection error probability of a BPSK signal with amplitudes $\pm\rho$, and can be bounded as
\begin{align}
	\Prob{\widehat{X}[\widehat{\mathbf{k}}] \neq X[\mathbf{k}]\Big|\widehat{\mathbf{k}}=\mathbf{k}}
	\leq 2e^{-\frac{\rho^2}{2\nu^2}P_1}.
\end{align}
Since the second term $\Prob{\widehat{\mathbf{k}}\neq\mathbf{k}}$ is essentially the error probability of the single-ton search, we prove the following lemmas for different bin detection schemes.

\begin{lem}[\bf Single-ton Search Error Probability of the NSO-SPRIGHT Algorithm]\label{lem_NSO_SPRIGHT}
The single-ton search error probability of the NSO-SPRIGHT algorithm is upper bounded as
\begin{align}\label{P_slow}
	\Prob{\widehat{\mathbf{k}} \neq \mathbf{k}}
	&
	\leq n e^{-\frac{(1-2\theta)^2}{8}P_1}
\end{align}
where $P_1$ is the number of random offsets in the NSO-SPRIGHT design.
\end{lem}
\begin{proof}
	See Appendix \ref{proof_lem_NSO_SPRIGHT}.
\end{proof}

\begin{lem}[\bf Single-ton Search Error Probability of the SO-SPRIGHT Algorithm]\label{lem_SO_SPRIGHT}
The single-ton search error probability of the SO-SPRIGHT algorithm is upper bounded as
\begin{align}\label{P_Fourier}
	\Prob{\widehat{\mathbf{k}} \neq \mathbf{k}}
	&\leq
	e^{-\frac{(\beta/\Pe-1)^2}{3}P_3}+e^{-\frac{(1-2\Pe)^2}{8}P_2},
\end{align}
where $P_1$ is the number of the coded offsets $\mathbf{G}$ and $P_2$ is the number of zero offsets in the SO-SPRIGHT design.
\end{lem}
\begin{proof}
	See Appendix \ref{proof_lem_SO_SPRIGHT}.
\end{proof}

\section{Proof of Crossed Verification Rates in Proposition \ref{prop_cross_detection}}\label{proof_prop_cross_detection}
A crossed verification implies that some wrong index-value pair $(\widehat{\mathbf{k}},\widehat{X}[\widehat{\mathbf{k}}])$ passes the single-ton verification
\begin{align*}
	\Prob{\mathcal{H}_{\textrm{S}}(\widehat{\mathbf{k}},\widehat{X}[\widehat{\mathbf{k}}])\leftarrow\mathcal{H}_{\textrm{S}}(\mathbf{k},\widehat{X}[\mathbf{k}])}
	&=\Prob{\frac{1}{P_1}\left\|\bdsb{U}-\widehat{X}[\widehat{\mathbf{k}}]\mathbf{s}_{\widehat{k}}\right\|^2\leq (1+\gamma)\nu^2}\\
	&=\Prob{\frac{1}{P_1}\left\|X[\mathbf{k}]\mathbf{s}_{\mathbf{k}}-\widehat{X}[\widehat{\mathbf{k}}]\mathbf{s}_{\widehat{\mathbf{k}}}+\bdsb{W}\right\|^2\leq (1+\gamma)\nu^2}.
\end{align*}
Letting $\bdsb{g}=X[\mathbf{k}]\mathbf{s}_{\mathbf{k}}-\widehat{X}[\widehat{\mathbf{k}}]\mathbf{s}_{\widehat{\mathbf{k}}}$, this can be re-written as
\begin{align*}
	\Prob{\mathcal{H}_{\textrm{S}}(\widehat{\mathbf{k}},\widehat{X}[\widehat{\mathbf{k}}])\leftarrow\mathcal{H}_{\textrm{S}}(\mathbf{k},\widehat{X}[\mathbf{k}])}
	=
	\Prob{\frac{1}{P_1}\left\|\bdsb{g}+\bdsb{W}\right\|^2\leq (1+\gamma)\nu^2}.
\end{align*}
Similar to \eqref{inequality_three}, we have
\begin{align*}
	\Prob{\frac{1}{P_1}\left\|\bdsb{g}+\bdsb{W}\right\|^2 \leq  (1+\gamma)\nu^2}
	&\leq \Prob{\frac{1}{P_1}\left\|\bdsb{g}+\bdsb{W}\right\|^2 \leq  (1+\gamma)\nu^2\Big|\frac{\left\|\bdsb{g}\right\|^2}{P_1}\geq 2\gamma\nu^2}
	+\Prob{\frac{\left\|\bdsb{g}\right\|^2}{P_1}\leq 2\gamma\nu^2}.
\end{align*}
Similar to \eqref{inequality_one}, the first term can be bounded as
\begin{align}
	\Prob{\frac{1}{P_1}\left\|\bdsb{g}+\bdsb{W}\right\|^2 \leq  (1+\gamma)\nu^2\Big|\frac{\left\|\bdsb{g}\right\|^2}{P_1}\geq 2\gamma\nu^2}
	&\leq e^{-\frac{P_1}{4}\frac{\gamma^2}{1+4\gamma}}.
\end{align}
Finally, similar to \eqref{mc_example} with $L=2$, the second term $\Prob{\frac{\left\|\bdsb{g}\right\|^2}{P_1}\leq 2\gamma\nu^2}$ can be bounded as
\begin{align}\label{mc_example}
	\Prob{\frac{\left\|\bdsb{g}\right\|^2}{P_1}\leq 2\gamma\nu^2}
	\leq 
	2Ne^{-\frac{1}{8}\left(1-\frac{\gamma \nu^2}{\rho^2}\right)^2P_1},
\end{align}

\section{Proof of Single-ton Search Error Probability in Lemma \ref{lem_NSO_SPRIGHT} and \ref{lem_SO_SPRIGHT}}
\subsection{Single-ton Search in the NSO-SPRIGHT Algorithm}\label{proof_lem_NSO_SPRIGHT}
From the MLE in \eqref{MLE_kq}, the error probability of the single-ton search for the $q$-th bit of $\mathbf{k}$ is 
\begin{align}
	\Prob{\widehat{k}[q]\neq k[q]}
	=
	\Prob{\sum_{p=1}^{P_1} \sgn{U_{p,q}}\oplus \sgn{U_p} \oplus \widehat{k}[q] < \sum_{p=1}^{P_1} \sgn{U_{p,q}}\oplus \sgn{U_p} \oplus k[q]}.
\end{align}
Recall that $\sgn{U_{p,q}}\oplus \sgn{U_p} = k[q]\oplus Z_{p,q}'$ in \eqref{NSO_prob} where $Z_{p,q}'$ is a Bernoulli variable with probability $\theta =2\Pe(1-\Pe)$. Therefore, we have
\begin{align}
	\Prob{\widehat{k}[q]\neq k[q]}
	&=
	\Prob{\sum_{p=1}^{P_1} k[q]\oplus \widehat{k}[q]\oplus Z_{p,q}' < \sum_{p=1}^{P_1} k[q]\oplus k[q] \oplus Z_{p,q}'}\\
	&
	=\Prob{\sum_{p=1}^{P_1} 1\oplus Z_{p,q}' < \sum_{p=1}^{P_1} Z_{p,q}'}.
\end{align}
Noticing that $\sum_{p=1}^{P_1} 1\oplus Z_{p,q}' = P_1 - \sum_{p=1}^{P_1} Z_{p,q}'$, we have
\begin{align}\label{kq_kq}
	\Prob{\widehat{k}[q]\neq k[q]}
	&=
	\Prob{\sum_{p=1}^{P_1} Z_{p,q}' > P_1/2}
	\leq 
	e^{-\frac{(1-2\theta)^2}{8}P_1},
\end{align}
where the inequality follows from the Hoeffding bound. By union bounding over all $n$ bits, we have
\begin{align}
	\Prob{\widehat{\mathbf{k}} \neq \mathbf{k}}
	\leq n e^{-\frac{(1-2\theta)^2}{8}P_1}.
\end{align}

\subsection{Single-ton Search in the SO-SPRIGHT Algorithm}\label{proof_lem_SO_SPRIGHT}
In the general setting, the index is decoded after obtaining the sign $\sgnhat{X[\mathbf{k}]}$. Therefore, we have
\begin{align}
	\Prob{\widehat{\mathbf{k}} \neq \mathbf{k}}
	&=
	\Prob{\widehat{\mathbf{k}} \neq \mathbf{k}\Big|\sgnhat{X[\mathbf{k}]}=\sgn{X[\mathbf{k}]}}\Prob{\sgnhat{X[\mathbf{k}]}=\sgn{X[\mathbf{k}]}}\\
	&~~~ +
	\Prob{\widehat{\mathbf{k}} \neq \mathbf{k}\Big|\sgnhat{X[\mathbf{k}]}\neq\sgn{X[\mathbf{k}]}}\Prob{\sgnhat{X[\mathbf{k}]}\neq\sgn{X[\mathbf{k}]}}\\
	&\leq
	\Prob{\widehat{\mathbf{k}} \neq \mathbf{k}\Big|\sgnhat{X[\mathbf{k}]}}
	+
	\Prob{\sgnhat{X[\mathbf{k}]}\neq\sgn{X[\mathbf{k}]}}.
\end{align}
If the codebook $\mathbf{G}$ with block length $P_3={O}(n)$ has a minimum distance of $\beta P_3$ such that $\beta>\Pe$, the $\mathbf{k}$ fails to be decoded when there are more than $\beta P_3$ sign flips. This can be bounded for the BSC($\Pe$) by the Chern-off bound
\begin{align}
	\Prob{\widehat{\mathbf{k}} \neq \mathbf{k}\Big|\sgnhat{X[\mathbf{k}]}}
	\leq
	e^{-\frac{(\beta/\Pe-1)^2}{3}P_3}.
\end{align}
Since the sign is obtained from $P_2$ sign observations through a majority test if $\Pe<1/2$ and a minority test if $\Pe>1/2$, the error in mistaking the sign can be bounded similarly to \eqref{kq_kq} as
\begin{align}
	\Prob{\sgnhat{X[\mathbf{k}]}\neq\sgn{X[\mathbf{k}]}}
	\leq e^{-\frac{(1-2\Pe)^2}{8}P_2}.
\end{align}

%%%%%%%%%%%
\section{Tail Bounds}\label{sec:general_tail_chi}
Here we derive some tail bounds that are useful in our analysis.

\begin{lem}[Non-central Chi-Square Tail Bounds in \cite{birge2001alternative}]\label{lem_general_tail_chi}
Let $Z\sim\chi_D^2$ be a non-central chi square variable with $D$ degrees of freedom and non-centrality parameter $\theta\geq 0$. Then for all $z\geq 0$, the following tail bounds hold:
\begin{align*}
	&\Prob{Z \geq (D+\theta)+2\sqrt{(D+2\theta) z} + 2z} \leq \exp(-z)\\
	&\Prob{Z \leq (D+\theta)-2\sqrt{(D+2\theta) z}} \leq \exp(-z)	
\end{align*}
\end{lem}

\begin{lem}\label{general_tail}
Given $\bdsb{g}=[g[0],\cdots,g[P-1]]^T$ and a vector $\bdsb{v}=[v[0],\cdots,v[P-1]]^T$ with i.i.d. Gaussian variates $v[p]\sim\mathcal{N}(0,\nu^2)$ for all $p\in[P]$, the following tail bound holds:
\begin{align}\label{eq:general_tail}
	&\Prob{\frac{1}{P}\left\|\bdsb{g}+\bdsb{v}\right\|^2 \geq \tau_1} \leq e^{-\frac{P}{4}\left(\sqrt{2\tau_1/\nu^2-1} - \sqrt{1+2\theta_0}\right)^2}\\
	&\Prob{\frac{1}{P}\left\|\bdsb{g}+\bdsb{v}\right\|^2 \leq \tau_2} \leq e^{-\frac{P}{4}\frac{\left(1+\theta_0 - \tau_2/\nu^2\right)^2}{1+2\theta_0}}
\end{align}
for any $\tau_1$ and $\tau_2$ that satisfy
\begin{align}\label{tau_requirement}
	\tau_1 &\geq \nu^2(1+\theta_0),\quad
	\tau_2 \leq \nu^2(1+\theta_0),
\end{align}
where $\theta_0$ is the {\it normalized non-centrality parameter} given by
\begin{align}\label{nu0_def}
	\theta_0 \defn \frac{\left\|\bdsb{g}\right\|^2}{P\nu^2}.
\end{align}	
\end{lem}

\begin{proof}
The quantity $\left\|\bdsb{g}+\bdsb{v}\right\|^2$ can be written element-wise as
\begin{align}
	\left\|\bdsb{g}+\bdsb{v}\right\|^2 
	=
	\sum_{p=0}^{P-1} \left(s[p]+v[p]\right)^2 
\end{align}
where each summand is a normal random variable with mean $u[p]$ and variance $\nu^2$. Therefore, according to the definition of non-central chi-square variables, the quantity
\begin{align}
	\frac{\left\|\bdsb{g}+\bdsb{v}\right\|^2}{\nu^2} \sim \chi_P^2
\end{align}
is a non-central $\chi^2$ random variable of $P$ degrees of freedom with a non-centrality parameter 
\begin{align}\label{nu_def}
	\theta = \sum_{p=0}^{P-1}\frac{|s[p]|^2}{\nu^2} = \frac{\left\|\bdsb{g}\right\|^2}{\nu^2}.
\end{align}
For notational convenience, we use the normalized non-centrality parameter $\theta_0$ in \eqref{nu0_def} such that $\theta = P\theta_0$. Without loss of generality, let the thresholds $\tau_1$ and $\tau_2$ take the following form with respect to $z_1$ and $z_2$:
\begin{align*}
	\tau_1 &= \frac{\nu^2}{P}\left[(P+P\theta_0)+2\sqrt{(P+2P\theta_0) z_1} + 2z_1\right]\\
	\tau_2 &= \frac{\nu^2}{P}\left[(P+P\theta_0)-2\sqrt{(P+2P\theta_0) z_2}\right],
\end{align*}
then the tail bounds in Lemma \ref{lem_general_tail_chi} can be obtained easily with respect to $z_1$ and $z_2$. Using \eqref{nu_def}, the corresponding $z_1$ and $z_2$ can be solved as
\begin{align*}
	z_1 &= \frac{P}{4}\left(\sqrt{2\tau_1/\nu^2-1} - \sqrt{1+2\theta_0}\right)^2\\
	z_2 &= \frac{P}{4}\frac{\left(1+\theta_0 - \tau_2/\nu^2\right)^2}{1+2\theta_0}
\end{align*}
as long as the thresholds $\tau_1$ and $\tau_2$ satisfy \eqref{tau_requirement}. Thus according to Lemma \ref{lem_general_tail_chi}, we have the tail bounds in \eqref{eq:general_tail}.
\end{proof}

\begin{cor}\label{general_tail_worst}
Suppose that the normalized non-centrality parameter $\theta_0$ in Lemma \ref{general_tail} is bounded between
\begin{align}
	0\leq \theta_{\min} \leq \theta_0\leq \theta_{\max},
\end{align}
then the following worst case tail bounds hold:
\begin{align*}%\label{eq:general_tail_worst}
	&\Prob{\frac{1}{P}\left\|\bdsb{g}+\bdsb{v}\right\|^2 \geq \tau_1} \leq e^{-\frac{P}{4}\left(\sqrt{2\tau_1/\nu^2-1} - \sqrt{1+2\theta_{\max}}\right)^2}\\
	&\Prob{\frac{1}{P}\left\|\bdsb{g}+\bdsb{v}\right\|^2 \leq \tau_2} \leq e^{-\frac{P}{4}\frac{\left(1+\theta_{\min} - \tau_2/\nu^2\right)^2}{1+2\theta_{\min}}}
\end{align*}
for any $\tau_1$ and $\tau_2$ that satisfy
\begin{align}\label{tau_requirement2}
	\tau_1 &\geq \nu^2(1+\theta_{\max}),\quad
	\tau_2 \leq \nu^2(1+\theta_{\min}).
\end{align}
\end{cor}
\begin{proof}
The first tail bound can be easily obtained since $\tau_1 \geq \nu^2(1+\theta_{\max})$, the exponent is monotonically decreasing with respect to $\theta_0$, and therefore substituting it with $\theta_{\max}$ leads to an upper bound. 

The second tail bound depends on the monotonicity with respect to $\theta_0$. The tail bound is monotonic with respect to the exponent, so in the following we examine the monotonicity of the exponent with respect to $\theta_0$. The exponent can be re-written as a form of the $x+1/x$ function:
\begin{align}
	\frac{\left(1+\theta_{\min} - \tau_2/\nu^2\right)^2}{1+2\theta_{\min}}
	&=\left(\theta_0 + \frac{1}{2}\right) + \frac{\left(\frac{1}{2}-\frac{\tau_2}{\nu^2}\right)^2}{\left(\theta_0 + \frac{1}{2}\right)} + 2\left(\frac{1}{2}-\frac{\tau_2}{\nu^2}\right),
\end{align}
which has a minimum at
\begin{align}
	\theta_0^\star = \left|\frac{1}{2}-\frac{\tau_2}{\nu^2}\right| - \frac{1}{2},
\end{align}
and monotonically increasing for any $\theta_0>\theta_0^\star$. Now it remains to see whether $\theta_0^\star$ is within the interval $[\theta_{\min},\theta_{\max}]$, which needs to be discussed separately depending on the choice of $\tau_2$:
\begin{enumerate}
	\item $\nu^2/2\leq \tau_2\leq \nu^2(1+\theta_{\min})$: in this case, we have
	\begin{align}
		\theta_0^\star = \frac{\tau_2}{\nu^2} - 1\leq \theta_{\min}.
	\end{align}	
	\item $0 < \tau_2 < \nu^2/2$: in this case, we have
	\begin{align}
		\theta_0^\star = - \frac{\tau_2}{\nu^2} \leq 0 \leq \theta_{\min}.
	\end{align}		
\end{enumerate}
Therefore, it has been shown that as long as $\tau_2$ satisfies \eqref{tau_requirement2}, the exponent is monotonically increasing with respect to $\theta_0\in[\theta_{\min},\theta_{\max}]$ and therefore the minimum exponent is achieved by substituting $\theta_0$ with $\theta_{\min}$.
\end{proof}

\end{document}